\documentclass[11pt]{article}

\usepackage{amsmath, amsfonts, amssymb, amsthm, graphicx, color, ulem, enumitem}
\usepackage{fullpage}
\usepackage{hyperref}

\newcommand{\HilH}{\mathcal{H}}
\newcommand{\realR}{\mathbb{R}}
\newcommand{\compC}{\mathbb{C}}

\newcommand{\Prob}{\mathbb{P}}

\newcommand{\ie}{i.e.}
\newcommand{\eg}{e.g.}
\newcommand{\cf}{cf.}

\newcommand{\iid}{i.i.d.}

\newcommand{\M}{\mathcal{M}}

\newcommand{\B}{\mathcal{B}}
\newcommand{\D}{\mathcal{D}}
\newcommand{\K}{\mathcal{K}}

\newcommand{\Q}{\vec{\Omega}}

\newcommand{\R}{\mathbb{R}}
\newcommand{\A}[1][n]{\mathbf{A}_{#1}}
\newcommand{\m}{\mathbf{m}}
\newcommand{\aaa}{a}
\newcommand{\acc}{\mathbf{a}_c}

\newcommand{\redge}{\mathbf{e}}
\newcommand{\gfn}{\mathbf{g}}
\newcommand{\Gfn}{\mathbf{G}}
\newcommand{\Hfn}{\mathbf{H}}
\newcommand{\lcons}{\mathbf{\ell}}
\newcommand{\bfGamma}{\mathbf{\Gamma}}

\newcommand{\FGUE}{F_0}
\newcommand{\FGOE}{F_1}
\newcommand{\Int}{I_n^T}
\newcommand{\Intx}{J_n^T}
\newcommand{\Intxx}{\hat{J}_n^{T}}
\newcommand{\CK}{\mathbb{K}}
\newcommand{\bfgamma}{\hat{\gamma}}
\newcommand{\crit}{\Sigma}
\newcommand{\erf}{G}

\newcommand{\factor}{\frac35}	% for factor 2/3
\newcommand{\twofactor}{\frac65}

\DeclareMathOperator{\Tr}{Tr}
\DeclareMathOperator{\Freddet}{det}
\DeclareMathOperator{\diag}{diag}

\DeclareMathOperator{\ess}{ess}

\DeclareMathOperator{\Ai}{Ai}
\DeclareMathOperator{\Airy}{Airy}

\newtheorem{lemma}{Lemma}[section]
\newtheorem{thm}{Theorem}[section]
\newtheorem{prop}{Proposition}[section]
\newtheorem{cor}{Corollary}[section]

\theoremstyle{definition}
\newtheorem{defn}{Definition}[section]

\theoremstyle{remark}
\newtheorem{rmk}{Remark}[section]

\title{On the largest eigenvalue of a Hermitian random matrix model with spiked external source I. Rank one case}
\author{Jinho Baik\thanks{Department of Mathematics, University of Michigan, Ann Arbor, MI, 48109, USA \newline
email: \texttt{baik@umich.edu}} \
and 
Dong Wang\thanks{Department of Mathematics, University of Michigan, Ann Arbor, MI, 48109, USA \newline
email: \texttt{dowang@umich.edu}}}

\date{\today}

\begin{document}

\maketitle

%%%	We may remove the table of contents later 

\begin{abstract} 

Consider a Hermitian matrix model under an external potential with spiked external source.
When the external source is of rank one, we compute the limiting distribution of the largest eigenvalue for general, regular, analytic potential for all values of the external source. There is a transitional phenomenon, which is universal for convex potentials. However, for non-convex potentials, new types of transition may occur. 
The higher rank external source is analyzed in the subsequent paper. 

%The case when the external source is of higher rank is studied in the subsequent paper. 

%This is the first part of a study on the limiting distribution of the largest eigenvalue of a Hermitian matrix model under an external potential with spiked external source. This paper concerns on the case when the external source is of rank one. The higher rank case is studied in the subsequent paper. 

%We compute the limiting distribution of the largest eigenvalue for general, regular, analytic potential for all values of the external source. For convex potentials, the transition phenomenon is universal. However, for non-convex potentials, new types of transition may occur. 

%The effect of the external source to the largest eigenvalue is significant only if it is larger than a certain value. For the Gaussian and the Whishart cases, this critical value was known. Moreover, the limiting distribution of the largest eigenvalue was evaluated at all of the sub-critical, the super-critical and the critical cases. We show that for convex potential, this transition phenomenon is unchanged. However, for non-convex potential, new types of transition may occur, and we determine the limiting distribution of the largest eigenvalue for general, regular, analytic potential for all values of the external source. 
\end{abstract}

%\tableofcontents	

\maketitle

\section{Introduction and results} \label{section:introduction}

\subsection{Introduction}
Fix an $n\times n$ Hermitian matrix $\A$ and consider the following density function on the set $\HilH_n$ of $n\times n$ Hermitian matrices: 
\begin{equation}\label{eq:pdf_of_external_source_model}
	p_n(M) = \frac{1}{Z_n} e^{-n \Tr(V(M)-\A M)}
\end{equation}
where $Z_n$ is the normalization constant. 
Here the `external potential' $V(x)$
is a real-valued function which decays fast enough as $|x|\to\infty$ so that $Z_n$ is convergent.
% (see Subsection~\ref{sec:assumV} below for the explicit conditions on $V$ considered in this paper). 
The matrix $\A$ is called the external source: see 
%The sequence of probability spaces $(\HilH_n, p_n)$ is sometimes called the Hermitian matrix model with potential $V$ and external source $\A$. See 
\eg\ \cite{Brezin-Hikami96}, \cite{Brezin-Hikami98}, \cite{Zinn_Justin97}, \cite{Zinn_Justin98}, \cite{Bleher-Kuijlaars04}, \cite{Bleher-Kuijlaars04a}, \cite{Aptekarev-Bleher-Kuijlaars05}. Note that the distribution of eigenvalues of $M$ is unchanged if $\A$ is replaced by $U\A U^{-1}$ for any unitary matrix $U$. Since we are only concerned on eigenvalues of $M$, we assume without loss of generality that $\A$ is a diagonal matrix.

A special case is when for all $n$, the external source has a \emph{fixed} number $\m$, called the \underline{rank of $\A$}, 
of fixed non-zero eigenvalues. 
In this case, 
the sequence of probability spaces $(\HilH_n, p_n)$ is called  a Hermitian matrix model with \emph{spiked} external source,  \underline{spiked source model} for short.
In this paper we only consider the case when $\m=1$. 
The higher rank case when $\m>1$ will be analyzed in the upcoming companion paper.  Throughout this paper, we assume that $n\ge 1$ and
\begin{equation} \label{eq:defination_of_ex_source_A}
	\A= \diag (a, \underbrace{0,\cdots,0}_{n-1} ),
	%\qquad \aaa_1\ge \aaa_2\ge \cdots \ge \aaa_\m>0.
\end{equation}
where $a$ is a real number, independent of $n$. 
%We study the asymptotics of the largest eigenvalue $\xi_{\max}(n)$ of $M$.

There are two important special cases. 
When $V(x)=x^2/2$, the spiked source model is called the GUE spiked model. The density $p_n(M)$ is that of  $M=H + \A$  where $H$ is an $n\times n$ GUE (Gaussian unitary ensemble) matrix. 
When  $V(x) = ((1+c)x - c\log x) \chi_{(0,\infty)}(x)$, $c  = (m-n)/n \geq 0$, the spiked source model is the complex Wishart spiked model. In this case, setting  $\Sigma := (1-(1+c)^{-1}\A)^{-1}$, the density $p_n(M)$ is that of  $M=\Sigma^{1/2}XX^{\dagger}\Sigma^{1/2}$ where $X$ is an $n \times m$ complex rectangular matrix with \iid\ standard complex Gaussian entries.\footnote{For the complex Wishart spiked model, $V$ is not real analytic at $x=0$. Throughout this paper, we only consider $V$ which is real analytic in the whole line. However, the method can be generalized to the Wishart-type potentials in a straightforward way.} 
For these two cases, the limit of the largest eigenvalue $\xi_{\max}(n)$ of $M$ was studied in great detail in 
 \cite{Baik-Ben_Arous-Peche05} and \cite{Peche06}. 
%These two special cases were studied in great detail; see for example, \cite{Baik-Ben_Arous-Peche05}, \cite{Baik06}, \cite{Peche06}. 
%For both cases, the key ingredient is a determinantal formula of the distribution function of the largest eigenvalue in terms of explicit contour integrals. 
An important feature is the following phase transition phenomenon. 
%The limiting empirical distribution of the eigenvalues of $M$ is unchanged from the null case when $a=0$. %These laws govern the global limit and are obtained via a so-called equilibrium measure for the potential $V$ (see section~\ref{sec:assumV}  below for more on the equilibrium measure). 
Let $\redge$ denote the right-end point of the limiting empirical distribution of the eigenvalues of the Hermitian matrix model with no external source (see~\eqref{eq:rightedge} below).\footnote{The limiting empirical distribution in the spiked source model is the same as the Hermitian model with no external source.} %Let  $\xi_{\max}(n)$ be the largest eigenvalue of $n$ by $n$ matrix $M$ in the spiked model. 
It was shown in both the GUE and the complex Wishart spiked models that as $n\to\infty$, with probability $1$, %\marginpar{Should we remove this part? Or should we discuss probability 1 limit in other theorems?}
\begin{equation}\label{eq:GUE1}
	\xi_{\max}(n) \to \begin{cases}
	\redge, \qquad &\text{if $\aaa \le \frac12 V'(\redge)$,} \\
	x_0(\aaa), &\text{if $\aaa> \frac12 V'(\redge)$,} 
	\end{cases}
\end{equation}
for some continuous, increasing function $x_0(\aaa)$ in $\aaa\in (\frac12 V'(\redge), \infty)$ satisfying $\lim_{a\downarrow \frac12 V'(\redge)} x_0(a)=\redge$. 
%Hence $\frac12 V'(\redge)$ is the `critical value' of the model. 
Moreover, there exists $\beta>0$ (see \eqref{eq:definition_of_beta} below) such that for each $T\in\R$, 
\begin{equation}\label{eq:GUE2}
	\Prob_n\big( (\xi_{\max}(n)- \redge)\beta n^{2/3} \le T\big)  \to \begin{cases}
	\FGUE(T), \qquad &\text{if $\aaa < \frac12 V'(\redge)$,} \\
	\FGOE(T),&\text{if $\aaa= \frac12 V'(\redge)$,} 
	\end{cases}
\end{equation}
and there exists $\gamma(a)$ such that for each $T\in\R$, 
\begin{equation}\label{eq:GUE3}
	\Prob_n\big( (\xi_{\max}(n)- x_0(\aaa)) \gamma(a) n^{1/2} \le T\big)  	\to 
	\erf(T), \quad \text{if $\aaa> \frac12 V'(\redge)$.} 
\end{equation}
Here the function $\erf(T)=\frac1{\sqrt{2\pi}} \int_{-\infty}^T e^{-\frac12 \xi^2} d\xi$ is the cumulative distribution function of the standard normal distribution, and $\FGUE$ and $\sqrt{\FGOE}$ are the  GUE and GOE Tracy-Widom distribution functions,  respectively. They are defined in~\eqref{eq:defn_of_F_TW} and~\eqref{eq:defn_of_F_1(T)} below, respectively. A limit theorem was also proven for the double scaling case when $a= \frac12 V'(\redge)+\frac{\alpha}{n^{1/3}}$. 
%See \cite{Baik-Ben_Arous-Peche05} and \cite{Peche06} for details.

The purpose of this paper is to extend the results~\eqref{eq:GUE1}--\eqref{eq:GUE3} %for the GUE and the complex Wishart spiked  model 
to the spiked source model with general potential $V$.
% (with the density function~\eqref{eq:pdf_of_external_source_model} above). 
%For a certain class of potentials, we prove the universality of the above 
%We obtain both universal and non-universal results. 
It turns out that if $V(x)$ is convex in the interval $x\in (\redge, \infty)$, then all of~\eqref{eq:GUE1}--\eqref{eq:GUE3} still hold. Especially, the  `critical value' of $\aaa$ is again given by $\frac12 V'(\redge)$. However, if $V$ is not convex in $(\redge, \infty)$, new features may occur. Two key new features are the followings.

\begin{itemize}
\item The critical value of $\aaa$ may be smaller than $\frac12 V'(\redge)$. See Lemma \ref{lem:Gprop} and Theorem~\ref{thm:thm_rank_1}. For such a case, when $\aaa$ equals this critical value, $\xi_{\max}(n)$  does not converge with probability 1. Instead it converges to two or more values, each with non-zero probability.  In this case, the fluctuation of $\xi_{\max}(n)$ is generically $\FGUE$ at the smallest limiting value and $\erf$ at the larger limiting values. See Theorem  \ref{thm:critical_traditional_split}.

\item There may be a discrete set of `secondary critical values' of $a$, which are greater than the critical value.  
If $\aaa$ is at a secondary critical value, then $\xi_{\max}(n)$ converge to 
two or more values, each with non-zero probability. In this case, the fluctuation of $\xi_{\max}(n)$ is generically $\erf$ at each of the limiting values.  See Theorem \ref{thm:supercritical_split}.

\end{itemize}

The exact assumptions on the potential $V$ is given in Subsection~\ref{sec:assumV}. The universality result for convex potentials is in Subsection~\ref{subsection:convex}. 
In Subsection~\ref{subsection:critical_values} we define the critical and the secondary critical values for non-convex potentials. The limit laws for the non-convex potentials are given in Subsection~\ref{subsection:main_results_rank_1}. 

\bigskip

While we were preparing for this paper and the companion paper for the higher rank case, we learned that M. Bertola, R. Buckingham, S. Y. Lee and V. Pierce were also working on the spiked source models (see \cite{Bertola-Buckingham-Lee-Pierce11} for the first part of their work). While we focus, especially in the second paper, on the limit laws when $\aaa_1, \cdots, \aaa_\m$ are distinct, Bertola, Buckingham, Lee and Pierce focus on the case when $\aaa_1=\cdots = \aaa_\m$ and $\m\to \infty$ slower than $n$. Also we use the asymptotics of usual orthogonal polynomials but Bertola, Buckingham, Lee and Pierce use asymptotics of multiple orthogonal polynomials via Riemann-Hilbert problem of size larger than $2$. 

%Their work and ours share the same starting point: multiple-orthogonal polynomials. However, the methods of computing the multiple-orthogonal polynomials asymptotically are different. In this paper, we express the multiple-orthogonal polynomials in terms of an integral involving the usual orthogonal polynomials. Since the asymptotics of orthogonal polynomials are well studied, the task becomes the evaluation of  the integral using the method of steepest-descent. The higher rank case in the companion paper is obtained then by realizing an determinant identity that expresses the higher rank case in terms of the rank one case. On the other hand,  Bertola, \etal\ analyzed the multiple-orthogonal polynomials using a Riemann-Hilbert problem. The technical part is the analysis of the Riemann-Hilbert problem which is of size higher than $2$.  

\bigskip
Before closing this subsection, we mention that the spiked real symmetric matrix model is much more difficult. Even for the GOE and the real Wishart case, the limiting distribution at the critical value is not yet known. For the quaternionic case, the limiting distribution is obtained when the rank $\m=1$  (see \cite{Wang08} for the Wishart model; Gaussian model is also similar).

We also mention that there are several results for the spiked Wigner ensembles and spiked sample covariance matrices.
% and spiked Wishart ensembles. %Especially the second model is of great interest for statistical applications. 
See, for example, \cite{Baik-Silverstein06}, \cite{Paul08},  \cite{Feral-Peche07}, \cite{Capitaine-Donati_Martin-Feral09},  \cite{Nadakuditi-Silverstein10}, \cite{Benaych_Georges-Nadakuditi11} and \cite{Benaych_Georges-Guionnet-Maida11}. %The other is the real and quaternionic 

\subsection{Assumptions on external potential $V$.}\label{sec:assumV}

Throughout this paper, we assume the following three conditions on $V$: 
\begin{gather}
 V(x) \textnormal{ is real analytic in } \R, \label{eq:condition_of_V_1} \\
 \frac{V(x)}{\sqrt{x^2+1}} \to +\infty \textnormal{ as } |x|\to\infty, \label{eq:condition_of_V_2} \\
V \textnormal{ is `regular'.} \label{eq:condition_of_V_3}
\end{gather}
The second condition is to ensure the convergence of the density function: compare this with the condition on $V$ in \cite{Deift-Kriecherbauer-McLaughlin-Venakides-Zhou99}.  
The third condition on being `regular' is a technical condition as defined in\ \cite{Deift-Kriecherbauer-McLaughlin-Venakides-Zhou99}. We need a few definitions to state it. 

%To reduce the technical difficulty in proofs, we assume that the weight function $V$ is regular, although our results should hold for a large class of irregular weight functions also.

%To define the term ``regular'' of a weight function, we recall terms arising from the Riemann-Hilbert Problem (RHP) of the Hermitian matrix model without external source, \ie, the \pdf\ of the matrix $M$ is given by \eqref{eq:pdf_of_external_source_model} with $A = 0$. The standard reference is \cite{Deift-Kriecherbauer-McLaughlin-Venakides-Zhou99}, in which we not only borrow notations, but also use results.

First, recall the equilibrium measure and the so-called $\gfn$-function. General references are \cite{Saff-Totik97} and \cite{Deift-Kriecherbauer-McLaughlin98}. 
For a given potential $V$, the empirical distribution of the eigenvalues of the matrix model with no external source converges to the associated equilibrium measure $\mu$. The equilibrium measure is characterized by a certain variational problem. If $V$ is real analytic, $\mu$ is supported on a finite union of intervals, 
\begin{equation}\label{eq:Jend}
	J = \bigcup^N_{j=0} (b_j, a_{j+1}), \quad \textnormal{with} \quad b_0 < a_1 < b_1 < \dots < a_{N+1},
\end{equation}
for some $N\ge 0$. 
%Note that $N+1$ is the number of the intervals of the support.  
We denote the right-most edge %and the left-most edge 
of the support by 
\begin{equation}\label{eq:rightedge}
	\redge :=a_{N+1} . %\quad \textnormal{and} \quad \redge^{(l)} := b_0.
\end{equation}
On $J$, $d\mu$ has the form $d\mu = \Psi(x)dx$,  
\begin{equation}\label{eq:eqmeasure}
	\Psi(x) = \frac{1}{2\pi i} R^{1/2}_+(x)h(x), \quad \textnormal{where} \quad R(z) = \prod^N_{j=0} (z-b_j)(z-a_{j+1}).
\end{equation}
The function $R(z)^{1/2}$ is defined to be analytic in $\mathbb{C}\setminus J$ and satisfy $R(z)^{1/2} \sim z^N$ as $z\to\infty$. The notation $R^{1/2}_+(x)$ for $x\in J$ denotes the limit of $R^{1/2}(z)$, $z\in \mathbb{C}_+$, as $z\to x$ from above. The function $h(x)$ is real analytic in $\R$ and is given by \cite[Formula (3.18)]{Deift-Kriecherbauer-McLaughlin-Venakides-Zhou99}.

The equilibrium measure, $d\mu(x)=\Psi(x)dx$, is characterized by the following conditions: there is a constant (called the Robin constant) $\lcons$ such that 
\begin{align}
	2\int_J \log \lvert x-s \rvert \Psi(s)ds -V(x) = \lcons & \quad \textnormal{for $x \in \bar{J}$}, \label{eq:first_eq_characterize_dmu} \\
2\int_J \log \lvert x-s \rvert \Psi(s)ds -V(x) \leq \lcons & \quad \textnormal{for $x \in \realR \setminus \bar{J}$.} \label{eq:second_eq_characterize_dmu}
\end{align}
The so-called $\gfn$-function is defined by 
\begin{equation} \label{eq:definition_of_g}
	\gfn(z) := \int_J \log(z-s)\Psi(s)ds, \quad \textnormal{for $z \in \compC \setminus (-\infty, \redge)$}.
\end{equation}

The potential $V$ is said to be regular (see \cite{Deift-Kriecherbauer-McLaughlin-Venakides-Zhou99}) if
\begin{itemize}
\item $h(x)\neq 0$ for $x\in \bar{J}$,  
\item the inequality in~\eqref{eq:second_eq_characterize_dmu} is strict.
\end{itemize}
The first condition implies that the function $\Psi(x)>0$ for all $x\in J$, and also that $\Psi(x)$ vanishes like a square-root at each end of the interval of the support. This in turn implies, in particular, that for the model with $\A=0$, the largest eigenvalue has the limiting distribution given by $\FGUE$ (see \eg\ \cite{Deift-Gioev07a} for the non-varying weight; varying weight case is similar using the analysis of \cite{Deift-Kriecherbauer-McLaughlin-Venakides-Zhou99}.) Note that the second condition restricted to the domain $x>\redge$ implies that 
\begin{equation} \label{eq:consequence_or_regularity}
	2\gfn(x)-V(x)<\lcons, \qquad x>\redge.
\end{equation}
We will use this fact later.

\subsection{Statement of results: convex potentials} \label{subsection:convex}

%\subsubsection{Convergence in probability one}

%The above two theorems follow from the following results and discusstions for the limiting distribution. 

Let $\FGUE(T)$ be the GUE Tracy-Widom distribution defined by 
\begin{equation} \label{eq:defn_of_F_TW}
	\FGUE(T) := \det(1 - \chi_{[T,\infty)}K_{\Airy}\chi_{[T,\infty)}),
\end{equation}
where  $\chi_{[T,\infty)}$ denotes the projection operator on $[T,\infty)$, and $K_{\Airy}$ is the Airy operator  defined by the kernel 
\begin{equation}
	K_{\Airy}(x,y) = \frac{\Ai(x)\Ai'(y) - \Ai'(x)\Ai(y)}{x-y}.
\end{equation}
Here $\Ai$ is the Airy function. % (See \cite[Section 10.4]{Abramowitz-Stegun64}).

%The rank $1$ generalized GUE Tracy-Widom distribution is also defind by a Fredholm determinant. 
For $\alpha \in \realR$, define the function 
\begin{equation}\label{eq:Calphadef}
	C_{\alpha}(\xi) := \frac1{2\pi} \int e^{i\frac13 z^3+i\xi z} \frac{dz}{\alpha+iz},
\end{equation}
where the contour is from $\infty e^{5\pi i/6}$ to $\infty e^{\pi i/6}$ and the pole $z=-i\alpha$ lies above the contour in the complex plane: see Figure \ref{figure:contour_above_-ia}. 

\begin{figure}
\centering
\includegraphics{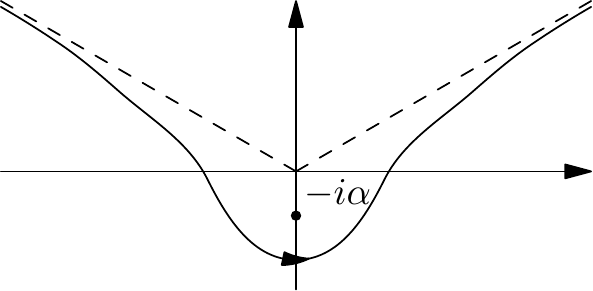}
\caption{The contour from $\infty e^{5\pi i/6}$ to $\infty e^{\pi i/6}$.}
\label{figure:contour_above_-ia}
\end{figure}
Define
\begin{equation} \label{eq:defn:of_F_1(T;alpha)}
	\FGOE(T;\alpha) := \FGUE(T)\cdot 
	\bigg(1 - \langle (1-\chi_{[T,\infty)} K_{\Airy} \chi_{[T,\infty)})^{-1} C_{\alpha},  \Ai  \rangle_{[T,\infty)}\bigg),
\end{equation}
where $\langle f,g \rangle_E$ denotes the real inner product over $E$, $\int_E f(x)g(x)dx$. (See \cite[Definition 1.3]{Baik-Ben_Arous-Peche05}.) When $\alpha = 0$, 
\begin{equation} \label{eq:defn_of_F_1(T)}
	\FGOE(T) := \FGOE(T;0)
\end{equation}
equals the square of the GOE Tracy-Widom distribution (see \cite[Formula (24)]{Baik-Ben_Arous-Peche05}). 
%Fix an external potential $V$ satisfying the assumptions~\eqref{eq:condition_of_V_1}-\eqref{eq:condition_of_V_3} and consider the corresponding spiked model defined by the density function~\eqref{eq:pdf_of_external_source_model}. 

\medskip
Fix a potential $V$ satisfying the assumptions~\eqref{eq:condition_of_V_1}--\eqref{eq:condition_of_V_3}. 
In the companion paper on the higher rank case, we need to consider the spiked source model of rank $1$ whose density function is same as in~\eqref{eq:pdf_of_external_source_model} but with the change that the matrix $M$ is now of size $n-j+1$ and $\A$ is replaced by $\mathbf{A}_{n-j+1}$ for fixed $j$:
\begin{equation} \label{eq:generalized_pdf}
	\frac{1}{Z_{n-j+1,n}} e^{-n\Tr(V(M)-\mathbf{A}_{n-j+1}M)}, \qquad M\in\mathcal{H}_{n-j+1}.
\end{equation}
Note that the factor $n$ in front of the potential is unchanged. 
For a subset $E\subset \R$, let $\Prob_{n-j+1,n}(a; E)$ denote the gap probability that there are no eigenvalues of $M$ in the set $E$, where $a$ represents the unique non-zero eigenvalue of $\mathbf{A}_{n-j+1}$. Hence $\Prob_{n-j+1,n}(\aaa; [t, \infty))$ is the probability that the largest eigenvalue of $M$ is less than $t$. 
%We denote by $\Prob_n(E)$ the gap probability for the  Hermitian matrix model without external source.

Let $\redge$ be as in~\eqref{eq:rightedge}.
Set (recall~\eqref{eq:eqmeasure})
\begin{equation} \label{eq:definition_of_beta}
	\beta := \bigg( \lim_{x\uparrow \redge} \frac{\pi \Psi(x)}{\sqrt{\redge-x}} \bigg)^{2/3} = \left( \frac{h(\redge)}{2}\right)^{2/3}\left( \frac{R(z)}{z-\redge} \right)^{1/3}\bigg|_{z=\redge}% \prod^N_{j=0}(\redge-b_j)^{1/3} \prod^N_{k=1}(\redge-a_k)^{1/3}.
\end{equation}
so that $\Psi(x)\sim \frac{\beta^{3/2}}{\pi}\sqrt{\redge-x}$ for $x\uparrow \redge$.
For $T\in\R$, define the intervals
\begin{equation}\label{eq:interval}
	\Int:=\left[ \redge + \frac{T}{\beta n^{2/3}}, \infty \right)
%	\Int:=\left( -\infty, \redge + \frac{T}{\beta n^{2/3}} \right)
\end{equation}
and
\begin{equation}\label{eq:interval2}
	\Intx(x_*):=\left[ x_* + \frac{T}{\sqrt{(V''(x_*)-\gfn''(x_*))n}}, \infty \right)
%	\Intx(x_0):=\left( -\infty, x_0 + \frac{T}{\sqrt{(V''(x_0)-\gfn''(x_0))n}} \right)
\end{equation}
for $x_*>\redge$, 
if $V''(x_*)-\gfn''(x_*)> 0$. Note that if $V(x)$ is convex in $x\ge \redge$, then $V''(x)-\gfn''(x)>0$ for all $x>\redge$. For later reference we note that $V''(x_*)-\gfn''(x_*)=-\Gfn''(x_*)$ in terms of the notation~\eqref{eq:definition_of_GH} that is defined below. 

The following is the first main result of this paper. Let $V(x)$ be a potential that is convex in $x\in ( \redge, \infty)$. 
For $a>\frac12 V'(\redge)$, let  $x_0(a)$ be the unique maximizer of the function $\gfn(x)-V(x)+\aaa x$ in $x\in (\redge, \infty)$. Such a maximizer exists since $\gfn(x)-V(x)$ in $x\in (\redge, \infty)$ 
is strictly concave and $\gfn'(\redge)-V'(\redge)+\aaa=-\frac12 V'(\redge)+a>0$ (see~\eqref{eq:properties_of_gfn})  and $\gfn'(x)-V'(x)+\aaa<0$ for all large enough $x$. This $x_0(\aaa)$ is same as in Lemma~\ref{lem:x0}.

\begin{thm}[convex potential]
\label{thm:convex}
Let  $V(x)$ be a potential that is convex in $x\in ( \redge, \infty)$. 
Set 
\begin{equation}
	\acc:= \frac12 V'(\redge).
\end{equation}
The following holds for each $T\in\R$ as $n\to\infty$ and $j=O(1)$.
\begin{enumerate}[label=(\alph*)]
\item \label{enu:thm:convex:a} For $\aaa<\acc$,  
\begin{equation} %\label{eq:main2TW}
	\lim_{n \rightarrow \infty} \Prob_{n-j+1,n} \left( \aaa; \Int \right) = \FGUE(T).
\end{equation}
\item \label{enu:thm:convex:b}
For
\begin{equation}
	\aaa=\acc+\frac{\beta\alpha}{ n^{1/3}},
\end{equation}
where $\alpha$ is in a compact subset of $\R$, we have 
\begin{equation}
	\lim_{n \rightarrow \infty} \Prob_{n-j+1,n} \left( \aaa; \Int\right) = \FGOE(T; -\alpha).
\end{equation}
\item \label{enu:thm:convex:c} For $\aaa>\acc$, 
\begin{equation} %\label{eq:main2erf}
	\lim_{n \rightarrow \infty} \Prob_{n-j+1,n} \left( \aaa; \Intx(x_0(a)) \right) = \erf(T).
\end{equation}
\end{enumerate}
\end{thm}

Hence the transition phenomenon is universal for convex potentials. The next two subsections are about non-convex potentials 

\subsection{Critical value and secondary critical values} \label{subsection:critical_values}

%Fix a function $V$ satisfying conditions \eqref{eq:condition_of_V_1}--\eqref{eq:condition_of_V_3}.
%We describe the critical values for the spiked model with $V$. 

In this Subsection, we define critical values and the secondary critical values of $\aaa$. 
%We assume in this section that $\aaa>0$. The case when $\aaa<0$ is easy and see 

By definition~\eqref{eq:definition_of_g}, $\gfn(x)$  is real analytic in $(\redge, \infty)$, is continuously differentiable in $[\redge, \infty)$ and satisfies 
\begin{equation} \label{eq:properties_of_gfn}
\begin{gathered}
\gfn'(x)>0, \quad \gfn''(x)<0, \qquad \text{for } x\in (\redge, \infty),\\ 
\gfn(\redge) = \frac{V(\redge)+\ell}{2}, \quad  \gfn'(\redge) = \frac{V'(\redge)}{2}, \quad \lim_{x \rightarrow \infty} \gfn'(x) = 0.
%\lim_{x \downarrow \redge} \gfn(x) = \frac{V(\redge)+l}{2}, \quad \lim_{x \downarrow \redge} \gfn'(x) = \frac{V'(\redge)}{2}, \quad \lim_{x \rightarrow \infty} \gfn'(x) = 0.
\end{gathered}
\end{equation}

\begin{defn}\label{def:ca}
For $a\in (0, \frac12 V'(\redge))$, define $c=c(a)$ as the unique point in $(\redge, \infty)$ satisfying 
\begin{equation} \label{eq:definition_of_c}
\gfn'(c(a)) = a.  
\end{equation}
For $a\ge \frac12V'(\redge)$, define $c(a):= \redge$.  
\end{defn}

Note that $c(a)$ decreases strictly in $a\in (0, V'(\redge)/2)$ and continuous in $a\in (0,\infty)$.

Define two auxiliary functions 
\begin{equation} \label{eq:definition_of_GH}
\begin{split}
	\Gfn(z) &=\Gfn(z;a):= \gfn(z)- V(z)+az, \\ 
	\Hfn(z) &=\Hfn(z;a):= -\gfn(z) + az +\ell
\end{split}
\end{equation}
for $z \in \compC \setminus (-\infty, \redge)$. 
Observe the following Lemma. The proof follows straightforwardly from the definition of $\gfn$, the variational condition~\eqref{eq:first_eq_characterize_dmu}, the assumption~\eqref{eq:condition_of_V_2} on $V$ and~\eqref{eq:consequence_or_regularity}. We omit the details.

\begin{lemma} \label{fact:first}
Let $a > 0$. We have the following properties. %The functions $\Gfn(x;a)$ and $\Hfn(x;a)$ satisfy the following properties for real $x\in [\redge,\infty)$. 
\begin{enumerate}[label=(\alph*)]
\item \label{enu:fact:first:a} $\Hfn(x)$ is a convex function in $x\in [\redge, \infty)$ with the unique minimum attained at $x=c(a)$.
\item \label{enu:fact:first:b} $\Hfn(x)> \Gfn(x)$ for all $x\in (\redge, \infty)$.
\item \label{enu:fact:first:c} $\Hfn(\redge)= \Gfn(\redge) = -\frac12 V(\redge)+a\redge+ \frac12 \ell$.
\item \label{enu:fact:first:d} $\displaystyle\lim_{x\downarrow \redge} \Hfn'(x)=\lim_{x\downarrow \redge} \Gfn'(x)= a-\frac12 V'(\redge)$.
\item \label{enu:fact:first:e} As $x\to +\infty$, $\Hfn(x)\to +\infty$, $\Hfn(x)/x \to a$, $\Gfn(x)\to -\infty$ and $\Gfn(x)/x \to -\infty$.
\end{enumerate}
\end{lemma}

\begin{figure}[htb]
\begin{minipage}[t]{0.45\textwidth}
\begin{center}
\includegraphics[width=\textwidth]{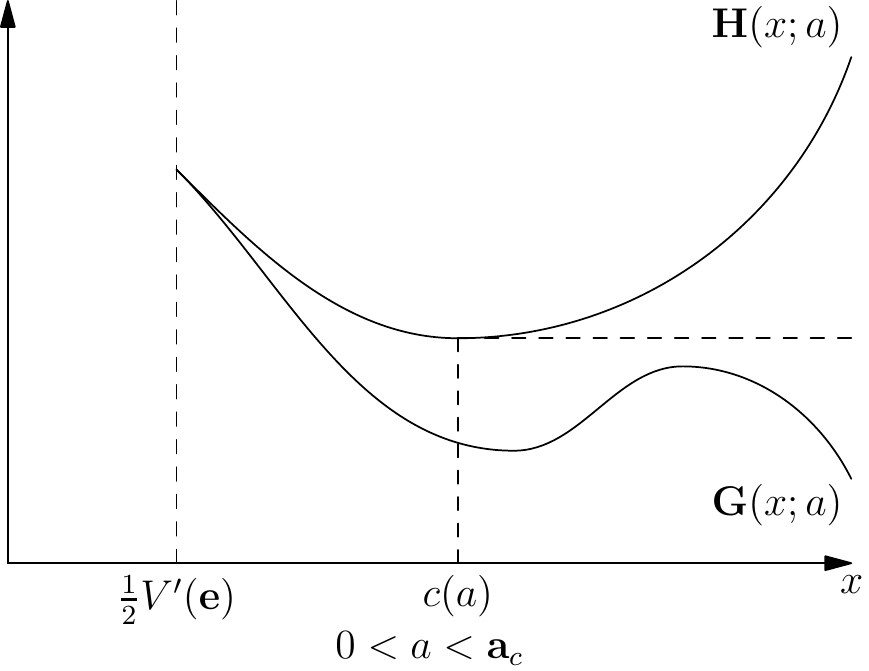}
\end{center}
\end{minipage}
\begin{minipage}[t]{0.1\textwidth}
\
\end{minipage}
\begin{minipage}[t]{0.45\textwidth}
\begin{center}
\includegraphics[width=\textwidth]{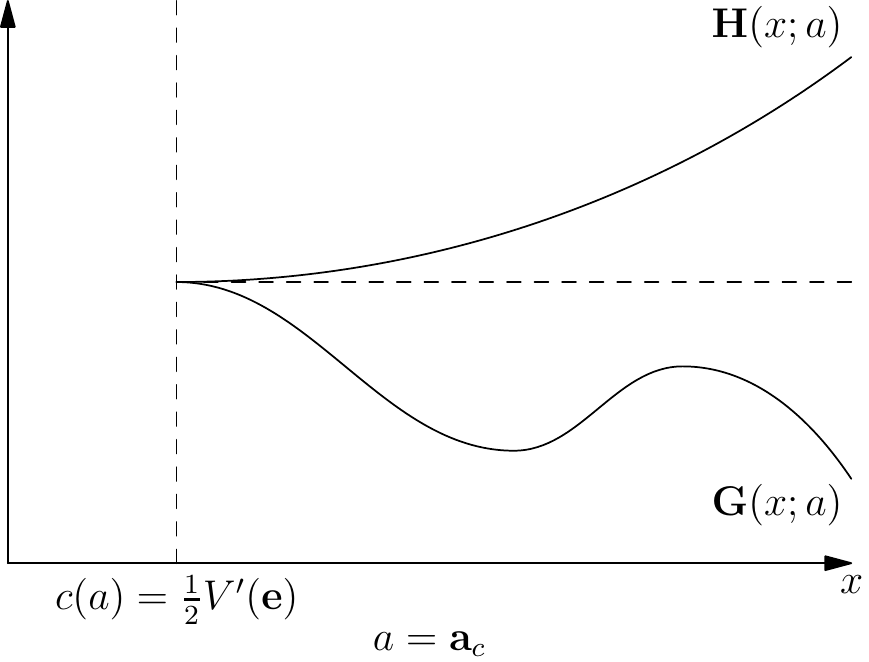}
\end{center}
\end{minipage}

\bigskip

\begin{minipage}[t]{0.45\textwidth}
\begin{center}
\includegraphics[width=\textwidth]{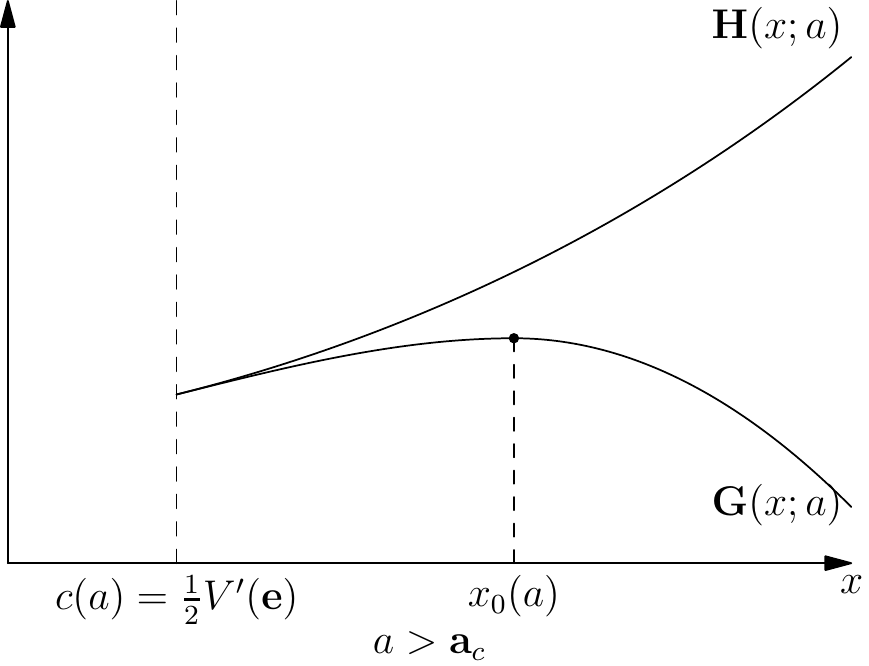}
\end{center}
\end{minipage}
\begin{minipage}[t]{0.1\textwidth}
\
\end{minipage}
\begin{minipage}[t]{0.45\textwidth}
\begin{center}
\
\end{center}
\end{minipage}

\caption{Schematic graphs of the functions $\Hfn(x;a)$ and $\Gfn(x;a)$ for a potential $V$ such that  $\acc = \frac{1}{2}V'(\redge)$, assuming that $a \not\in \mathcal{J}_V$. 
%Graphs here and in Figures \ref{figure:newGH_2} and \ref{figure:newGH_3} are schematic but not from numerical data.
} \label{figure:newGH_1}
\end{figure}

\begin{figure}[htb]
\begin{minipage}[t]{0.45\textwidth}
\begin{center}
\includegraphics[width=\textwidth]{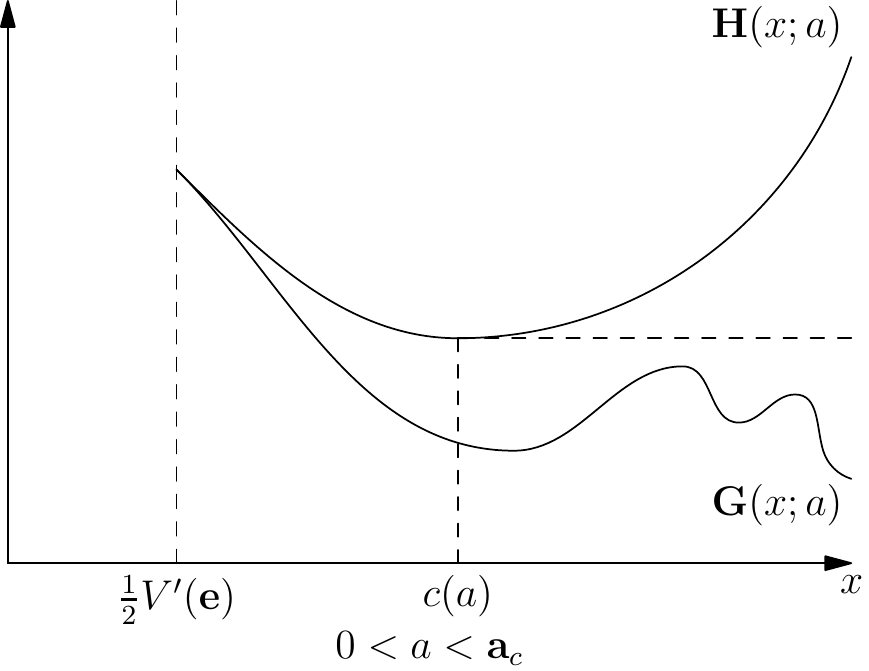}
\end{center}
\end{minipage}
\begin{minipage}[t]{0.1\textwidth}
\
\end{minipage}
\begin{minipage}[t]{0.45\textwidth}
\begin{center}
\includegraphics[width=\textwidth]{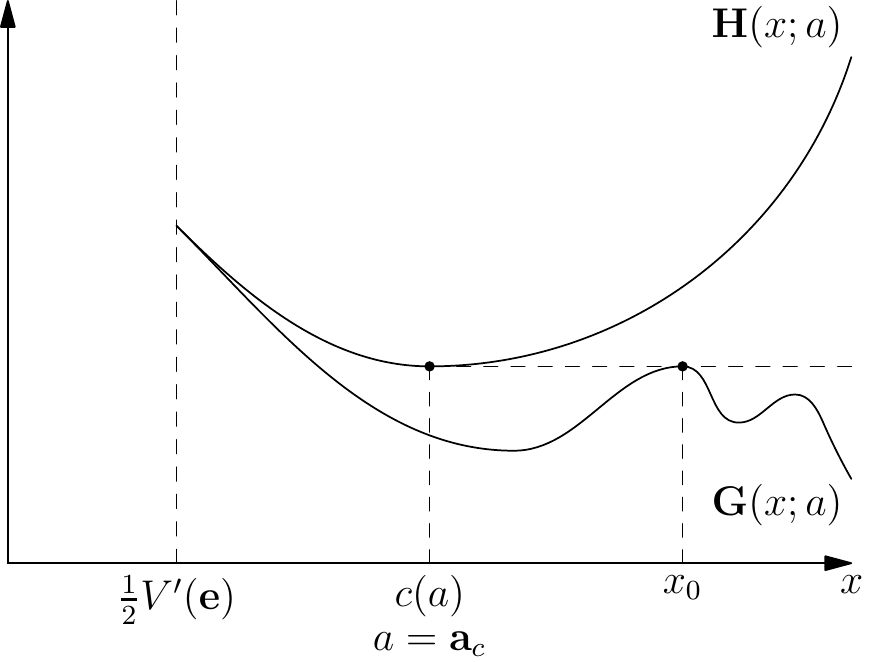}
\end{center}
\end{minipage}

\bigskip

\begin{minipage}[t]{0.45\textwidth}
\begin{center}
\includegraphics[width=\textwidth]{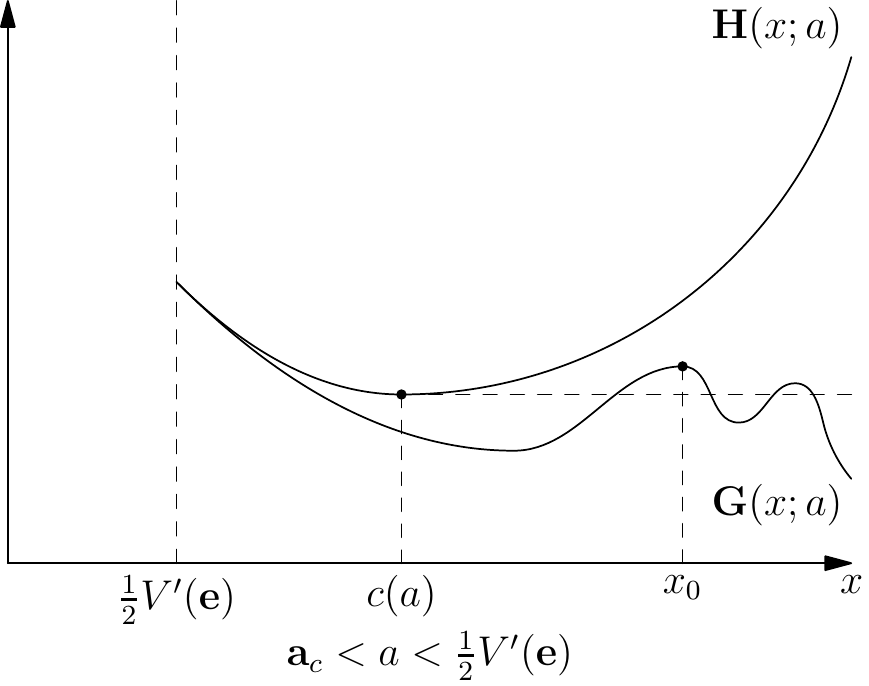}
\end{center}
\end{minipage}
\begin{minipage}[t]{0.1\textwidth}
\
\end{minipage}
\begin{minipage}[t]{0.45\textwidth}
\begin{center}
\includegraphics[width=\textwidth]{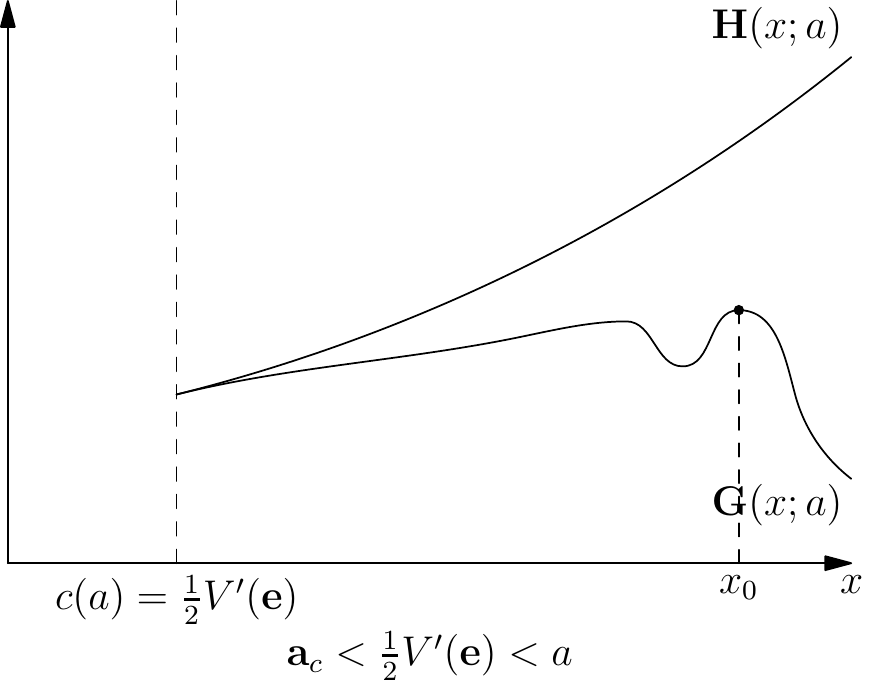}
\end{center}
\end{minipage}
\caption{Schematic graphs of the functions $\Hfn(x;a)$ and $\Gfn(x;a)$ for a potential $V$ such that  $\acc < \frac{1}{2}V'(\redge)$, assuming that $a \not\in \mathcal{J}_V$.} \label{figure:newGH_2}
\end{figure}

\begin{figure}[htb]
\begin{minipage}[t]{0.45\textwidth}
\begin{center}
\includegraphics[width=\textwidth]{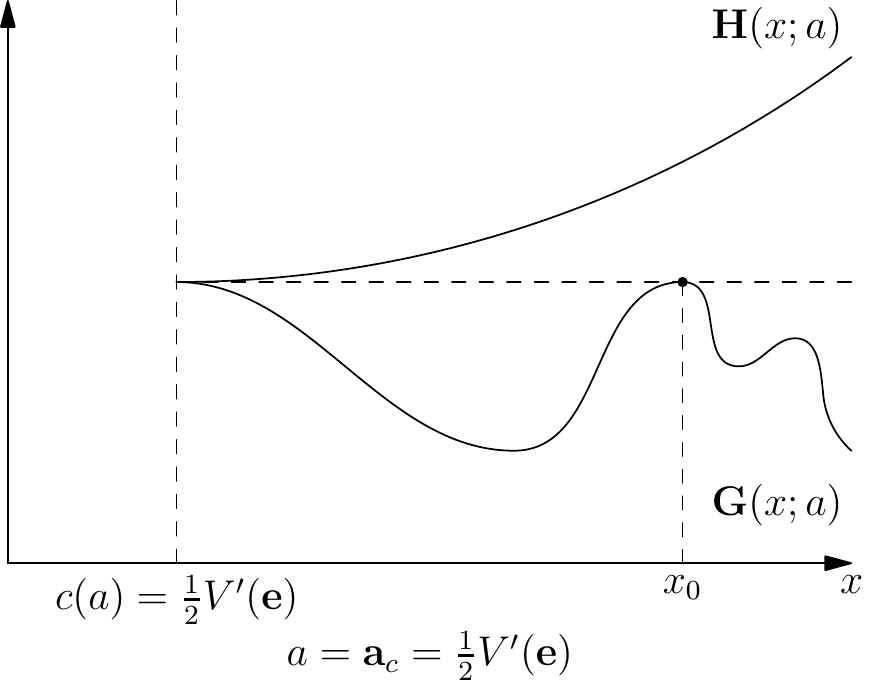}
\end{center}
\end{minipage}
\begin{minipage}[t]{0.1\textwidth}
\
\end{minipage}
\begin{minipage}[t]{0.45\textwidth}
\begin{center}
\includegraphics[width=\textwidth]{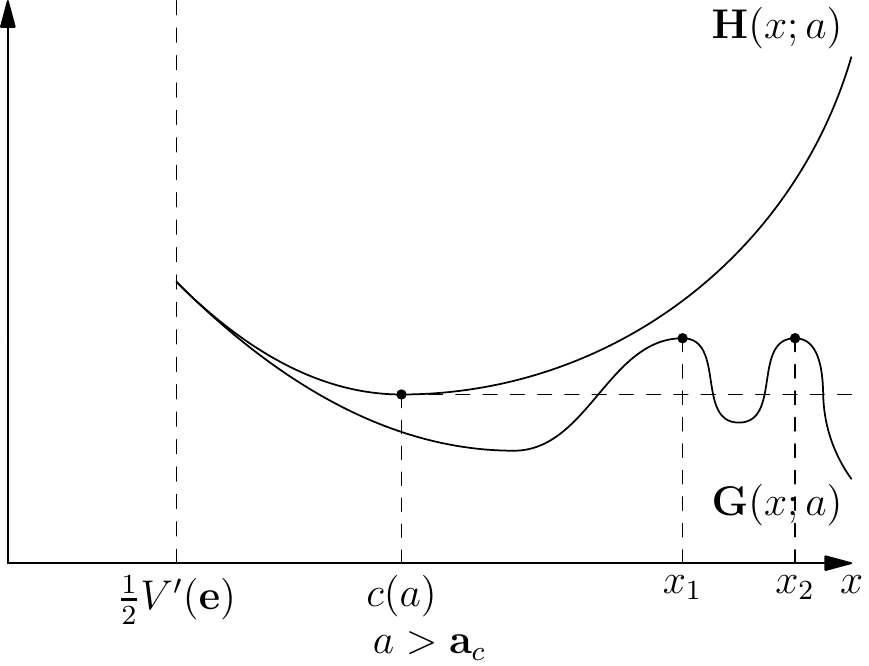}
\end{center}
\end{minipage}
\caption{Schematic graphs of functions $\Hfn(x;a)$ and $\Gfn(x;a)$ 
when $a=\acc\in \mathcal{J}_V$ and $\acc>a\in \mathcal{J}_V$.} \label{figure:newGH_3}
\end{figure}
See Figures \ref{figure:newGH_1}, \ref{figure:newGH_2} and \ref{figure:newGH_3} for a few examples of the graphs of $\Gfn$ and $\Hfn$.

\bigskip

Define the set
\begin{equation} \label{eq:definition_of_calA}
%\begin{split}
	\mathcal{A}_V:= \{ a \in (0,\infty) | \text{ there exists $\bar{x}\in (c(a), \infty)$} %\\
	\text{such that $\Gfn(\bar{x};a)> \Hfn(c(a);a)$}\}.
%\end{split}
\end{equation}
 
\begin{defn} \label{defn:definition_of_acc}
The \underline{critical value} for the spiked source model with potential $V$ is defined as 
\begin{equation} \label{eq:definition_of_acc}
	\acc:=\inf\mathcal{A}_V. %\min\bigg\{\frac12 V'(\redge), \inf\mathcal{A} \bigg\}.
\end{equation}
%as the \underline{critical value} for the spiked model with potential $V$.
\end{defn}

\begin{lemma}\label{lem:Gprop}
We have the following properties. 
\begin{enumerate}[label=(\alph*)]
\item \label{lemma_enu:a} $(\frac12V'(\redge), \infty)\subset \mathcal{A}_V$. Hence $\acc\le \frac12 V'(\redge)$.
\item \label{lemma_enu:b} The set $\mathcal{A}_V$ is an open, semi-infinite interval. Hence $\mathcal{A}_V = (\acc, \infty)$. 
%We define the \underline{critical point}
%\begin{equation} \label{eq:definition_of_acc}
%	\acc=\acc(V):= \inf\mathcal{A} %\min\bigg\{\frac12 V'(\redge), \inf\mathcal{A} \bigg\}.
%\end{equation}
%so that $\mathcal{A} = (\acc, \infty)$.
\item \label{lemma_enu:c} $\acc>0$.
\item \label{lemma_enu:d} For $0 < a<\acc$, we have $\Gfn(x;a)<\Hfn(c(a);a)$ for all $x\in (c(a),\infty)$.
\item \label{lemma_enu:e} If the potential $V(x)$ is convex for $x\ge \redge$, then $\acc= \frac12 V'(\redge)$, and $\Gfn(x; \acc)<\Hfn(\redge; \acc)$ for all $x>\redge$. (Note that $c(\acc)=\redge$ and $\Gfn(\redge;\acc)=\Hfn(\redge, \acc)$.)
\item \label{lemma_enu:f} If the potential $V$ is such that $\acc<\frac12 V'(\redge)$, then $\Gfn(x; \acc)\le \Hfn(c(\acc); \acc)$ for all $x\in (c(a),\infty)$, and the equality is attained at least at  one point. 
\end{enumerate}
\end{lemma}

\begin{proof}
\begin{itemize}
\item[\ref{lemma_enu:a}] Let $a\in (\frac12 V'(\redge), \infty)$. Since $\lim_{x\downarrow \redge}\Gfn'(x)=a-\frac12V'(\redge)>0$, there is $\bar{x}>\redge$ such that $\Gfn(\bar{x}) > \Gfn(\redge)= \Hfn(\redge)$. Thus $a\in \mathcal{A}_V$.

\item[\ref{lemma_enu:b}] The continuity of $\Gfn$ and $\Hfn$ in $a$ implies that $\mathcal{A}_V$ is an open set. 
Now we show that $\mathcal{A}$ is a semi-infinite interval. Suppose that $a\in \mathcal{A}_V$ and $a< \frac12 V'(\redge)$. Let $\bar{x}\in (c(a), \infty)$ be the point such that $\Gfn(\bar{x};a)>\Hfn(c(a);a)$. Let $a'\in (a, \frac12 V'(\redge)]$. From Definition~\ref{def:ca} of $c(a)$, we see that $c(a')<c(a)$, and hence $\bar{x}\in (c(a'), \infty)$.
Moreover, 
\begin{equation}\label{eq:GandHineq}
	\Gfn(\bar{x}; a')-\Hfn(c(a');a') = %\\ 
[\Gfn(\bar{x};a)-\Hfn(c(a);a)] %\\ 
	%&
 + [ \Hfn(c(a);a') - \Hfn(c(a');a')]+ [(a'-a)(\bar{x}-c(a))]
\end{equation} 
is strictly positive since each term in bracket is strictly positive. Thus $a'\in \mathcal{A}_V$, and this, together with (a),  implies that $\mathcal{A}_V$ is a semi-infinite interval.

\item[\ref{lemma_enu:c}] We have $\Gfn(x)-\Hfn(c(a)) 
= -V(x)+a(x-c(a))+ \gfn(x)+\gfn(c(a))-\ell\le -V(x)+ax + 2\gfn(x)-\ell$. This tends to $-\infty$ as $x\to +\infty$ due to the growth condition~\eqref{eq:condition_of_V_2} on $V$. 
Also $c(a)\to +\infty$ as $a\downarrow 0$.  Therefore when $a$ is close to $0$, $\Gfn(x)-\Hfn(c(a)) <0$ for $x>c(a)$. Hence, $a \not\in \mathcal{A}_V$ if $a$ is small enough.

\item[\ref{lemma_enu:d}] Let $0 < a<\acc$. Suppose that there is $\bar{x}\in (c(a), \infty)$ such that  $\Gfn(\bar{x};a) = \Hfn(c(a);a)$. For any $a'\in (a, \acc)$, we have $c(a)>c(a')$ since $a<a'<\frac12V'(\redge)$. Thus we find from~\eqref{eq:GandHineq} that $\Gfn(\bar{x}; a')-\Hfn(c(a');a')>0$. This implies that $a'\in\mathcal{A}_V$ which is a contradiction. 

\item[\ref{lemma_enu:e}] Let $0< a <\frac12V'(\redge)$. We will show that $a\notin\mathcal{A}$. Since $V$ is convex, $\Gfn(x)$ is concave in $x\in (\redge, \infty)$. As $\Gfn'(\redge;a)<0$, this implies that $\Gfn(x)$ is decreasing in $x\in (\redge, \infty)$. Thus for $x\in (c(a), \infty)$, $\Gfn(x)< \Gfn(c(a))<\Hfn(c(a))$. Hence $a\notin\mathcal{A}_V$. 
When $a=\frac12V'(\redge)$, a similar argument implies that $\Gfn(x)<\Hfn(x)$ for all $x>\redge$.

\item[\ref{lemma_enu:f}] This follows from the continuity of $\Gfn$ and $\Hfn$ in $a$ and the fact that $\acc=\inf\mathcal{A}_V$. 
\end{itemize}
\end{proof}

See typical graphs of $\Gfn$ and $\Hfn$ for $\acc=\frac{1}{2}V'(\redge)$ in Figure \ref{figure:newGH_1}, and typical graphs of $\Gfn$ and $\Hfn$ for $0 < \acc < \frac{1}{2}V'(\redge)$ in Figure \ref{figure:newGH_2}.
\begin{rmk} \label{rmk:nonconvexity_and_sec_crit}
When $V$ is non-convex, there  may exist $\bar{x}>\redge$ such that $\Gfn(\bar{x}; \acc)= \Hfn(\redge, \acc)$ even if $\acc=\frac12 V'(\redge)$.
\end{rmk}

By Definition~\ref{defn:definition_of_acc} of $\acc$, when $a>\acc$, $\Gfn(x;a)>\Hfn(c(a);a)$ for some $x>c(a)$. The point $x$ at which $\Gfn(x;a)$ attains its maximum plays an important role. Indeed, we will show in the below  that if the maximum is attained at a unique point, then $\xi_{\max}(n)$ converges to this point (see Theorem~\ref{thm:thm_rank_1}). However, it may happen that for some $a$'s, the function $\Gfn(x;a)$ attain its maximum at more than one point. 
%Note that these $a$'s are `exceptional' in the sense that if we perturb $V$, then such point and we call them `secondary critical points'. 
Let 
\begin{equation}
	\Gfn_{\max}(a) := \max_{x\in [c(a), \infty)} \Gfn(x;a),
\end{equation}
and define
\begin{align} 
%\Gfn_{\max}(a) := & \textnormal{the maximal value of $\Gfn(x;a)$ for $x \in [c(a), \infty)$,} \label{eq:estimation_of_linear_integral_case_1_Laplace} \\
\mathcal{J}_V:= & \{  a\in [\acc, \infty) | \text{ $\Gfn_{\max}(a)$ is attained at more than one point}\}. \label{eq:definition_of_calJ}
\end{align}
This set is discrete since $\Gfn(x;a)$ is analytic in both $x$ and $a$. 
Note that when $V$ is convex, %if $\acc=\frac12V'(\redge)$, then 
$\acc\notin\mathcal{J}_V$ from Lemma~\ref{lem:Gprop}\ref{lemma_enu:e}. %  when $V$ is convex. 
For a non-convex $V$, as indicated in Remark \ref{rmk:nonconvexity_and_sec_crit}, $\acc$ may or may not be in $\mathcal{J}_V$ no matter if $\acc=\frac12V'(\redge)$. See typical graphs of $\Gfn$ and $\Hfn$ for $a \in \mathcal{J}_V$ in Figure \ref{figure:newGH_3}.

We have the following Lemma.

\begin{lemma}\label{lem:x0} 
\begin{itemize}
\item[(a)]
For $a\in [\acc, \infty)$ such that $a \not\in \mathcal{J}_V$, let $x_0(a)$ be the unique point in $[c(a), \infty)$ at which $\Gfn(x;a)$ attains its maximum. Then 
$x_0(a)$ is a continuous, strictly increasing function in $a\in [\acc, \infty)\setminus \mathcal{J}_V$. 
\item[(b)]
If $a_0\in \mathcal{J}_V$ and $a_0>\acc$, then 
\begin{equation}
	\lim_{a\uparrow a_0} x_0(a)< \lim_{a\downarrow a_0} x_0(a) .
\end{equation}
\end{itemize}
\end{lemma} 

%\begin{rmk}
%Throughout the paper, the notation $x_0(a)$ is defined as in Lemma \ref{lem:x0} except for in Subsection \ref{subsection:Proof_of_Theorem1.4}.
%\end{rmk}

%Before giving the proof of lemma \ref{lem:x0}, we remark the following property of $\Gfn(x;a)$ for $a \in \mathcal{J}_V$.
Note that if $a\in \mathcal{J}_V$ satisfies $a>\acc$ or $a=\acc<\frac12 V'(\redge)$, then there exist points $ x_1(a) < x_2(a) < \dots < x_r(a)$ in $(c(a), \infty)$, for some $r\ge 2$, such that 
\begin{equation} \label{eq:the_r_maxima_not_traditional_critical}
	\Gfn_{\max}(a)=\Gfn(x_1(a);a)= \cdots = \Gfn(x_r(a); a) . 
	%\textnormal{ attains its maximum in $[c(a), \infty)$ at } x_i(a), \quad i=1, \dots, r.
\end{equation}
On the other hand, if $V$ is a potential such that $\acc= \frac12 V'(\redge)$ and $\acc\in \mathcal{J}_V$, then there exist, for some $r\ge 1$,  $x_1(\acc) < x_2(\acc) < \dots < x_r(\acc)$ in $(\redge, \infty)$ such that 
\begin{equation} \label{eq:the_r_maxima_traditional_critical}
	\Gfn(\redge;a)= \Gfn_{\max}(a) = \Gfn(x_1(\acc);\acc) = \cdots = \Gfn(x_r(\acc); \acc).
\end{equation}

\begin{proof}[Proof of Lemma~\ref{lem:x0}]
The continuity of $x_0(a)$ for $a \not\in \mathcal{J}_V$ is a direct consequence of the continuity of $\Gfn(x;a)$ in both $x$ and $a$. Let $\acc \leq a_1 < a_2$ and $a_1, a_2 \not\in \mathcal{J}_V$. If we assume $x_0(a_1) \geq x_0(a_2)$, then since $\Gfn(x_0(a_1);a_1) > \Gfn(x_0(a_2);a_1)$, we have 
\begin{equation}
\Gfn(x_0(a_1);a_2) = \Gfn(x_0(a_1);a_1) + (a_2-a_1)x_0(a_1) > %\\
\Gfn(x_0(a_2);a_1) + (a_2-a_1)x_0(a_2) = \Gfn(x_0(a_2);a_2).
\end{equation}
This is contradictory to the assumption that $x_0(a_2)$ is the maximizer of $\Gfn(x;a_2)$. Thus $x_0(a_1) < x_0(a_2)$. 

If $a_0 \in \mathcal{J}_V$  and $a_0>\acc$, then $\Gfn(x;a_0)$ attains its maximum in $[c(a_0), \infty)$ at $x_1(a_0), \dots, x_r(a_0)$ for some $r\ge 2$ as in \eqref{eq:the_r_maxima_not_traditional_critical}. It is easy to check from the continuity of $\Gfn(x;a)$ in $a$ that $\lim_{a\uparrow a_0} x_0(a) = x_1(a_0)$ and $\lim_{a\downarrow a_0} x_0(a) = x_r(a_0)$.
\end{proof}

\begin{defn} \label{defn:definition_of_secondary_critical_values}
The  \underline{secondary critical values} for the spiked model are defined as the points $a\in\mathcal{J}_V$ such that $a>\acc$.
\end{defn}

\begin{rmk}\label{rmk:convJ}
For a potential $V$ such that $V(x)$ is convex for $x\ge \redge$, $\mathcal{J}_V=\emptyset$ since $\Gfn'(x;a)$ is a decreasing function in $x\ge \redge$. Hence there is no secondary critical value. 
\end{rmk}

%We define the points in $\mathcal{J}$ as the \underline{jump transition points}. This terminology will be clear from Theorem~\ref{thm:rank1LLN} below. 
%Note that for the potential such that $\acc<\frac12V'(\redge)$, $\acc$ is a jump transition point from Lemma~\ref{lem:Gprop} (f). 

\subsection{Statement of results: non-convex potentials} \label{subsection:main_results_rank_1}

%In this section, we do not assume that $V(x)$ is convex in $x\ge \redge$. 
Let $V(x)$ be a potential satisfying the conditions~\eqref{eq:condition_of_V_1}--\eqref{eq:condition_of_V_3}.
Let $\Int$ and $\Intx(x_*)$ be the intervals defined in~\eqref{eq:interval} and~\eqref{eq:interval2}, respectively. 

\begin{thm}[away from critical values]
\label{thm:thm_rank_1}
The following holds for each $T\in\R$ as $n\to\infty$ and $j=O(1)$.
\begin{enumerate}[label=(\alph*)]
\item \label{enu:thm:thm_rank_1:a} For $\aaa<\acc$,  
\begin{equation} %\label{eq:main2TW}
	\lim_{n \rightarrow \infty} \Prob_{n-j+1,n} \left( \aaa; \Int \right) = \FGUE(T).
\end{equation}

\item \label{enu:thm:thm_rank_1:b} For $\aaa>\acc$ such that $\aaa\notin\mathcal{J}_V$, if $\Gfn''(x_0(\aaa))\neq 0$, then 
\begin{equation}\label{eq:main2erf}
	\lim_{n \rightarrow \infty} \Prob_{n-j+1,n} \left( \aaa; \Intx(x_0(a)) \right) = \erf(T).
\end{equation}

\end{enumerate}
\end{thm}

%The cases which are not covered by the above theorem and remark are  the critical case $\aaa=\acc$  when the potential $V$ is such that $\acc<\frac12 V'(\redge)$ and  the case when $\aaa\in \mathcal{J}_V$. Then next theorems and remarks cover all the remaining cases. 

When $a$ is at or near the critical value $\acc$, we have the following result. The case when $a=\acc$ is attained by setting $\alpha=0$.

\begin{thm}[at or near the critical value] 
\label{thm:critical_traditional_split}
We have the following for each $T\in\R$.
\begin{enumerate}[label=(\alph*)]
\item \label{enu:thm:critical_traditional_split:a} Suppose that  $V$ is a potential such that $\acc=\frac12 V'(\redge)$ and $\acc\notin \mathcal{J}_V$. Then for
\begin{equation}
	\aaa=\acc+\frac{\beta\alpha}{ n^{1/3}},
\end{equation}
where $\alpha$ is in a compact subset of $\R$, we have 
\begin{equation}
	\lim_{n \rightarrow \infty} \Prob_{n-j+1,n} \left( \aaa; \Int\right) = \FGOE(T; -\alpha).
\end{equation}

\item \label{enu:thm:critical_traditional_split:b} Let $V$ be  a potential such that $\acc<\frac12 V'(\redge)$. If $\acc\notin \mathcal{J}_V$ and $\Gfn''(x_0(\acc);\acc)\neq 0$, then for
\begin{equation}
	\aaa=\acc+\frac{\alpha}{ n},
\end{equation}
where $\alpha$ is in a compact subset of $\R$, we have, as $n\to\infty$, 
\begin{equation} \label{eq:part_1_of_thm_4}
	 \Prob_{n-j+1,n} \left( \aaa; \Int\right)  = p_{j,n}(\alpha) \FGUE(T) +o(1)
\end{equation}
and 
\begin{equation} \label{eq:part_2_of_thm_4}
	\Prob_{n-j+1,n} \left( \aaa; \Intx(x_0(\acc))\right)  = p_{j,n}(\alpha) + (1-p_{j,n})(\alpha)\erf(T) +o(1), 
%	\Prob_{n-j+1,n} \left( \aaa; \Intx(x_0(\acc))\right)  = p_{j,n}(\alpha) + p^{(0)}_{j,n}(\alpha)\erf(T) +o(1)
\end{equation}
where the constant $p_{j,n}(\alpha)  \in (0,1)$ is defined by~\eqref{eq:defn_of_p(alpha)_pseudo_crit_1}. As a function of $\alpha$, $p_{j,n}(\alpha)$ is decreasing and satisfies $p_{j,n}(\alpha)\to 0$ as $\alpha\to \infty$ and $p_{j,n}(\alpha)\to 1$ as $\alpha\to -\infty$ for each fixed $n$. Also for each fixed $\alpha$, $p_{j,n}(\alpha)$ lies in a compact subset of $(0,1)$ for all large $n$. 

%where $p_{j,n}(\alpha)$ and $p^{(0)}_{j,n}(\alpha) \in (0,1)$ are defined by~\eqref{eq:defn_of_p(alpha)_pseudo_crit_1}, and $p_{j,n}(\alpha) + p^{(0)}_{j,n}(\alpha) = 1$. As a function of $\alpha$, $p_{j,n}(\alpha)$ is decreasing and satisfies $p_{j,n}(\alpha)\to 0$ as $\alpha\to \infty$ and $p_{j,n}(\alpha)\to 1$ as $\alpha\to -\infty$ for each fixed $n$. Also for each fixed $\alpha$ $p_{j,n}(\alpha)$ is in a compact subset of $(0,1)$ independent of $n$. 
\end{enumerate}
\end{thm}

\begin{rmk}
When the potential $V(x)$ is convex for $x\ge\redge$, then $\acc = \frac12 V'(\redge)$, $\mathcal{J}_V=\emptyset$ (see Remark~\ref{rmk:convJ}) and $\Gfn''(x)<0$ for all $x>\redge$. Hence Theorem~\ref{thm:thm_rank_1} and Theorem~\ref{thm:critical_traditional_split}\ref{enu:thm:critical_traditional_split:a} imply Theorem~\ref{thm:convex}.  
\end{rmk}

For $a$ at or near the  secondary critical values $\mathcal{J}_V\setminus\{\acc\}$, we have the following result. 

\begin{thm}[at or near the secondary critical values]
\label{thm:supercritical_split}
Let $V$ be a potential such that $\mathcal{J}_V\neq \emptyset$.
Let $\aaa_0\in \mathcal{J}_V\setminus\{\acc\}$ be a secondary critical point. If $\Gfn(x; \aaa_0)$ attains its maximum $\Gfn_{\max}(a_0)$ at two points $x_1(\aaa_0)<x_2(\aaa_0)$ in $(c(a), \infty)$ and if $\Gfn''(x_1(\aaa_0); \aaa_0)\neq 0$ and $\Gfn''(x_2(\aaa_0); \aaa_0)\neq 0$, then for
\begin{equation}
	\aaa = \aaa_0+ \frac{\alpha}{n},
\end{equation}
where $\alpha$ is in a compact subset of $\realR$, we have 
\begin{equation}
\Prob_{n-j+1,n} \left( \aaa; \Intx(x_1(a_0)) \right) = p^{(1)}_{j,n}(\alpha)\erf(T) +o(1)
\end{equation}
and 
\begin{equation}
\Prob_{n-j+1,n} \left( \aaa; \Intx(x_2(a_0)) \right) = p^{(1)}_{j,n}(\alpha) + p^{(2)}_{j,n}(\alpha)\erf(T) +o(1),
\end{equation} 
where $p^{(1)}_{j,n}(\alpha)$ and $p^{(2)}_{j,n}(\alpha) \in (0,1)$ are defined in~\eqref{eq:defpj1} and~\eqref{eq:defpj2}, and $p^{(1)}_{j,n}(\alpha) + p^{(2)}_{j,n}(\alpha) = 1$. As a function of $\alpha$, $p^{(1)}_{j,n}(\alpha)$ is decreasing and satisfies $p^{(1)}_{j,n}(\alpha) \to 0$ as $\alpha\to \infty$ and $p^{(1)}_{j,n}(\alpha) \to 1$ as $\alpha\to -\infty$ for each fixed $n$. Also for each fixed $\alpha$, $p^{(1)}_{j,n}(\alpha)$ is in a compact subset of $(0,1)$ independent of $n$. 
\end{thm}

The above three theorems describe the `generic' cases. 
The next part describes the three  `exceptional cases'. 

As the first exceptional case, 
suppose that in Theorem~\ref{thm:supercritical_split},  the maximum of $\Gfn$ is attained at more than two points. %Suppose that $\Gfn_{\max}(a_0)$ is attained at 
Let $x_1(a)<x_2(a)<\cdots<x_r(a)$ be these maximizers. If $\Gfn''(x_j(a_0);a_0)\neq 0$ for all $k=1,\cdots, r$, then we have, for each $k=1,\cdots,r$,
\begin{equation}\label{eq:threeormoremax}
	\Prob_{n-j+1,n} \left( \aaa; \Intx(x_k(a_0)) \right) = \sum_{i=1}^{k-1} p^{(i)}_{j,n}(\alpha) + p^{(k)}_{j,n}(\alpha)\erf(T)+o(1), 
\end{equation} 
for some $p_{j,n}^{(i)}(\alpha)\in (0,1)$ such that $p^{(1)}_{j,n}(\alpha)+\cdots +p^{(r)}_{j,n}(\alpha)=1$. 
%where $p_j(\alpha)=
Explicitly, $p^{(i)}_{j,n}(\alpha)  :=\frac{A_i(\alpha)}{A_1(\alpha)+\cdots + A_r(\alpha)}$ where $A_i(\alpha)$ is defined in~\eqref{eq:defpj2}.
The situation when $\acc\in \mathcal{J}_V$ in Theorem~\ref{thm:critical_traditional_split}\ref{enu:thm:critical_traditional_split:b} is similar. In this case, the maximum of $\Gfn(x;\acc)$ in $(c(\acc), \infty)$ is attained at $x_1(\acc)<x_2(\acc)<\cdots<x_r(\acc)$ for some $r\ge 2$ (see~\eqref{eq:the_r_maxima_not_traditional_critical}). Assume that $\Gfn''(x_i(\acc); \acc)\neq 0$ for all $i=1, \cdots, r$. Then with $C_i(\alpha)$, $i=1,\cdots, r$, defined by~\eqref{eq:defn_of_C_j(alpha)_pseudo_crit_1} with
$x_0(\acc)$ replaced by $x_i(\acc)$, set $p_{j,n}^{(i)}(\alpha):= \frac{C_i(\alpha)}{C_0+C_1(\alpha)+\cdots+ C_r(\alpha)}$, $i=1,\cdots, r$, where $C_0$ is defined by \eqref{eq:defn_of_C_j(alpha)_pseudo_crit_0}. Then~\eqref{eq:part_1_of_thm_4} holds with $p_{j,n}(\alpha)$ replaced by $p^{(0)}_{j,n}(\alpha) := 1 - p^{(1)}_{j,n}(\alpha) - \dots - p^{(r)}_{j,n}(\alpha)$. The limit in~\eqref{eq:part_2_of_thm_4} is replaced by 
\begin{equation}\label{eq:ppaa11}
	 \Prob_{n-j+1,n} \left( \aaa; \Intx(x_k(a_0)) \right) = \sum_{i=0}^{k-1} p_{j,n}^{(i)}(\alpha) + p_{j,n}^{(k)}(\alpha)\erf(T) +o(1)
\end{equation} 
for $k=1,\cdots,r$.

The second exceptional case is when $\acc\in \mathcal{J}_V$ in Theorem~\ref{thm:critical_traditional_split}, case 1. This case is given in the following Theorem. In this case, there are two natural scalings in $a$.

\begin{thm} \label{thm:rank_1_transit_crit}
Let $V$ be a potential such that $\acc= \frac12 V'(\redge)$. Suppose that $\acc\in \mathcal{J}_V$. Assume that $\Gfn(x; \acc)$ attains its maximum at the unique point $x_0(\acc)\in (c(\acc), \infty)$ and $\Gfn''(x_0(\acc); \acc)\neq 0$. Then the following holds.
\begin{enumerate}[label=(\alph*)]
\item \label{enu:thm:rank_1_transit_crit:a}
For 
\begin{equation}\label{eq:aina}
	\aaa=\acc+\frac{\beta \alpha}{ n^{1/3}},
\end{equation}
where $\alpha$ is in a compact subset of $(-\infty, 0)$, we have 
\begin{equation}
	\lim_{n \rightarrow \infty} \Prob_n \left( \aaa; \Int\right) = F_1(T; -\alpha).
\end{equation}
\item \label{enu:thm:rank_1_transit_crit:b}
For 
\begin{equation}\label{eq:ainb}
	\aaa = \acc+ \frac{\alpha'}{n},
\end{equation}
where $\alpha'$ is in a compact subset of $\realR$, we have 
\begin{equation}\label{eq:ppaarr}
	 \Prob_{n-j+1,n} \left( \aaa; \Int \right) = p_{j,n}(\alpha') \FGOE(T) +o(1)
\end{equation}
where $\FGOE(T)=\FGOE(T; 0)$, 
and 
\begin{equation}\label{eq:ppaa}
	\Prob_{n-j+1,n} \left( \aaa;\Intx(x_0(\acc)) \right) = p_{j,n}(\alpha') + (1-p_{j,n}(\alpha'))\erf(T)
\end{equation}
where $p_{j,n}(\alpha')$ is defined in \eqref{eq:tq30}. As a function of $\alpha$, $p_{j,n}(\alpha')$ is decreasing and satisfies $p_{j,n}(\alpha')\to 0$ as $\alpha'\to \infty$ and $p_{j,n}(\alpha')\to 1$ as $\alpha'\to -\infty$ for each fixed $n$. Also for each fixed $\alpha$, $p_{j,n}(\alpha)$ is in a compact subset of $(0,1)$ independent of $n$. 

\end{enumerate}
\end{thm}
 
If the maximum of $\Gfn(x; \acc)$ is attained at more than one point, then~\eqref{eq:ppaa} should be changed in a natural way as in~\eqref{eq:ppaa11}.

The third exceptional case is when the double derivative of $\Gfn(x;\aaa)$ vanishes at its maximizers. Then the function $\erf(x)$ is replaced by its higher analogue, and the scalings in the interval and $a$ are also changed accordingly. Concretely, in 
Theorem~\ref{thm:thm_rank_1}\ref{enu:thm:thm_rank_1:b},  if $\Gfn''(x_0(\aaa))=0$, then  since $x_0(\aaa)$ is the maximum point, there exists $k > 1$ such that  $\Gfn^{(2k)}(x_0(\aaa))<0$ and $\Gfn^{(j)}(x_0(\aaa))=0$ for all $j=1,2,\cdots, 2k-1$. Then~\eqref{eq:main2erf} is changed to  
\begin{equation}\label{eq:higher1010}
	\lim_{n \to \infty} \Prob_{n-j+1,n}\left( \aaa; \Intxx(x_0(a);k) \right) = 
 \frac{\int_{-\infty}^T e^{-x^{2k}}dx}{\int_{-\infty}^{\infty} e^{-x^{2k}}dx},
\end{equation} 
where the interval $\Intxx(x_*;k)$ is defined by 
\begin{equation}\label{eq:highererf}
\Intxx(x_*;k):= \left[ x_* + \left( \frac{n(V^{(2k)}(x_*)) - \gfn^{(2k)}(x_*))}{(2k)!}  \right)^{-1/(2k)}T, \infty \right).
\end{equation}
In Theorem~\ref{thm:critical_traditional_split}\ref{enu:thm:critical_traditional_split:b}, Theorem~\ref{thm:supercritical_split} and Theorem~\ref{thm:rank_1_transit_crit}, $a$ is scaled as $a=a_0+\frac{\alpha}{n}$. When $\Gfn''(x_i(a_0);a_0)=0$, ($i=0$ in Theorems \ref{thm:critical_traditional_split}\ref{enu:thm:critical_traditional_split:b} and \ref{thm:rank_1_transit_crit}, and $i=1,2$ in Theorem~\ref{thm:supercritical_split},) then this scaling also needs to be changed. For example, Theorem~\ref{thm:supercritical_split} is changed to the following theorem. 

\begin{thm}\label{thm:noname}
Let $\aaa_0\in \mathcal{J}_V\setminus\{\acc\}$ be a secondary critical value. Assume that $\Gfn(x; \aaa_0)$ attains its maximum $\Gfn_{\max}(a_0)$ at two points $x_1(\aaa_0)<x_2(\aaa_0)$ in $(c(a), \infty)$. 
Suppose that $\Gfn''(x_1(\aaa_0); \aaa_0)\neq 0$, and for some $k>1$, and suppose that $\Gfn^{(2k)}(x_2(\aaa_0); \aaa_0) \neq 0$ and $\Gfn^{(i)}(x_2(\aaa_0); \aaa_0) = 0$ for all $i = 1, \dots, 2k-1$. Then for
\begin{equation}
	\aaa = \aaa_0 - q\frac{\log n}{n} + \frac{\alpha}{n}, \quad \textnormal{where} \quad q := \frac{\frac{1}{2}-\frac{1}{2k}}{x_2(\aaa_0)-x_1(\aaa_0)}
\end{equation}
where $\alpha$ is in a compact subset of $\realR$, we have 
\begin{equation}
	 \Prob_{n-j+1,n} \left( \aaa; \Intx(x_1(a_0)) \right) = p^{(1)}_{j,n}(\alpha)\erf(T) +o(1),
\end{equation}
and 
\begin{equation}
	\Prob_{n-j+1,n} \left( \aaa; \Intxx(x_2(a_0);k) \right) = p^{(1)}_{j,n}(\alpha) + p^{(2)}_{j,n}(\alpha)  \frac{\int_{-\infty}^T e^{-x^{2k}}dx}{\int_{-\infty}^{\infty} e^{-x^{2k}}dx} +o(1),
\end{equation} 
where the interval $\Intxx(x;k)$ is defined by~\eqref{eq:highererf}, $p^{(1)}_{j,n}(\alpha)$ and $p^{(2)}_{j,n}(\alpha)$ are defined in~\eqref{eq:defpj1_2}, and $p^{(1)}_{j,n}(\alpha) + p^{(2)}_{j,n}(\alpha) = 1$. As a function of $\alpha$, $p^{(1)}_{j,n}(\alpha)$ is decreasing and satisfies $p^{(1)}_{j,n}(\alpha) \to 0$ as $\alpha\to \infty$ and $p^{(1)}_{j,n}(\alpha) \to 1$ as $\alpha\to -\infty$ for each fixed $n$. Also for each fixed $\alpha$, $p^{(1)}_{j,n}(\alpha)$ is in a compact subset of $(0,1)$ independent of $n$.
\end{thm}

The changes needed for Theorem~\ref{thm:critical_traditional_split}\ref{enu:thm:critical_traditional_split:b},  and Theorem~\ref{thm:rank_1_transit_crit} are analogous. Also it may happen that the two or more of the exceptional cases occur simultaneously. Then one needs simply combine the results together in a straightforward way, and we skip the details. 

\bigskip
We remark that if the support $J$ of the equilibrium measure is of one interval i.e. $N=0$ (see~\eqref{eq:Jend}), then the probabilities $p_{j,n}$ and $p_{j,n}^{(j)}$ in Theorem \ref{thm:critical_traditional_split}\ref{enu:thm:critical_traditional_split:b},  \ref{thm:supercritical_split}, \ref{thm:rank_1_transit_crit}\ref{enu:thm:rank_1_transit_crit:b} and \ref{thm:noname} do not depend on $n$. This follows from Remark \ref{rmk:when_N=0} on $\M_{j,n}$, $\tilde{\M}_{j,n}$ and $\B_{j,n}(\redge)$ and the definition of these probabilities. When $N>0$, the dependence of these probabilities on $n$ is from the theta function formula of  $\M_{j,n}$, $\tilde{\M}_{j,n}$ and $\B_{j,n}(\redge)$ in Section \ref{section:Result_of_RHP}, and is 
in a quasi-periodic way.

\bigskip

We also remark that one can obtain the convergence in probability 1 as in~\eqref{eq:GUE1} from the above theorems together with the fact that all of the  limiting distributions decay rapidly at the tails. 

\bigskip

An explicit example of a potential such that $\acc<\frac12 V'(\redge)$ can be constructed as follows. 
We use the potential defined in \cite[Formula (4.14)]{Eynard06} (we change the original notation $e$ into $\bar{e}$ here): 
\begin{equation}
V_{\bar{e},\epsilon}(x) = \frac{1-\epsilon}{1+\bar{e}\tilde{e}} \left( \frac{1}{4}x^4 - \frac{\bar{e}+\tilde{e}}{3}x^3 + \frac{\bar{e}\tilde{e}-2}{2}x^2 + 2(\bar{e}+\tilde{e})x \right),
\end{equation}
where $\epsilon$ is a very small positive number, $\bar{e}$ is any number $>2$ and $\tilde{e}$ is determined by $e$ from the condition that 
\begin{equation}
\int^{\bar{e}}_2 (x-\bar{e})(x-\tilde{e}) \sqrt{x^2-4} dx = 0.
\end{equation}
From results in \cite[Section 4]{Eynard06}, it is known that $V_{\bar{e},\epsilon}$ is a regular potential with the support of the equilibrium measure given by $[-2 + O(\epsilon), 2 + O(\epsilon)]$. For all $x > \redge =2 + O(\epsilon)$, \eqref{eq:consequence_or_regularity} holds. However, at $x = \bar{e}$,
\begin{equation} \label{eq:birth_of_cut}
2\gfn(\bar{e}) -V(\bar{e}) - \ell = -E(\epsilon)
\end{equation}
for some $E(\epsilon)$ satisfying $E(\epsilon)=O(\epsilon)$ and $E(\epsilon)>0$. 
Hence for any $a$, $\Gfn(\bar{e};a) - \Hfn(\bar{e};a) = O(\epsilon)$. 
Now there exists $a\in (0, \frac12 V'(\redge))$ such that $c(a) \in (\redge, \bar{e})$ since the minimizer $c(\aaa)$ of $\Hfn(x;\aaa)$ is continuous in $a\in (0, \infty)$,  decreases strictly in $\aaa\in (0, \frac12 V'(\redge))$,  $\lim_{a\downarrow 0} c(a)=+\infty$ and $c(\frac12 V'(\redge))= \redge$ (see the sentence after Definition~\ref{def:ca}).
Since $\Hfn(c(a);a) < \Hfn(\bar{e};a)$, we have
$\Gfn(\bar{e};a) > \Hfn(c(a);a)$ if $\epsilon$ is small enough. Then $a\in \mathcal{A}_V$ and $\acc<\frac12 V_{\bar{e},\epsilon}'(\redge)$ for each fixed $\bar{e}>2$ if $\epsilon>0$ is small enough.

%One can also construct an example from $V(x)=\frac16 x^6+\frac14 \alpha x^4+ \frac12 \beta x^2$ for specific $\alpha<0$ and $\beta>0$. 

%For fixed $\bar{e}$ and $\epsilon$, we can find $\redge$, $\gfn(z)$ and $\ell$ by solving a variational problem. Numerical results show that for generic $\bar{e}$ and $\epsilon$, $\acc \not\in \mathcal{J}_V$ and Theorem \ref{thm:critical_traditional_split}\ref{enu:thm:critical_traditional_split:a} applies.

\bigskip

The paper is organized as follows. The outline of the proof of theorems is given in Section~\ref{subsection:outline_of_the_proof}. The results on the orthonormal polynomials and the kernel $K_{n-j,n}$ are summarized in Section \ref{section:Result_of_RHP}. The proofs of the theorems are given in Sections~\ref{subsection:generic_case_1},~\ref{subsection:generic_case_3} and~\ref{subsection:The_critical_case_rank_1}. We consider three cases, $\acc<\frac12 V'(\redge)$, $\acc>\frac12 V'(\redge)$ and $\acc=\frac12 V'(\redge)$ separately. Throughout this paper we only consider $a > 0$. The $a < 0$ case is discussed briefly in the end of Section \ref{subsection:outline_of_the_proof}.

\subsubsection*{Acknowledgments}
We would like to thank Marco Bertola, Robbie Buckingham, Seung-Yeop Lee and Virgil Pierce for keeping us informed of the progress of their work. 
The work of Jinho Baik was supported in part by NSF grants DMS075709.

%\section{Structure of Proof}

\section{Outline of the proof} \label{subsection:outline_of_the_proof}

Let 
\begin{equation} \label{eq:definition_of_p_and_gamma}
p_j(x;n) = \gamma_j(n)x^j + \cdots
\end{equation}
be the orthonormal polynomial of degree $j$ with respect to the weight $e^{-nV(x)}$. Here take $\gamma_j(n) > 0$ to make $p_j$ unique. Set
\begin{equation} \label{eq:definition_of_psi_and_varphi}
\psi_j(x;n) := p_j(x;n)e^{-\frac{n}{2}V(x)}, \qquad \varphi_j(x;n) := p_j(x;n)e^{-nV(x)}.
\end{equation}
%By definition $\langle \psi_j(x), \psi_k(x) \rangle_{\realR} = \delta_{jk}$. 
Let
\begin{equation} \label{eq:Christoffel_Darboux_K_n-1,n}
%\begin{split}
	K_{j, n}(x,y)  :=  \sum_{i=0}^{j-1} \psi_i(x;n)\psi_i(x;n) %\\
	=  \frac{\gamma_{j-1}(n)}{\gamma_{j}(n)} \frac{\psi_{j}(x;n)\psi_{j-1}(y;n)-\psi_{j-1}(x;n)\psi_{j}(y;n)}{x-y}
%\end{split}
\end{equation}
be the Christoffel-Darboux kernel.
Define the constant 
\begin{equation} \label{eq:new_formula_of_tilde_gamma} 
	\bfGamma_j(a;n) := \langle e^{n(ax-V(x)/2)}, \psi_j(x;n) \rangle_{\realR}, 
\end{equation}
%where the notation $\langle, \rangle_E$ is the inner product on $E$, 
%defined below \eqref{eq:defn:of_F_1(T;alpha)}, 
and the function
\begin{equation} \label{eq:new_formula_of_tilde_psi}
\begin{split}
	\tilde{\psi}_j(x;a;n) 
	:= & \frac1{\bfGamma_j(a;n)} \left( e^{n(ax-V(x)/2)} - \sum^{j-1}_{i=0} %\langle e^{n(ay-V(y)/2)}, \psi_i(y;n) \rangle_{\realR} 
	\bfGamma_i(a;n) \psi_i(x;n) \right)\\
	= & \frac1{\bfGamma_j(a;n)} \left( e^{n(ax-V(x)/2)} - \int_{\R} K_{j,n}(x, y) e^{n(ay-V(y))} dy \right).
\end{split}
\end{equation}
It is easy to check that $\tilde{\psi}_j(x;a;n)$ is characterized by the orthonormality conditions
\begin{equation}
\langle \tilde{\psi}_j(x;a;n), \psi_k(x;n) \rangle_{\realR} = \delta_{jk}, \quad \textnormal{for $k = 0, 1, \dots, j$,}
\end{equation}
in the vector space spanned by $\{e^{-\frac{n}{2}V(x)}, xe^{-\frac{n}{2}V(x)}, \cdots, x^{j-1}e^{-\frac{n}{2}V(x)}, e^{n(ax-V(x)/2)}\}$.
In the multiple orthogonal polynomial theory, $\tilde{\psi}_j(x;a;n)e^{nV(x)/2}$ is the multiple orthogonal polynomial of type I with potentials $e^{-nV(x)}$ and $e^{ax}$, see \eg\ \cite{Bleher-Kuijlaars04a}. From this follows the well-definedness of $\tilde{\psi}_j(x;a;n)$, \ie, $\bfGamma_j(a;n) \neq 0$.

We sometimes drop the dependence on $n$ or $a$ in $\psi_j(x;n)$, $\bfGamma_j(a;n)$ and $\tilde{\psi}_j(x;a;n)$ and write $\psi_j(x)$, $\bfGamma_j(a)$ and $\tilde{\psi}_j(x)$ for convenience. 

The starting point of our analysis is the following basic result in the theory of Hermitian matrix model with external source, specialized to the spiked source model of rank $1$(see \eg\ \cite{Zinn_Justin97}, \cite{Bleher-Kuijlaars04a}):  for any $E \subset \realR$, 
\begin{equation}\label{eq:Fredholm_det_formulahr1}
	\Prob_{n-j+1,n}(a; E) = \Freddet \left( 1 - \chi_{E} \tilde{K}_{n-j+1, n} \chi_{E} \right),
\end{equation}
where 
\begin{equation}\label{eq:Fredholm_det_formulahr2}
	\tilde{K}_{n-j+1,n} := K_{n-j, n} + \tilde{\psi}_{n-j}\otimes \psi_{n-j}.
\end{equation}
Here $\chi_{E}$ denotes the projection operator on $E$ and $K_{j,n}$ is the operator on $L^2(\R)$, defined by the kernel~\eqref{eq:Christoffel_Darboux_K_n-1,n}. 
Note that the only term in~\eqref{eq:Fredholm_det_formulahr2} that depends on $a$ is $\tilde{\psi}_{n-j}(x;a;n)$.
The kernel $K_{n-j,n}(x,y)$ is precisely  the reproducing kernel in the Hermitian random matrix model of size $n-j$ with weight $e^{-nV}$ with no external source. 
Hence for the rank 1 spiked Hermitian model, the reproducing kernel is a rank $1$ perturbation of $K_{n-j, n}$. For the higher rank spiked Hermitian model, the reproducing kernel is a rank $r$ perturbation of $K_{n-j, n}$ which will be studied in the subsequent paper. 
%In the below, we denote $\tilde{\psi}_{n-j}(x;a;n)$ and $\psi_{n-j}(x;n)$ by $\tilde{\psi}_{n-j}(x)$ and $\psi_{n-j}(x)$ for the notational compactness. 

For the asymptotic result for this paper, we need asymptotics of $K_{n-j, n}$, $\psi_{n-j}$ and $\tilde{\psi}_{n-j}$. 
The asymptotics of orthogonal polynomials and the Christoffel-Darboux kernel with respect to a varying weight have been studied extensively. Most notably precise strong asymptotics were obtained for a general class of potentials using the Deift-Zhou steepest-descent method for the associated Riemann-Hilbert problem (see \cite{Deift-Kriecherbauer-McLaughlin-Venakides-Zhou99} and \cite{Deift-Gioev07a}). We use the results of \cite{Deift-Kriecherbauer-McLaughlin-Venakides-Zhou99} extensively. However, in \cite{Deift-Kriecherbauer-McLaughlin-Venakides-Zhou99}, only the case of $j=0,1$ are stated explicitly. For more general $j=O(1)$, the same analysis of \cite{Deift-Kriecherbauer-McLaughlin-Venakides-Zhou99} can be carried out after a few changes. These have been studied in various other papers (see e.g. \cite{Baik-Kriecherbauer-McLaughlin-Miller07} for the discrete orthogonal polynomials case). We summarize the necessary changes and state the explicit asymptotic formulas in Section~\ref{section:Result_of_RHP} below.

The main part of this paper is the asymptotic analysis of $\tilde{\psi}_{n-j}(x;a;n)$ defined in~\eqref{eq:new_formula_of_tilde_psi}. For this purpose, we first evaluate the asymptotics of $\bfGamma_{n-j}(a;n)$. This can be achieved in principle by plugging in the asymptotics of $\psi_{n-j}(x;n)$ in the definition~\eqref{eq:new_formula_of_tilde_gamma} of $\bfGamma_{n-j}(a;n)$, and then evaluating the inner product asymptotically. However, the oscillatory nature of $\psi_{n-j}(x;n)$ in the support of the equilibrium measure makes it cumbersome to evaluate the inner product in this way. Instead we re-express $\bfGamma_{n-j}(a;n)$ in terms of a sum of integrals involving both $\psi_{n-j}(x;n)$ and its Cauchy transform (see~\eqref {eq:division_of_psi_n_exponent}). This removes the oscillation  and the asymptotic analysis becomes more straightforward. A similar trick is also used in evaluating $\int_{\R} K_{n-j,n}(x,y) e^{n(ay-V(y))}dy$ (see Lemma~\ref{lem:tpsinew}). The analysis is divided into several cases depending on the location of the critical points. Each of these cases correspond to the theorems in Introduction. 

For the  set $E=\Int$ or $E=\Intx(x_*)$ of interest in each theorem, we can show that $\tilde{\psi}_{n-j}(x;a;n)$ is in $L^2(E)$ for each $n$, and hence $\chi_E\tilde{\psi}_{n-j}\otimes \psi_{n-j}\chi_E$ in~\eqref{eq:Fredholm_det_formulahr2} is a trace class operator. However, the $L^2$-norm of $\tilde{\psi}_{n-j}$ is not uniformly bounded in $n$ nor there is a simple estimate on the $L^2$ norm. For example, we will find that in~\eqref{eq:tildepsiwhenalessacc} that 
$\tilde{\psi}_{n-j}(x;a;n) = O( \sqrt{n} e^{n(\Gfn(x)+\Hfn(x)-2\Hfn(c))/2} )$ for $x>c$ when $a<\min\{\acc, \frac12 V'(\redge)\}$. The function $\Gfn(x)+\Hfn(x)-2\Hfn(c)= -V(x)+2a(x-c)+2\gfn(c)-\ell$ tends to $-\infty$ as $x\to +\infty$, but it may be positive for some value $x>c$. This implies that we do not have a good trace norm of $\chi_E\tilde{\psi}_{n-j}\otimes \psi_{n-j}\chi_E= \|\tilde{\psi}_{n-j}\|_{L^2(E)} \|\psi_{n-j}\|_{L^2(E)}$, and we cannot compare the size of $\chi_E K_{n-j,n}\chi_E$ and $\chi_E\tilde{\psi}_{n-j}\otimes \psi_{n-j}\chi_E$. However, the rapid decay of the operator $(1-\chi_E K_{n-j,n}\chi_E)^{-1}$ can be used to control the estimates. We proceed as follows.

%%%%%%%

From~\eqref{eq:Fredholm_det_formulahr1}
and~\eqref{eq:Fredholm_det_formulahr2}, 
\begin{equation}\label{eq:maind}
\begin{split}
	& \Prob_{n-j+1,n}(a; E)= \det \big( 1 - \chi_E \tilde{K}_{n-j+1, n} \chi_E \big) \\ 
	= & \det \left( 1 - \chi_E K_{n-j,n} \chi_E \right) \cdot \det \big( 1 - \left(1 - \chi_E K_{n-j,n} \chi_E \right)^{-1}  \chi_E\tilde{\psi}_{n-j}  \otimes \psi_{n-j}  \chi_E \big)\\
	= & \det \left( 1 - \chi_E K_{n-j,n} \chi_E \right) \cdot  \big[ 1 - \langle \left( 1 - \chi_E K_{n-j,n} \chi_E \right)^{-1}  \tilde{\psi}_{n-j} ,  \psi_{n-j}   \rangle_{E} \big] \\
 	= & \det \left( 1 - \chi_E K_{n-j,n} \chi_E \right) \cdot  \big[ 1 -  \langle \tilde{\psi}_{n-j} ,  \psi_{n-j}   \rangle_{E} \\
	& \phantom{\det \left( 1 - \chi_E K_{n-j,n} \right) \cdot } -  \langle \left( 1 - \chi_E K_{n-j,n} \chi_E \right)^{-1} \chi_E K_{n-j,n} \chi_E\tilde{\psi}_{n-j} ,  \psi_{n-j}   \rangle_{E} \big].
\end{split}
\end{equation}
The advantage of using this formula is that $\tilde{\psi}_{n-j}$ appears in the inner product $\langle \tilde{\psi}_{n-j}, \psi_{n-j} \rangle_E$ and the function $(K_{n-j,n}\chi_E \tilde{\psi}_{n-j})(x)$. We will see that we have good estimates on both of  these quantities due to the fast decay of $\psi_{n-j}(y;n)$ and $K_{n-j,n}(x,y)$ as $n\to \infty$ for all $y\in E$. 

We study two kinds of intervals $E=\Int$ and $E=\Intx(x_*)$, $x_*>\redge$. 

\begin{enumerate}[label=(\alph*)]

\item
For $E=\Intx(x_*)$ where $x_*$, which may depend on $n$, is in a compact subset of $(\redge, \infty)$,  
from~\eqref{eq:Fred_det_of_K_n-1,n_2}, 
$ \det \big( 1 - \chi_{\Intx(x_*)} K_{n-j,n} \chi_{\Intx(x_*)} \big)  \to 1$.
For the last inner product in~\eqref{eq:maind}, note that the operator norm of  $( 1 - \chi_{\Intx(x_*)} K_{n-j,n} \chi_{\Intx(x_*)})^{-1}$ is uniformly bounded  from~\eqref{eq:Kinverse2}, and $\psi_{n-j}(x) \to 0$ in $L^2(\Intx(x_*))$ from~\eqref{eq:psiL20}. We will show that the $L^2(\Intx(x_*))$ norm of $K_{n-j,n}\chi_{\Intx(x_*)} \tilde{\psi}_{n-j}$ is uniformly bounded. Hence we will have 
\begin{equation}\label{eq:maind2}
\begin{split}
	 \Prob_{n-j+1,n}(a; \Intx(x_*))
 	= 1 -  \langle \tilde{\psi}_{n-j} ,  \psi_{n-j}   \rangle_{\Intx(x_*)} +o(1) 
\end{split}
\end{equation}
Therefore, we need 
\begin{enumerate}[label=(\roman*)]
\item asymptotic evaluation of $\langle \tilde{\psi}_{n-j} ,  \psi_{n-j}   \rangle_{\Intx(x_*)}$,
\item uniform boundedness of $L^2(\Intx(x_*))$-norm of $K_{n-j,n} \chi_{\Intx(x_*)}\tilde{\psi}_{n-j}$. 
\end{enumerate}

\item
For $E=\Int$, 
from~\eqref{eq:Fred_det_of_K_n-1,n_1}, 
$ \det \big( 1 - \chi_{\Int} K_{n-j,n} \chi_{\Int} \big)  \to \FGUE(T)$. 
For the cases in Section~\ref{subsection:generic_case_1}, we will show that $K_{n-j,n} \chi_{\Int}\tilde{\psi}_{n-j}$ is uniformly bounded in $L^2(\Int)$ but for the cases in Section~\ref{subsection:The_critical_case_rank_1}, we will see that the $L^2(\Int)$-norm of $K_{n-j,n} \chi_{\Int}\tilde{\psi}_{n-j}$ is $O(n^{1/6})$. We here state the necessary estimates separately. 

\begin{enumerate}[label=(\alph{enumi}\arabic*)]
\item \label{emu:sub_method_b1}
For Sections~\ref{subsection:generic_case_1},
we need  
\begin{enumerate}[label=(\roman*)]
\item asymptotic evaluation of $\langle \tilde{\psi}_{n-j} ,  \psi_{n-j}   \rangle_{\Int}$,
\item uniform $L^2(\Int)$ boundedness of $K_{n-j,n} \chi_{\Int}\tilde{\psi}_{n-j}$. 
\end{enumerate}
Then it follows that, since the operator norm of $( 1 - \chi_{\Int} K_{n-j,n} \chi_{\Int})^{-1}$ is uniformly bounded from Corollary~\ref{lemma:trace_norm_convergence_K_n-1,n}
~\ref{enu:lemma:uniformbddinverse_K_n-1,n:c}, and $\psi_{n-j}\to 0$ in $L^2(\Int)$ from~\eqref{eq:L^2_norm_of_psi_1}, that 
\begin{equation}\label{eq:maind1}
	\Prob_{n-j+1,n}(a; \Int) 
= (\FGUE(T)+o(1)) \cdot  \big[ 1 -  \langle \tilde{\psi}_{n-j} ,  \psi_{n-j}   \rangle_{\Int} +o(1) \big].
\end{equation}

\item \label{emu:sub_method_b2}
For Section~\ref{subsection:The_critical_case_rank_1}, 
we need 
\begin{enumerate}[label=(\roman*)]
\item asymptotic evaluation of $\langle \tilde{\psi}_{n-j} ,  \psi_{n-j}   \rangle_{\Int}$,
\item asymptotics evaluation of 
\begin{equation}\label{eq:defu}
	u_{j,n}(\xi): =\frac1{\sqrt{n}} (K_{n-j,n} \chi_{\Int}\tilde{\psi}_{n-j})(\redge+\beta^{-1}n^{-2/3}\xi) 
\end{equation}
in $L^2([T, \infty))$. 
\end{enumerate}
Then since $( 1 - \chi_{[T, \infty)} \K_{n-j,n} \chi_{[T,\infty)})^{-1}\to (1-\chi_{[T, \infty)} K_{\Airy}\chi_{[T, \infty)})^{-1}$ in operator norm from~\eqref{eq:Kinverse3} and $n^{-1/6}\psi_{n-j}(\redge+\beta^{-1}n^{-2/3}\xi) - \B_{j,n}(\redge) \Ai(\xi)\to 0 $ in $L^2([T, \infty))$ by Corollary \ref{cor:only_one_last_sect}\ref{enu:cor:only_one_last_sect:d}, we find that 
if $u_{j,n}- u_n\to 0$ in $L^2([T, \infty))$, then 
\begin{equation}\label{eq:maind3}
\begin{split}
	 \Prob_{n-j+1,n}(a; \Int) 
 	= & (\FGUE(T)+o(1)) \cdot  \big[ 1 -  \langle \tilde{\psi}_{n-j} ,  \psi_{n-j}   \rangle_{\Int} \\
	&-  \frac1{\beta} \B_{j,n}(\redge) \langle   ( 1 - \chi_{[T, \infty)} K_{\Airy} \chi_{[T, \infty)} )^{-1} u_n,   \Ai   \rangle_{[T, \infty)} + o(1) \big].
\end{split}
\end{equation}

%we denote
%\begin{align}
%u_{j,n}(\xi) := & \frac1{\sqrt{n}} (K_{n-j,n} \chi_{\Int}\tilde{\psi}_{n-j})(\redge+\beta^{-1}n^{-2/3}\xi), \label{eq:defu} \\
%v_{j,n}(\xi) := & n^{-1/6}\psi_{n-j}(\redge+\beta^{-1}n^{-2/3}\xi). \label{eq:def_pf_v_jn}
%\end{align}

%Here $1/\beta$ comes from the change of variables $x\mapsto \redge+\beta^{-1}n^{-2/3}\xi$ in the inner product.
\end{enumerate}
\end{enumerate}

%\bigskip

%\Baik{Check again. Are these the only conditions?}
%We apply the Laplace's method and the steepest-desent method repeatedly in the paper to the integral of the form
%\begin{equation}
%	I_n = \int h_n(z) e^{nf_n(z)}dz
%\end{equation}
%as $n\to\infty$. We remind that even if $h_n$ depends on $n$, the standard Laplace's method applies as long as $h_n$, $h^{-1}_n$ and $h'_n$ are bounded in $n$ in any compact subset of the integral domain, and its growth is slow enough as $\lvert z \rvert \to \infty$. For example, if $z_0$ is the unique maximizer or saddle point and $f''_n(z_0)\neq 0$, then up to a sign
%\begin{equation}
%	I_n = h_n(z_0) \sqrt{\frac{2\pi}{- nf''_n(z_0)}}e^{nf_n(z_0)} (1+o(1)).
%\end{equation}
%This formula, though by no means universally true, applies for all calculations in our paper.

\bigskip

In Sections \ref{subsection:generic_case_1}, \ref{subsection:generic_case_3} and \ref{subsection:The_critical_case_rank_1} we only consider $a>0$. When $a < 0$, the largest eigenvalue in the spiked source model defined by \eqref{eq:generalized_pdf} has the same distribution as the negative value of the smallest eigenvalue of the spiked source model that is defined by the same formula but with the potential function $\hat{V}(x) = V(-x)$ and the external source matrix $-\A[n-j+1]$. Since $\hat{V}(x)$ is regular as long as $V(x)$ is, and the non-zero eigenvalue of $−\A[n-j+1]$ is positive, the analysis in this paper applies for that spiked source model. We need to keep track of the smallest eigenvalue in the new spiked source model, and it can be done in the same way that we analyze the largest one. It can be checked that the limiting distribution of the smallest eigenvalue is not affected by the positive external source eigenvalue $a$, corresponding to the $a < 0$ case of Theorem \ref{thm:thm_rank_1}\ref{enu:thm:thm_rank_1:a}. We skip any further remarks.

\section{When $0 < a< \frac12 V'(\redge)$} \label{subsection:generic_case_1} 

As outlined in Section~\ref{subsection:outline_of_the_proof}, we need to show that the $L^2$ norm of $K_{n-j,n} \chi_{\Int}\tilde{\psi}_{n-j}$ is uniformly bounded in $n$, and need to evaluate $\langle \tilde{\psi}_{n-j} ,  \psi_{n-j}   \rangle_{E}$ asymptotically for appropriate choices of the interval $E$. The key part is the asymptotic evaluation of the function $\tilde{\psi}_{n-j}(x)$. 
In Subsection~\ref{sec:GammalessV} we first evaluate $\bfGamma_{n-1}(a)$ asymptotically and then use this in Subsection~\ref{sec:sub21} to evaluate $\tilde{\psi}_{n-j}(x)$. The remaining subsections are devoted to the proof of the main theorems in each sub-case.

\subsection{Asymptotic evaluation of $\bfGamma_{n-j}(a):=\bfGamma_{n-j}(a; n)$}\label{sec:GammalessV}

From the definition~\eqref{eq:new_formula_of_tilde_gamma},  
\begin{equation}\label{eq:eq80}
\begin{split}
	\bfGamma_{n-j}(a) %&:=\bfGamma_{n-1}(a;n) = \int_{-\infty}^\infty  \varphi_{n-1}(y) e^{nay} dy \\
	&= \int_{-\infty}^c \varphi_{n-j}(y) e^{nay} dy + \int_c^\infty \varphi_{n-j}(y) e^{nay} dy
\end{split}
\end{equation}
for any $c\in \R$. We take $c=c(a)$ as in Definition~\ref{def:ca}. Note that $c(a)>\redge$ since $a<\frac12 V'(\redge)$. 
%Recall that $\Hfn(x)=\Hfn(x;a)$ attains its minimum at $x=c$ in $[\redge, \infty)$. 
The reason that we split the integral at $y=c$ will be clear in the below, particularly the paragraph before~\eqref{eq:estimation_of_contour_integral_case_1}. 

Let 
\begin{equation} \label{eq:formula_of_Cauchy_transform}
	(Cf)(z):= \frac1{2\pi i} \int_{\R} \frac{f(y)}{y-z}dy
\end{equation}
denote the Cauchy transform of function $f\in L^2(\R)$. Using $C_+-C_-=1$ and noting that $\varphi_{n-j}$ is analytic, we have 
\begin{equation}
\begin{split}
	 \int_{-\infty}^c \varphi_{n-j}(y) e^{nay} dy 
	 = \int_{-\infty}^c \big( (C_+\varphi_{n-j})(y) - (C_-\varphi_{n-j})(y)\big)  e^{nay} dy.  
\end{split}
\end{equation}
Note that $(C\varphi_{n-j})(z)e^{naz} \to 0$ exponentially as $\Re(z)\to -\infty$ since $a>0$. Therefore, we can deform the contour and obtain 
\begin{equation}
	\int^c_{-\infty} \varphi_{n-j}(y) e^{nay} dy =  - \int_{\Gamma_+\cup\Gamma_-} (C\varphi_{n-j})(z) e^{naz} dz,
\end{equation}
where, %the contour $\Gamma$ is defined by 
with a constant $C_{\Gamma}>\frac1{2a}$,
\begin{equation} \label{eq:definition_of_contour_Gamma}
\begin{split}
	\Gamma_+ := & \{ c+ it \mid 0<t<C_\Gamma \}\cup \{ c+iC_{\Gamma} - t \mid t \ge 0 \} , \\
	\Gamma_- := & \textnormal{ complex conjugation of }{\Gamma}_+.
%\Gamma_0 = & \{ c+ it \mid 0 < \lvert t \rvert \leq C_{\Gamma} \}.
\end{split}
\end{equation}
The contours are oriented as  indicated in Figure~\ref{fig:Gamma_complex}.
Therefore we find 
\begin{equation} \label{eq:division_of_psi_n_exponent}
\bfGamma_{n-j}(a) =  - \int_{\Gamma_+\cup\Gamma_-} (C\varphi_{n-j})(z)e^{naz} dz + \int^{\infty}_c \varphi_{n-j}(y)e^{nay} dy.
\end{equation}
\begin{figure}[htp]
\centering
\includegraphics{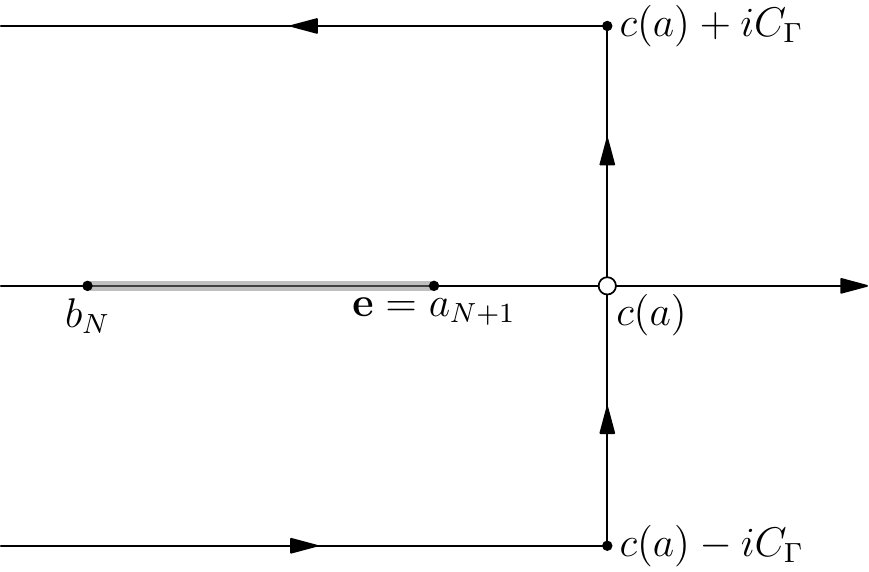}
\caption{The contours $\Gamma_+$ and $\Gamma_-$.} \label{fig:Gamma_complex}
\end{figure}
%We note that $(C\varphi_{n-j})(z)$ is discontinuous at $z=c$.  
%on $\Gamma$ is not a connected contour since $c \notin \Gamma$, on which $C\varphi_j(z)$ is not well-defined.

\bigskip

Let $\delta$ be given in Proposition~\ref{prop:asy_for_mult_cut}.
Let $\epsilon<\min\{c(a)-\redge, 2\delta\}$ be a small enough positive constant, independent of $n$, such that all maximizers of $\Gfn(x;a)$ in $[c, \infty)$ are in the interval $(\redge+\epsilon, \infty)$.
Recall the asymptotics of $(C\varphi_{n-j})(z)$ summarized in Section~\ref{section:Result_of_RHP}. Since the contours $\Gamma_{\pm}$ lie in $B_\delta$ (in Figure~\ref{fig:AB_delta}) in Section~\ref{section:Result_of_RHP}, from the asymptotic formula~\eqref{eq:asy_of_Cvarphi_n-j} for $(C\varphi_{n-j})(z)$, 
\begin{equation}\label{eq:CphyMremark}
\int_{\Gamma_+\cup\Gamma_-} (C\varphi_{n-j})(z)e^{naz} dz  = \int_{\Gamma_+\cup\Gamma_-} \tilde{M}_{j,n} (z) e^{n(\Hfn(z;a) - \ell/2)}  dz.
\end{equation}
Here we recalled the definition~\eqref{eq:definition_of_GH}, $\Hfn(z;a):=-\gfn(z)+az+\ell$.

%Note that even though $(C\varphi_{n-1})(z)$ is not defined at $z=c$, the integrand of the last integral is now analytic on $\Gamma$. Indeed, it is analytic in $\mathbb{C}$ minus an arbitrarily small neighborhood of $(-\infty, \redge]$. 

We now use the method of steepest-descent to evaluate the integral asymptotically. By Lemma~\ref{fact:first}, $\Hfn'(c(a);a)=0$ and $\Hfn''(c(a);a)> 0$. 
It is straightforward to check, with the help of the formula of $\gfn(x)$ in \eqref{eq:definition_of_g}, that for $z(t) = c+it$, $t>0$,  the function $\Re\Hfn(z(t);a)$ in $t$ satisfies
\begin{equation} \label{eq:decreasing_counterclockwise_1}
\frac{d}{dt} \Re\Hfn(z(t);a) = -\int \frac{t}{(c-s)^2+t} \Psi(s) ds < 0.
\end{equation}
Also 
for $z(t) = c+iC_{\Gamma}-t \in \Gamma_+$, $t\ge 0$, 
\begin{equation} \label{eq:decreasing_counterclockwise_2}
\frac{d}{dt} \Re\Hfn(z(t);a) = -\int \frac{t-c+s}{(t-c+s)^2+C^2_{\Gamma}} \Psi(s) ds - a,
\end{equation}
is negative for all $t\ge 0$ if $C_{\Gamma}>1/(2a)$. Hence $\Re\Hfn(z;a)$ decreases as $z$ moves along $\Gamma_+$ counterclockwise. Similarly $\Re\Hfn(z;a)$ increases as $z$ moves along $\Gamma_-$ counterclockwise. Therefore $\overline{\Gamma_+\cup\Gamma_-}$ is a curve of steep-descent for $\Hfn$ with the saddle point at $z=c$. The fact that $z=c$ is a saddle point of $\Hfn$ is the reason that we have split the integral in~\eqref{eq:eq80} at $c$.

From Proposition~\ref{prop:asy_for_mult_cut}\ref{enu:prop:asy_for_mult_cut:b}, $\tilde{M}_{j,n}=\tilde{\M}_{j,n}(z)(1+O(n^{-1}))$ and $\tilde{\M}_{j,n}(z)$ is analytic in $B_{\delta}$. Moreover,  $\tilde{M}_{j,n}$, $\tilde{M}'_{j,n}$ and $1/\tilde{M}_{j,n}$ are uniformly $O(1)$ in a neighborhood of $c(a)$ and $\tilde{M}_{j,n}=O(z^{j-1})$ uniformly in $n$ as $z\to\infty$. Thus the method of steepest-descent can be applied to~\eqref{eq:CphyMremark}  and we obtain %and with the help of \eqref{eq:asy_of_tilde_M_j_mult}, the result is
\begin{equation} \label{eq:estimation_of_contour_integral_case_1}
\begin{split}
	 \int_{\Gamma_+\cup\Gamma_-} (C\varphi_{n-j})(z)e^{naz} dz  
= & \frac{i\sqrt{2\pi} \tilde{\M}_{j,n}(c) e^{n(\Hfn(c(a);a)-\ell/2)}}{\sqrt{n\Hfn''(c(a);a)}} (1+o(1)).
%= & \frac{i\sqrt{2\pi} \tilde{\M}_{1,0}(c) e^{n(\Hfn(c;a)-\ell/2)}}{\sqrt{-\gfn''(c)n}} (1+o(1)).
\end{split}
\end{equation}
%We note that $\tilde{\M}_{1,0}(c)\geq m(\epsilon)$ by~\eqref{eq:sign_of_Mand_tilde_M} below. 

%\bigskip

Now consider the second integral in~\eqref{eq:division_of_psi_n_exponent}. From the asymptotics~\eqref{eq:asy_of_varphi_n-j} for $\varphi_{n-j}$ in $B_{\delta}$,  
\begin{equation}\label{eq:GtoMthen0}
	\int^{\infty}_c \varphi_{n-j}(y)e^{nay} dy 
	= \int^{\infty}_c M_{j,n}(y) e^{n(\Gfn(y;a)-\ell/2)}dy
\end{equation}
Using that $M_{j,n}$ is uniformly bounded in any compact subset of $[\redge+\epsilon, \infty)$, and $M_{j,n}(y) = O(y^{-j})$ as $y \to \infty$, and using that $\Gfn(y;a) \to -\infty$ at least linearly, we obtain the trivial estimate  that
\begin{equation}\label{eq:GtoMthen}
	\int^{\infty}_c \varphi_{n-j}(y)e^{nay} dy = O(e^{n(\Gfn_{\max}(a)-\ell/2)})
\end{equation}
where $\Gfn_{\max}(a):= \max\{ \Gfn(y;a) \mid y\ge c(a)\}$. Together with~\eqref{eq:estimation_of_contour_integral_case_1}, we obtain the following result. 
%We state the evaluation of $\bfGamma_{n-1}(a)$ in detail when $a \not\in \mathcal{J}_V$, \ie, we need only consider one maximizer $x_0 = x_0(a)$. 
Recall the properties of $\Gfn$ and $\Hfn$ in Subsection~\ref{subsection:critical_values}.

\begin{itemize}
\item
Suppose that $a<\acc$. Then $\Hfn(c;a)> \Gfn_{\max}(a)$. Therefore,~\eqref{eq:estimation_of_contour_integral_case_1} is exponentially larger than~\eqref{eq:GtoMthen} and  we obtain
\begin{equation} \label{eq:tilde_gamma_case_1_sub}
	\bfGamma_{n-j}(a) = \sqrt{\frac{2\pi}{n}}e^{-n\ell/2} \frac{-i \tilde{\M}_{j,n}(c(a))}{\sqrt{\Hfn''(c(a);a)}} e^{n\Hfn(c(a);a)} (1+o(1)).
\end{equation}

\item 
Suppose that $\acc< a< \frac12V'(\redge)$ (assuming that $V$ is such that $\acc<\frac12 V'(\redge)$). Then $\Gfn_{\max}(a)> \Hfn(c;a)$ and hence~\eqref{eq:GtoMthen} is exponentially larger than~\eqref{eq:estimation_of_contour_integral_case_1}. %In this case, the maximizers of $\Gfn$ are in the open interval $(c,\infty)$. 
Suppose that $a\notin\mathcal{J}_V$ and let $x_0=x_0(a)\in (c(a),\infty)$ be the unique point $\Gfn_{\max}(a)$ is attained. 
If $\Gfn''(x_0;a)\neq 0$, using the Laplace's method applied to~\eqref{eq:GtoMthen0} (using the properties of $M_{j,n}$ in Proposition~\ref{prop:asy_for_mult_cut} \ref{enu:prop:asy_for_mult_cut:a}), we obtain
\begin{equation} \label{eq:tilde_gamma_case_1_super}
\begin{split}
	\bfGamma_{n-j}(a) 
	&= \sqrt{\frac{2\pi}{-n \Gfn''(x_0;a)}} \M_{j,n}(x_0) e^{n\Gfn(x_0;a)-n\ell/2} (1+o(1)).
%	&= \sqrt{\frac{2\pi}{n}}e^{-n\ell/2} \frac{\M_{1,0}(x_0)}{\sqrt{V''(x_0)-\gfn''(x_0)}} e^{n\Gfn(x_0;a)} (1+o(1)).
\end{split}
\end{equation}
If $\Gfn''(x_0;a)=0$ and $\Gfn^{(2k)}(x_0;a)\neq 0$, $\Gfn^{(j)}(x_0;a) = 0$ for $j=1, \cdots, 2k-1$, then the Laplace's method implies that 
\begin{equation} \label{eq:tilde_gamma_case_1_super11}
	\bfGamma_{n-j}(a) = \bigg(\frac{(2k)!}{-n\Gfn^{(2k)}(x_0;a)}\bigg)^{1/(2k)} \int_{-\infty}^\infty e^{-x^{2k}}dx %\\
%\times 
\M_{j,n}(x_0) e^{n\Gfn(x_0;a)-n\ell/2} (1+o(1)).
\end{equation}
When $a\in \mathcal{J}_V$, the contributions to~\eqref{eq:GtoMthen} at each maximizer should be added. Examples of this case are in~\eqref{eq:bla} and~\eqref{eq:calculation_of_Gamma_thm_7}.

\item 
If $a=\acc$ (note that since we assumed $a<\frac12 V'(\redge)$, this implies that $V$ is such that $\acc<\frac12 V'(\redge)$), then $\Hfn(c;\acc)=\Gfn_{\max}(\acc)$. If we further assume that $\acc\notin \mathcal{J}_V$ and $\Gfn''(x_0(\acc); \acc)\neq 0$, \eqref{eq:estimation_of_contour_integral_case_1} and~\eqref{eq:GtoMthen} are of same order. Then by Laplace's method applied to \eqref{eq:GtoMthen},
\begin{equation} %\label{eq:tilde_gamma_case_1_critical_II}
	\bfGamma_{n-j}(\acc) = \sqrt{\frac{2\pi}{n}}e^{-n\ell/2} \left( \frac{\M_{j,n}(x_0(\acc))}{\sqrt{-\Gfn''(x_0(\acc))}} + \frac{-i \tilde{\M}_{j,n}(c(\acc))}{\sqrt{\Hfn''(c(\acc))}}   \right) %\\
%	\times 
e^{n\Gfn(x_0(\acc);\acc)} (1+ o(1)).
\end{equation}

We can consider a double scaling case when 
\begin{equation} \label{eq:a_is_around_acc<halfV'}
	a=\acc+ \frac{\alpha}{n}
\end{equation}
where $\alpha$ is in a compact subset of $\R$.  By the definition of $c$ and $x_0$, a direct computation shows that
\begin{align}
	\Gfn(x_0(a);a) = & \Gfn(x_0(\acc);\acc) + \frac{\alpha x_0(\acc)}{n} + O(n^{-2}), \label{eq:fluctuation_of_G(x_0,a)} \\
\Hfn(c(a);a) = & \Hfn(c(\acc);\acc) + \frac{\alpha c(\acc)}{n} + O(n^{-2}).
\end{align}
This implies that 
\begin{multline} \label{eq:tilde_gamma_case_1_critical_II}
	\bfGamma_{n-j}(a) = \\ \sqrt{\frac{2\pi}{n}}e^{-n\ell/2} \left( \frac{\M_{j,n}(x_0(\acc))}{\sqrt{-\Gfn''(x_0(\acc);\acc)}} \right. %\\
+ \left. \frac{-i \tilde{\M}_{j,n}(c(\acc))}{\sqrt{\Hfn''(c(\acc);\acc)}} e^{\alpha(c(\acc) - x_0(\acc))} \right) e^{n\Gfn(x_0(a);a)} (1 + o(1)).
\end{multline}
If $\Gfn''(x_0(\acc))=0$, then the term~\eqref{eq:GtoMthen} is greater than~\eqref{eq:estimation_of_contour_integral_case_1} by a fractional power of $n$ (see~\eqref{eq:tilde_gamma_case_1_super11}), and hence the term involving $\tilde{\M}_{j,n}(c(\acc))$ disappears in~the expression of $\bfGamma_{n-j}(\acc)$. On the other hand, if $\acc\in \mathcal{J}_V$, then there are more than one maximizers of $\Gfn(y;\acc)$ making contributions in~\eqref{eq:GtoMthen}. We do not state the formulas explicitly here but instead state them in the appropriate subsections where they arise. 

%See \eqref{eq:bla} and \eqref{eq:calculation_of_Gamma_thm_7} as examples. \Baik{Add the section for Theorem 1.4}
\end{itemize}

%We note that $\bfGamma_{n-1}(a) > 0$ from the positive condition \eqref{eq:sign_of_Mand_tilde_M}.

%\begin{rmk}
%The calculation of a more subtle case of $a \not\in \mathcal{J}_V$ is given in Subsection \ref{subsection:proof_of_thm:critical_traditional_split}. Later we calculate two typical cases of $a \in \mathcal{J}_V$ in \eqref{eq:bla} and \eqref{eq:calculation_of_Gamma_thm_7}. 
%\end{rmk}

\subsection{Asymptotic evaluation of $\tilde{\psi}_{n-j}(x):=\tilde{\psi}_{n-j}(x;a)$}\label{sec:sub21}

\subsubsection{Algebraic formula}

From~\eqref{eq:new_formula_of_tilde_psi}, 
\begin{equation}\label{eq:bfGamma_temp1}
	\bfGamma_{n-j}(a)\tilde{\psi}_{n-j}(x)
	=  e^{n(ax-V(x)/2)} - \sum^{n-j-1}_{i=0}\psi_i(x) 
	\int_{-\infty}^\infty \varphi_i(y)e^{nay} dy. 
\end{equation}
This can be written in the following way. Let $c$ be any constant such that $c>\redge$. 
We will take $c=c(a)$ as in Definition~\ref{def:ca} in the subsequent sections for asymptotic analysis, but the following result holds for any $c>\redge$. 

\begin{lemma}\label{lem:tpsinew}
For $0<a<\frac12 V'(\redge)$, we have, with $\Gamma_{\pm}$ given in~\eqref{eq:definition_of_contour_Gamma},
for $x \in \R\setminus\{ c\}$, 
\begin{equation} \label{eq:formula_of_tilde_psi_less_than_c}
\begin{split}
	&\bfGamma_{n-j}(a) \tilde{\psi}_{n-j}(x) = e^{n( ax - V(x)/2)} 1_{(c, \infty)}(x) \\
&\quad + \int_{\Gamma_+\cup\Gamma_-} \CK_{n-j,n}(x,z) e^{naz} dz  - \int^{\infty}_{c} K_{n-j,n}(x,y)  e^{n(ay-V(y)/2)} dy \\
% &\qquad \qquad\quad  - \int^{\infty}_{c} \frac{\psi_{n-1}(x)\varphi_{n-2}(y) - \psi_{n-2}(x)\varphi_{n-1}(y)}{x-y} e^{nay} dy  \bigg]
\end{split}
\end{equation}
where
\begin{equation} \label{eq:CK}
\begin{split}
	&\CK_{k,n}(x,z):= \frac{\gamma_{k-1}}{\gamma_{k}} \frac{\psi_{k}(x)(C\varphi_{k-1})(z) - \psi_{k-1}(x)(C\varphi_{k})(z)}{x-z} , \quad x\neq z
\end{split}
\end{equation}
%The formula has the same limit as $x\downarrow c$ and as $x\uparrow c$, and the limit equals $\bfGamma_{n-1}(a)\tilde{\psi}_{n-1}(c)$.
and $K_{k,n}(x,y)$ is defined in~\eqref{eq:Christoffel_Darboux_K_n-1,n}.
\end{lemma}

\begin{proof}
By the same calculation that leads to~\eqref{eq:division_of_psi_n_exponent}, ~\eqref{eq:bfGamma_temp1} equals 
\begin{equation} \label{eq:first_integral_formula_of_gamma_psi}
\begin{split}
	&\bfGamma_{n-j}(a;n)\tilde{\psi}_{n-j}(x)
	= e^{n(ax-V(x)/2)}\\
	&+ \sum^{n-j-1}_{i=0} \psi_i(x) \int_{\Gamma_+\cup\Gamma_-}  (C\varphi_i) (z)e^{naz} dz 
	 - \sum^{n-j-1}_{i=0} \psi_i(x) \int^{\infty}_c \varphi_i(y)e^{nay} dy.
\end{split}
\end{equation}
We exchange the sum and the integral in both terms. 
The second sum can be simplified by using the Christoffel-Darboux formula % for orthogonal polynomials and 
and becomes the last integral in~\eqref{eq:formula_of_tilde_psi_less_than_c}. 
To analyze the first sum, we first take $x \in \compC \setminus (\realR \cup \Gamma_+ \cup \Gamma_-)$. By using the definition of Cauchy operator and from the Christoffel-Darboux formula, we have
\begin{align}
	&\sum^{n-j-1}_{i=0} \psi_i(x) (C\varphi_i)(z) 
	 =  \frac{1}{2\pi i} \int_{\realR} \frac{1}{w-z}  \sum^{n-j-1}_{i=0}  \psi_i(x)\varphi_i(w)  dw \notag \\
	= & \frac{\gamma_{n-j-1}}{\gamma_{n-j}} \frac{1}{2\pi i} \int_{\realR} \frac{\psi_{n-j}(x)\varphi_{n-j-1}(w) - \psi_{n-j-1}(x)\varphi_{n-j}(w)}{(w-z)(x-w)} dw. \notag 
\end{align}
%\intertext{(using the partial fraction formula and the definition of the Cauchy transformation again)}
Using the partial fraction formula and the definition of the Cauchy transformation again, this equals 
\begin{align}
%& \sum^{n-j-1}_{i=0} \psi_i(x) (C\varphi_i)(z) \\
& \frac{\gamma_{n-j-1}}{\gamma_{n-j}} \frac{1}{x-z}  \frac{1}{2\pi i} \bigg[ \int_{\realR} \frac{\psi_{n-j}(x)\varphi_{n-j-1}(w) - \psi_{n-j-1}(x)\varphi_{n-j}(w)}{w-z} dw \notag \\
& \phantom{\frac{1}{x-z} \frac{\gamma_{n-j-1}}{\gamma_{n-j}} \frac{1}{2\pi i}} -  \int_{\realR} \frac{\psi_{n-j}(x)\varphi_{n-j-1}(w) - \psi_{n-j-1}(x)\varphi_{n-j}(w)}{w-x} dw \bigg] \notag \\
= & \frac{\gamma_{n-j-1}}{\gamma_{n-j}} \bigg[ \frac{\psi_{n-j}(x)(C\varphi_{n-j-1})(z)  - \psi_{n-j-1}(x)(C\varphi_{n-j})(z)}{x-z}  \notag \\
& \phantom{\frac{1}{x-z} \frac{\gamma_{n-j-1}}{\gamma_{n-j}} \frac{1}{2\pi i}} - \frac{\psi_{n-j}(x)(C\varphi_{n-j-1})(x) - \psi_{n-j-1}(x)(C\varphi_{n-j})(x)}{x-z} \bigg] \notag \\
= & \frac{\gamma_{n-j-1}}{\gamma_{n-j}} \bigg[ \frac{\psi_{n-j}(x)(C\varphi_{n-j-1})(z)  - \psi_{n-j-1}(x)(C\varphi_{n-j})(z)}{x-z} \bigg] \notag \\
& %\phantom{\frac{1}{x-z} \frac{\gamma_{n-j-1}}{\gamma_{n-j}} \frac{1}{2\pi i}} 
+ \frac1{2\pi i}\frac{e^{-nV(x)/2}}{x-z}  , \label{eq:Christoffel_Darboux_Cauchy}
\end{align}
where the identity~\eqref{eq:algebraic_property_of_RHP} is used in the last line.
%\begin{equation} \label{eq:identity_of_Cauchy_transiform}
%\psi_{k}(x) (C\varphi_{k-1})(x) - \psi_{k-1}(x) (C\varphi_{k})(x) = \frac{-1}{2\pi i} \frac{\gamma_{k}}{\gamma_{k-1}} e^{-nV(x)/2}
%\end{equation}
%for all $x\in \mathbb{C} \setminus \realR$ up to the boundary, for the last term in~\eqref{eq:Christoffel_Darboux_Cauchy}. This formula simplifies the bottom line of \eqref{eq:Christoffel_Darboux_Cauchy}. 
Hence the first sum on the right-hand-side of~\eqref{eq:first_integral_formula_of_gamma_psi} satisfies 
\begin{equation}
\begin{split}
	& \sum^{n-j-1}_{i=0} \psi_i(x) \int_{\Gamma_+\cup\Gamma_-}  (C\varphi_i) (z)e^{naz} dz  \\
	= & \frac{\gamma_{n-j-1}}{\gamma_{n-j}}   \int_{\Gamma_+\cup\Gamma_-} 
	\frac{\psi_{n-j}(x)(C\varphi_{n-j-1})(z)  - \psi_{n-j-1}(x)(C\varphi_{n-j})(z)}{x-z} e^{naz}dz  \\
	& + \frac{e^{-nV(x)/2} } {2\pi i} \int_{\Gamma_+\cup\Gamma_-}  \frac{e^{naz}}{x-z} dz.
\end{split}
\end{equation}
Note that this was proven for $x\in \compC\setminus(\R\cup\Gamma_+\cup\Gamma_-)$, but the identity holds for $x\in \R\setminus\{c\}$ as well by analytic continuation. 
The last integral equals $-2\pi i e^{nax}$ for $x\in (-\infty, c)$ and equals $0$ for $x\in (c,\infty)$ by Cauchy integral formula. Therefore we obtain~\eqref{eq:formula_of_tilde_psi_less_than_c}.
% for $x \in \compC \setminus (\realR \cup \Gamma_+ \cup \Gamma_-)$. The extension to $x \in \realR \setminus \{ c \}$ is obvious.
%The statement about the limits $x\downarrow c$ and $x\uparrow c$ follows by noting that the right hand side of \eqref{eq:formula_of_tilde_psi_less_than_c} is continuous at $x=c$. 
\end{proof}

\subsubsection{For $x\ge \redge+\epsilon$:}

%We now consider the asymptotics of $\tilde{\psi}_{n-1}(x)$ as $n\to\infty$. 

Take $c=c(a)$ as in Definition~\ref{def:ca} in the formula of Lemma~\ref{lem:tpsinew}.
Fix $\epsilon>0$ small enough so that $[c, \infty)\subset [\redge+\epsilon, \infty)$. 

%We have the following Lemma.

\begin{lemma} \label{eq:approx_of_tilde_psi_far}
%Let $a\in (0, \frac12 V'(\redge))$. 
%Fix $\epsilon>0$. Then 
For $x \in [\redge+\epsilon, \infty)$, % \setminus \{ c \}$,
\begin{equation} \label{eq:tilde_psi_n-1_less_than_c_1}
\tilde{\psi}_{n-j}(x) = e^{n (\Gfn(x)-\Hfn(x))/2}  \bigg\{ \frac1{e^{n\ell/2} \bfGamma_{n-j}(a)}e^{n\Hfn(x)} 1_{(c,\infty)}(x)  + O(\sqrt{n}(1 + \lvert x \rvert)^{-j}) \bigg\}
\end{equation}
as $n\to\infty$ and $j=O(1)$.
\end{lemma}

\begin{proof}
Fix $\epsilon'\in (0, c-\redge-\epsilon)$.
Assume that $x$ satisfies $|x-c|\ge \epsilon'$. Noting $ax-V(x)/2 = (\Gfn(x)+ \Hfn(x)-\ell)/2$, by Lemma \ref{lem:tpsinew} we have
\begin{equation} \label{eq:Gammapsiall}
	\tilde{\psi}_{n-j}(x) =  e^{n (\Gfn(x)-\Hfn(x))/2} 
	\left\{\frac1{e^{n\ell/2} \bfGamma_{n-j}(a)}e^{n\Hfn(x)} 1_{(c,\infty)}(x) + Q_{j,n}(x) \right\},
\end{equation}
where
\begin{multline} \label{eq:expression_of_I_n(xa)}
	Q_{j,n}(x) = \frac{1}{\bfGamma_{n-j}(a)} \left[ \int_{\Gamma_+ \cup \Gamma_-} e^{n(\Hfn(x)-\Gfn(x))/2}\CK_{n-j, n}(x,z)e^{naz} dz \right. \\
 - \left. \int^{\infty}_c e^{n(\Hfn(x)-\Gfn(x))/2}K_{n-j,n}(x,y) e^{n(ay - V(y)/2)} dy \right].
\end{multline}
From~\eqref{eq:lemma_enu:estimate_K_n-1,n:1},  the second integral over $(c,\infty)$ in~\eqref{eq:expression_of_I_n(xa)} is
\begin{equation} \label{eq:estimate_of_Q_2}
 O\bigg( (1 + \lvert x \rvert)^{-j}\int_c^\infty e^{n(\Gfn(y)-\ell/2)} (1 + \lvert y \rvert)^{-j} dy\bigg).
\end{equation}
On the other hand, for the integral over $\Gamma_+\cup\Gamma_-$, \eqref{eq:asy_of_varphi_n-j} and~\eqref{eq:asy_of_Cvarphi_n-j} imply that 
\begin{equation}
\CK_{n-j, n}(x,z)e^{naz} = \frac{(M_{j,n}(x)\tilde{M}_{j+1,n}(z) - M_{j+1,n}(x)\tilde{M}_{j,n}(z))}{x-z}  e^{n(\Gfn(x)-\Hfn(x))/2}  e^{n(\Hfn(z)-\ell/2)}.
\end{equation}
By Proposition \ref{prop:asy_for_mult_cut}\ref{enu:prop:asy_for_mult_cut:a}\ref{enu:prop:asy_for_mult_cut:b}, 
\begin{equation}
\frac{(M_{j,n}(x)\tilde{M}_{j+1,n}(z) - M_{j+1,n}(x)\tilde{M}_{j,n}(z))}{x-z} = O \left( \frac{(1 + \lvert x \rvert)^{-j}(1 + \lvert z \rvert)^j}{\lvert x-z \rvert} \right).
\end{equation}
Thus using the fact that $c$ is the saddle point of $\Hfn(z)$, we have that the first integral over $\Gamma_+ \cup \Gamma_-$  in~\eqref{eq:expression_of_I_n(xa)} is
\begin{equation} \label{eq:estimate_of_Q_1001}
O \left( (1+ \lvert x \rvert)^{-j} \int_{\Gamma_+ \cup \Gamma_-} e^{n(\Hfn(z)-\ell/2)} \frac{(1 + \lvert z \rvert)^j}{\lvert x-z \rvert} dz \right) = 
O \left( (1 + \lvert x \rvert)^{-j}  n^{-1/2} e^{n\Hfn(c)} \right)
\end{equation}
On the other hand, for $\Gamma_{n-j}(a)$, we have from \eqref{eq:division_of_psi_n_exponent}, \eqref{eq:CphyMremark},~\eqref{eq:estimation_of_contour_integral_case_1} and \eqref{eq:GtoMthen0} that 
%in Subsection~\ref{sec:GammalessV} and asymptotics of $M_{j,n}$ and $\tilde{M}_{j,n}$ in Proposition \ref{prop:asy_for_mult_cut}\ref{enu:prop:asy_for_mult_cut:a}\ref{enu:prop:asy_for_mult_cut:b}, we obtain
\begin{equation}\label{eq:bfGammainother}
\bfGamma_{n-j}(a) = O \left( n^{-1/2} e^{n\Hfn(c)} \right) + O \left( \int_c^\infty e^{n(\Gfn(y)-\ell/2)} (1 + \lvert y \rvert)^{-j} dy \right).
\end{equation}
Note that we also have a matching lower bound.
Comparing the estimate \eqref{eq:bfGammainother} of $\bfGamma_{n-j}(a)$ and two estimates \eqref{eq:estimate_of_Q_2} and \eqref{eq:estimate_of_Q_1001}, we find that  $Q_{j,n}(x) = O((1 + \lvert x \rvert)^{-j})$ uniformly for $x\ge \redge+\epsilon$ if $\lvert x-c \rvert \geq \epsilon'$ for a positive constant $\epsilon'$. Note that in this case the error term in~\eqref{eq:tilde_psi_n-1_less_than_c_1} does not contain $\sqrt{n}$.

Now let $x$ satisfy $\lvert x-c \rvert < \epsilon'$. In this case, we start with the formula~\eqref{eq:formula_of_tilde_psi_less_than_c} with a different choice of $c$. We replace $c$ by $c\pm n^{-1/2}$ and let $(\Gamma_+)^{\pm} \cup (\Gamma_-)^{\pm}$ be a contour  deformed from $\Gamma_+ \cup \Gamma_-$ by a semicircle of radius $n^{-1/2}$ to the right/left, respectively, as illustrated in Figures \ref{figure:x-c_in_epsilon'} and \ref{x-c_in_-epsilon'}. Here we take  the $+$ sign if $x-c \geq 0$ and take the $-$ sign if $x-c < 0$.
Then 
\begin{equation}
	\tilde{\psi}_{n-j}(x) =  e^{n (\Gfn(x)-\Hfn(x))/2} 
	\left\{\frac1{e^{n\ell/2} \bfGamma_{n-j}(a)}e^{n\Hfn(x)} 1_{(c \pm n^{-1/2},\infty)}(x) + Q^{\pm}_{j,n}(x) \right\},
\end{equation}
and
\begin{multline} \label{eq:second_Q}
Q^{\pm}_{j,n}(x) = \frac{1}{\bfGamma_{n-j}(a)}\bigg[  \int_{(\Gamma_+)^{\pm}\cup(\Gamma_-)^{\pm}} e^{n(\Hfn(x)-\Gfn(x))}\CK_{n-j, n}(x,z)e^{naz} dz  \\
-  \int_{c \mp n^{-1/2}}^\infty e^{n(\Hfn(x)-\Gfn(x))}K_{n-j,n}(x,y) e^{n(ay - V(y)/2)}  dy \bigg].
\end{multline}
The second integral has the same estimation as that in \eqref{eq:expression_of_I_n(xa)}. 
For the first integral, note that $|x-z|\ge n^{-1/2}$. Using this, and by recalling the asymptotics of a Cauchy-type integral $\int_{(\Gamma_+)^{\pm}\cup(\Gamma_-)^{\pm}} \frac{1}{x-z}  e^{n\Hfn(z)}dz = O(e^{n\Hfn(c)})$ for such $x$, we find that the first integral of~\eqref{eq:second_Q} is $O \left( (1 + \lvert x \rvert)^{-j} e^{n\Hfn(c)} \right)$
instead of $O \left( (1 + \lvert x \rvert)^{-j} n^{-1/2} e^{n\Hfn(c)} \right)$ in the case when $|x-c|\ge \epsilon'$.
Hence we obtain~\eqref{eq:tilde_psi_n-1_less_than_c_1} for $\lvert x - c \rvert < \epsilon$ by  noting that $e^{n\Hfn(c)}/(e^{n\ell} \bfGamma_{n-j}) = O(\sqrt{n})$ for $|x-c|\le n^{-1/2}$  from~\eqref{eq:bfGammainother}.
\begin{figure}[htp]
\begin{minipage}[t]{0.35\linewidth}
\centering
\includegraphics{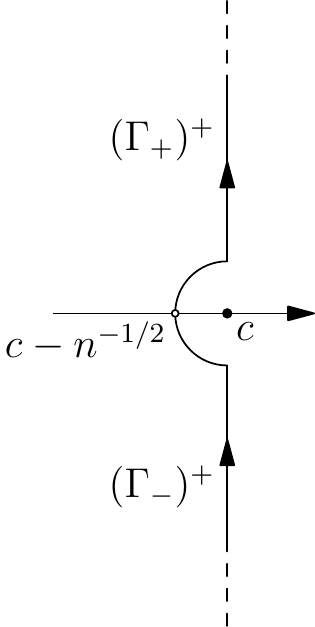}
\caption{The deformed $(\Gamma_+)^+ \cup (\Gamma_-)^+$ for $x-c \in [0,\epsilon')$.}
\label{figure:x-c_in_epsilon'}
\end{minipage}
\begin{minipage}[t]{0.2\linewidth}
\
\end{minipage}
\begin{minipage}[t]{0.35\linewidth}
\centering
\includegraphics{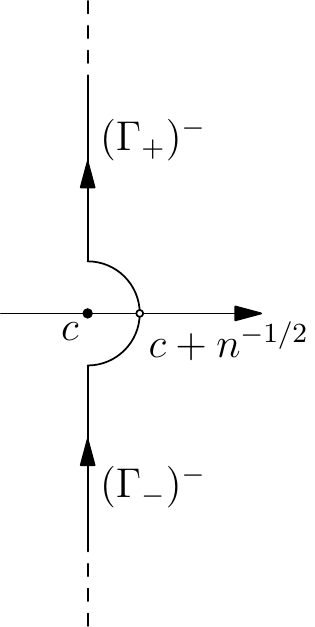}
\caption{The deformed $(\Gamma_+)^- \cup (\Gamma_-)^-$ for $x-c \in (-\epsilon',0)$.}
\label{x-c_in_-epsilon'}
\end{minipage}
\end{figure}
\end{proof}

\subsubsection{For $x$ near $\redge$:}

Let $T$ be a fixed constant and let $\epsilon$ be a small positive constant such that $0 < \epsilon< \min \{c-\redge, \delta_0\}$ where $\delta_0$ is the constant in Proposition \ref{prop:asy_for_mult_cut} and its corollaries in Section~\ref{section:Result_of_RHP}. Define the interval
\begin{equation} \label{eq:defn_of_E_T_epsilon}
E_{T,\epsilon} := \Int \setminus (\redge+\epsilon, \infty) = [\redge + \beta^{-1}n^{-2/3}T, \redge+\epsilon].
\end{equation}
For a given $x \in E_{T,\epsilon}$, define $\xi$ by the relation
\begin{equation} \label{eq:scaling_of_x_around_a_N+1}
	x := \redge+\beta^{-1}n^{-2/3}\xi.
\end{equation} 

\begin{lemma}\label{lem:tilpsiedge}
We have for all $0 < a < V'(\redge)/2$,
\begin{equation} \label{eq:estimation_of_tilde_psi_near}
	\tilde{\psi}_{n-j}(x) = O(n^{1/6} e^{-\factor \lvert \xi \rvert^{3/2}}),
\end{equation}
uniformly in $x \in E_{T,\epsilon}$  and in $n$. 
\end{lemma}

\begin{proof}
We use the formula~\eqref{eq:formula_of_tilde_psi_less_than_c}.
From~\eqref{eq:lemma_enu:estimate_K_n-1,n:4} and \eqref{eq:psin101}, the integral over $(c,\infty)$ in~\eqref{eq:formula_of_tilde_psi_less_than_c} is
\begin{equation} \label{eq:tilde_psi_n-1_linear_part_around_a_N+1}
	O\bigg( n^{1/6} e^{-\factor|\xi|^{3/2}} \int_c^\infty e^{n(\Gfn(y)-\ell/2)} (1 + \lvert y \rvert)^{-j} dy \bigg).
\end{equation}
On the other hand, substituting \eqref{eq:psin102} and \eqref{eq:asy_of_Cvarphi_n-j} into \eqref{eq:CK}, the integrand of the first integral over $\Gamma_+ \cup \Gamma_-$ in~\eqref{eq:formula_of_tilde_psi_less_than_c} is
\begin{equation} \label{eq:tilde_psi_n-1_contour_part_around_a_N+1}
	\CK_{n-j,n}(x,z) e^{naz}= O\big( n^{1/6} e^{-\frac23|\xi|^{3/2}}  (1 + |z|)^j e^{n(\Hfn(z)-\ell/2)} \big),
\end{equation}
for all  $x \in E_{T,\epsilon}$ and $z \in \Gamma_{\pm}$ since $|x-z|>c-\redge-\epsilon>0$. Thus %analogous to \eqref{eq:estimate_of_Q_1}, 
the the first integral over $\Gamma_+ \cup \Gamma_-$ in~\eqref{eq:formula_of_tilde_psi_less_than_c} is
\begin{equation} \label{eq:estimate_of_Q_1}
O \left( n^{1/6} e^{-\factor|\xi|^{3/2}} \int_{\Gamma_+ \cup \Gamma_-} e^{n(\Hfn(z)-\ell/2)} (1 + \lvert z \rvert)^j dz \right) = 
O \left( n^{1/6} e^{-\factor|\xi|^{3/2}}  n^{-1/2} e^{n\Hfn(c)} \right)
\end{equation}
Substituting \eqref{eq:tilde_psi_n-1_linear_part_around_a_N+1} and \eqref{eq:estimate_of_Q_1} into \eqref{eq:formula_of_tilde_psi_less_than_c} and noting that $1_{(c,\infty)}(x)=0$ for $x \in E_{T,\epsilon}$, we obtain
\begin{equation} \label{eq:estimation_of_tilde_psi_around_a_N+1_subcritical}
 	\bfGamma_{n-j}(a)\tilde{\psi}_{n-j}(x) = 
	n^{1/6} e^{-\factor \lvert \xi \rvert^{3/2}} \left[ O \left( n^{-1/2} e^{n\Hfn(c)} \right) + O \left( \int^{\infty}_c  e^{n(\Gfn(y)-\ell/2)} (1 + \lvert y \rvert)^{-j} dy \right) \right].
\end{equation}
Comparing with~\eqref{eq:bfGammainother} as in the previous subsection, we obtain~\eqref{eq:estimation_of_tilde_psi_near}.
\end{proof}

\subsection{Proof of Theorem~\ref{thm:convex}\ref{enu:thm:convex:a} and Theorem \ref{thm:thm_rank_1}\ref{enu:thm:thm_rank_1:a}} \label{subsubsection:a<acc}

Recall the outline of the proof described in Section~\ref{subsection:outline_of_the_proof}. 
The proof proceed exactly same for both convex and non-convex potentials. The only important assumption is that $0 < a < \acc$.

We first evaluate $\tilde{\psi}_{n-j}(x)$. 
Fix $0<\epsilon<\delta_0$ to satisfy the conditions in Subsection \ref{sec:sub21} where $\delta_0$ is the constant in Proposition \ref{prop:asy_for_mult_cut} and its corollaries in Section~\ref{section:Result_of_RHP}. 
Since $0 < a<\acc$, the asymptotics \eqref{eq:tilde_gamma_case_1_sub} implies that 
\begin{equation} \label{eq:bound_of_Gamma_first}
	 \frac1{e^{n\ell/2} \bfGamma_{n-j}(a)}e^{n\Hfn(x)} = O\big( \sqrt{n} e^{n(\Hfn(x)-\Hfn(c))} \big).
\end{equation}
Since $\Hfn(x)> \Hfn(c)$ for all $x > c$ and $\Hfn(x) \to \infty$ fast by Lemma \ref{fact:first}, this term is larger than $O(\sqrt{n}(1+|x|)^{-j})$. Inserting this into~\eqref{eq:tilde_psi_n-1_less_than_c_1},  %noting that $\Hfn(x) \to \infty$ fast enough by Lemma \ref{fact:first}\ref{enu:fact:first:e}, 
we obtain 
\begin{equation}\label{eq:tildepsiwhenalessacc}
  \tilde{\psi}_{n-j}(x) = 
\begin{cases}
  O\big( \sqrt{n} e^{n(\Gfn(x)+\Hfn(x)-2\Hfn(c))/2} \big), &x>c, \\
  O\big( \sqrt{n} e^{n(\Gfn(x)-\Hfn(x))/2} \big), & \redge+\epsilon\le x\le c. 
\end{cases}
\end{equation}
On the other hand, for $x\in E_{T, \epsilon}:= [\redge + \frac{T}{\beta_{N+1}n^{2/3}}, \redge+\epsilon]$ (see~\eqref{eq:defn_of_E_T_epsilon}), we have from Lemma~\ref{lem:tilpsiedge} that 
\begin{equation} \label{eq:estimation_of_tilde_psi_near212}
	\tilde{\psi}_{n-j}(x) = O(n^{1/6} e^{-\factor \lvert \xi \rvert^{3/2}}),
\end{equation}
where $\xi$ is defined by~\eqref{eq:scaling_of_x_around_a_N+1}.

Now evaluate the inner product $\langle \tilde{\psi}_{n-j}, \psi_{n-j} \rangle_{\Int}$. We divide the interval $\Int$ into two parts: $(\redge+\epsilon, \infty)$ and  $E_{T, \epsilon}$. From the asymptotics \eqref{eq:tildepsiwhenalessacc} of $\tilde{\psi}_{n-j}$ and \eqref{eq:psin101} of $\psi_{n-j}$, 
\begin{equation} \label{eq:product_of_two_psis_sub_far}
\begin{split}
\langle \tilde{\psi}_{n-j}, \psi_{n-j} \rangle_{(\redge+\epsilon, \infty)} = & \int^{c}_{\redge+\epsilon} O(\sqrt{n} e^{n(\Gfn(x)-\Hfn(x))} (1 + \lvert x \rvert)^{-j} ) dx \\
& + \int^{\infty}_{c} O(\sqrt{n} e^{n(\Gfn(x)-\Hfn(c))} (1 + \lvert x \rvert)^{-j} ) dx \\
= & O(e^{-\epsilon' n})
\end{split}
\end{equation}
for a constant $\epsilon'>0$ since when $a<\acc$, $\Gfn(x)<\Hfn(c)$ for all $x>c$ and $\Gfn(x)-\Hfn(x)<0$ for all $x>\redge$ (see Lemmas~\ref{fact:first} and~\ref{lem:Gprop}). 
On the other hand, by the asymptotics \eqref{eq:estimation_of_tilde_psi_near212} of $\tilde{\psi}_{n-j}$ and \eqref{eq:psin102} of $\psi_{n-j}$, we find, after the change of variables $x\mapsto\xi$ defined in~\eqref{eq:scaling_of_x_around_a_N+1}, that 
\begin{equation} \label{eq:product_of_two_psis_sub_near}
	\langle \tilde{\psi}_{n-j}, \psi_{n-j} \rangle_{E_{T,\epsilon}} = \int^{\beta n^{2/3}\epsilon}_T O(n^{1/3} e^{-\twofactor \lvert \xi \rvert^{3/2}}) \frac{d\xi}{\beta n^{2/3}} = O(n^{-1/3}).
\end{equation}
Therefore, 
\begin{equation} \label{eq:inner_prod_of_two_psi_sub}
	\langle \tilde{\psi}_{n-j}, \psi_{n-j} \rangle_{\Int} = O(n^{-1/3}).
\end{equation}

Finally, we show the uniform boundedness of $K_{n-j,n}\chi_{\Int}\tilde{\psi}_{n-j}$ in $L^2(\Int)$. %Indeed we show that this converges to $0$ in the $L^2$ space. 
%We first compute $(K_{n-j,n}\chi_{\Int}\tilde{\psi}_{n-j})(x)$ for each $x$. 
From the asymptotics of $K_{n-j, n}$ given in Corollary~\ref{lemma:various_estimates_of_K_n-1,n} and the asymptotics~\eqref{eq:tildepsiwhenalessacc}  and~\eqref{eq:estimation_of_tilde_psi_near212} of $\tilde{\psi}_{n-j}$, we find  for $x\in E_{T,\epsilon}$ that 
\begin{equation} \label{eq:Kpsitildesub0011}
\begin{split}
	(K_{n-j,n}\chi_{\Int}\tilde{\psi}_{n-j})(x) 
	= & \int_{E_{T, 2\epsilon}} K_{n-j,n}(x,y)\tilde{\psi}_{n-j}(y) dy 
	+ \int_{\redge+2\epsilon}^\infty K_{n-j,n}(x,y)\tilde{\psi}_{n-j}(y) dy \\
	= & \int_{T}^{2\epsilon \beta n^{2/3}} O(n^{2/3} e^{-\factor|\xi|^{3/2}-\factor |\eta|^{3/2}} 
	n^{1/6} e^{-\factor |\eta|^{3/2}}) \frac{d\eta}{\beta n^{2/3}} \\
	& + \int_{\redge+2\epsilon}^{c} O(n^{1/6}e^{-\factor |\xi|^{3/2}} e^{n(\Gfn(y)-\Hfn(y))} (1 + \lvert y \rvert)^{-j}) dy\\
	& + \int_{c}^\infty O(n^{1/6}e^{-\factor |\xi|^{3/2}} e^{n(\Gfn(y)-\Hfn(c))} (1 + \lvert y \rvert)^{-j}) dy \\
	= & O(n^{1/6}e^{-\factor |\xi|^{3/2}})
\end{split}
\end{equation}
where $\xi$ is defined by~\eqref{eq:scaling_of_x_around_a_N+1}.
Similarly, for $x\ge\redge+\epsilon$, 
\begin{equation}\label{eq:Kpsitildesub0012}
\begin{split}
	(K_{n-j,n}\chi_{\Int}\tilde{\psi}_{n-j})(x)  
	= & \int_{E_{T, \epsilon/2}} K_{n-j,n}(x,y)\tilde{\psi}_{n-j}(y) dy 
	+ \int_{\redge+\epsilon/2}^\infty K_{n-j,n}(x,y)\tilde{\psi}_{n-j}(y) dy \\
	= & \int_{T}^{\frac12\epsilon \beta n^{2/3}} O(n^{1/6} e^{n(\Gfn(x)-\Hfn(x))/2} e^{-\factor |\eta|^{3/2}} 
	n^{1/6} e^{-\factor |\eta|^{3/2}}) \frac{d\eta}{\beta n^{2/3}} \\
	& + \int_{\redge+\epsilon/2}^{c} O(\sqrt{n}e^{n(\Gfn(x)-\Hfn(x))/2} e^{n(\Gfn(y)-\Hfn(y))} ) dy\\
	& + \int_{c}^\infty O(\sqrt{n} e^{n(\Gfn(x)-\Hfn(x))/2} e^{n(\Gfn(y)-\Hfn(c))} ) dy \\
	= & O(e^{n(\Gfn(x)-\Hfn(x))/2}).
\end{split}
\end{equation}
From these,  we find that 
\begin{equation}
%\begin{split}
	\| K_{n-j,n}\chi_{\Int}\tilde{\psi}_{n-j}\|_{L^2(\Int)}^2 
	=  \int_{T}^\infty O(n^{1/3}e^{-\twofactor|\xi|^{3/2}}) \frac{d\xi}{\beta n^{2/3}}
	 	+ \int_{\redge+\epsilon}^\infty O(e^{n(\Gfn(x)-\Hfn(x))})dx %\\
	=  O(n^{-1/3}).
%\end{split}
\end{equation}

Hence from~\eqref{eq:maind1}, we obtain 
\begin{equation}
\lim_{n \to \infty} \det \left( 1 - \chi_{\Int} \tilde{K}_{n-j+1,n} \chi_{\Int} \right) = \FGUE(T).
\end{equation}
Theorem~\ref{thm:convex}\ref{enu:thm:convex:a} and Theorem \ref{thm:thm_rank_1}\ref{enu:thm:thm_rank_1:a} are proved.

\subsection{Proof of  Theorem~\ref{thm:thm_rank_1}\ref{enu:thm:thm_rank_1:b} when $a<\frac12 V'(\redge)$}
\label{sec:abigacclesshalfVp}

The proof of Theorem~\ref{thm:thm_rank_1}\ref{enu:thm:thm_rank_1:b} is divided into three cases, $a<\frac12 V'(\redge)$, $a>\frac12 V'(\redge)$ and $a=\frac12 V'(\redge)$. The first case is in this subsection. The second case is in Section~\ref{subsection:generic_case_3}. The third case is discussed at the beginning of Section~\ref{subsection:The_critical_case_rank_1}.

We assume that $a\in (\acc, \frac12 V'(\redge))$ and $a\notin\mathcal{J}_V$. 
Let $x_0=x_0(a)$ be the unique maximizer of $\Gfn(x)$ in $(c, \infty)$ as in Lemma~\ref{lem:x0}. We assume that $\Gfn''(x_0)\neq 0$. See Remark~\ref{rmk:higherin25} at the end of this subsection for a discussion when $\Gfn''(x_0)=0$ (see~\eqref{eq:higher1010}). 

Recall the definition of the interval $\Intx(x_0)$ in~\eqref{eq:interval2}. Note that $\Intx(x_0)\subset (c,\infty)$. 
We first evaluate $\tilde{\psi}_{n-j}(x)$. By using~\eqref{eq:tilde_gamma_case_1_super}, we have 
\begin{equation}
	\frac1{e^{n\ell/2} \bfGamma_{n-j}(a)}e^{n\Hfn(x)} = O\big( \sqrt{n} e^{n(\Hfn(x)-\Gfn(x_0))} \big), \quad x>c.
\end{equation}
Lemma \ref{fact:first}\ref{enu:fact:first:a} and \ref{enu:fact:first:b} imply that $\Hfn(x)$ increases monotonically in $x>c$ and $\Hfn(x)>\Gfn(x)$ for all $x>\redge$. Hence
there exists $\epsilon'>0$ such that $\Hfn(x)>\Gfn(x_0)$ for all $x>x_0-\epsilon'$. In particular, $\Hfn(x)>\Gfn(x_0)$ for $x\in \Intx(x_0)$. 
Therefore~\eqref{eq:tilde_psi_n-1_less_than_c_1} yields, noting that $\Hfn(x) \to \infty$ fast enough by Lemma \ref{fact:first}\ref{enu:fact:first:e},
\begin{equation} \label{eq:exact_approx_tilde_psi_super}
	\tilde{\psi}_{n-j}(x) = e^{n(\Gfn(x)-\Hfn(x))/2} \frac{e^{n\Hfn(x)}}{e^{n\ell/2}\bfGamma_{n-j}(a)}(1+o(1)), 
	\qquad x\in \Intx(x_0).
\end{equation}
Inserting the explicit asymptotics \eqref{eq:tilde_gamma_case_1_super} for $\bfGamma_{n-j}(a)$ into \eqref{eq:exact_approx_tilde_psi_super}, we have for $x > x_0 - \epsilon'$ where $\epsilon'$ is the positive constant mentioned above, and in particular for $x\in \Intx(x_0)$ that 
\begin{equation}\label{eq:spc1}
	\tilde{\psi}_{n-j}(x)= 
	\sqrt{\frac{-n\Gfn''(x_0)}{2\pi}} \frac{1}{\M_{j,n}(x_0)} e^{n(\Gfn(x)+\Hfn(x)-2\Gfn(x_0))/2} \big(1+o(1) \big).
\end{equation}

The inner product $\langle \tilde{\psi}_{n-j}, \psi_{n-j} \rangle_{\Intx(x_0)}$ is evaluated by using~\eqref{eq:spc1} and \eqref{eq:psin101}. For $x\in \Intx(x_0)$, 
\begin{equation} \label{eq:exact_formula_of_the_product}
\begin{split}
	&\tilde{\psi}_{n-j}(x)\psi_{n-j}(x) = \sqrt{\frac{-n\Gfn''(x_0)}{2\pi}} \frac{M_{j,n}(x)}{\M_{j,n}(x_0)} e^{n(\Gfn(x)-\Gfn(x_0))} \big(1+o(1) \big).
\end{split}
\end{equation}
From the assumptions for the Theorem \ref{thm:thm_rank_1}, $\Gfn(x)$ in $(c, \infty)$ has the unique maximum at $x=x_0$ and $\Gfn(x)= \Gfn(x_0)+\frac12\Gfn''(x_0) (x-x_0)^2+O(|x-x_0|^3)$ for $x$ close to $x_0$ where  $\Gfn''(x_0)< 0$. Also $M_{j,n}(x)$, $M'_{j,n}(x)$ and $1/M_{j,n}(x)$ are bounded uniformly in $n$ for $x$ in a compact subset of $(\redge, \infty)$ and $M_{j,n}(x) = \M_{j,n}(x)(1 + o(1))$ from Proposition~\ref{prop:asy_for_mult_cut}\ref{enu:prop:asy_for_mult_cut:a},. Hence the standard Laplace's method applies and we obtain  
\begin{equation}\label{eq:abigaccGa}
	\langle \tilde{\psi}_{n-j}, \psi_{n-j} \rangle_{\Intx(x_0)} = \frac1{\sqrt{2\pi}} \int_T^\infty e^{-\frac12 \xi^2} d\xi +o(1). 
\end{equation}

We now show that $K_{n-j,n}\chi_{\Intx(x_0)} \tilde{\psi}_{n-j}$ is uniformly bounded in $L^2(\Intx(x_0))$. From~\eqref{eq:spc1} and the part \ref{lemma_enu:estimate_K_n-1,n:1} of Corollary~\ref{lemma:various_estimates_of_K_n-1,n}, for $x\in \Intx(x_0)$, 
\begin{equation} \label{eq:pointwise_est_of_K_tilde_psi_super}
\begin{split}
	 & (K_{n-j,n}\chi_{\Intx(x_0)} \tilde{\psi}_{n-j})(x)  \\
	 = & \int_{\Intx(x_0)} O(\sqrt{n} e^{n(\Gfn(x)-\Hfn(x))/2} (1 + \lvert x \rvert)^{-j} e^{n(\Gfn(y)- \Gfn(x_0))} (1 + \lvert y \rvert)^{-j}) dy \\
	 = & O(e^{n(\Gfn(x)-\Hfn(x))/2} (1 + \lvert x \rvert)^{-j}).
\end{split}
\end{equation}
%Because $\Gfn(x)-\Hfn(x) > 0$ in $E^c_n$ and $\Gfn(x)-\Hfn(x) \to \infty$ as $x \to \infty$, we find that 
Therefore, 
\begin{equation} \label{eq:the_end_of_est_of_K_tilde_psi_super}
	\lVert K_{n-j,n}\chi_{\Intx(x_0)} \tilde{\psi}_{n-j}\rVert_{L^2(\Intx(x_0))} = O(e^{-\epsilon' n}),
\end{equation}
for some $\epsilon'>0$. 

Therefore, from~\eqref{eq:maind2}, we obtain 
\begin{equation} \label{eq:final_calculation_super}
\lim_{n \to \infty} \det \left( 1 - \chi_{\Intx(x_0))} \tilde{K}_{n-j+1,n} \chi_{\Intx(x_0))}  \right) = 1 - \frac1{\sqrt{2\pi}} \int_T^\infty e^{-\frac12 \xi^2} d\xi = \erf(T),
\end{equation}
and Theorem \ref{thm:thm_rank_1}\ref{enu:thm:thm_rank_1:b} is proved.

\begin{rmk}\label{rmk:higherin25}
When $\Gfn''(x_0)= 0$, the Gaussian function $e^{-\frac12 \xi^2}$ in \eqref{eq:abigaccGa} is replaced by a higher-order function such as $e^{-\xi^{2k}}$($k > 1$). The rest of the proof is very similar. The result is the limit theorem as in~\eqref{eq:higher1010}.
\end{rmk}

%%%%%%%%%%%

\subsection{Proof of Theorem~\ref{thm:critical_traditional_split}\ref{enu:thm:critical_traditional_split:b}} \label{subsection:proof_of_thm:critical_traditional_split}

Let $V$ be a potential such that $\acc<\frac12 V'(\redge)$ and $\acc \not\in \mathcal{J}_V$. 
%Then there exits $x_0(\acc) \in (c(\acc), \infty)$ such that $\Gfn(x_0(\acc); \acc)> \Gfn(x; \acc)$ for all $x\in [c(\acc), \infty)\setminus\{x_0(\acc)\}$ and $\Gfn(x_0(\acc); \acc)= \Hfn(c(\acc); \acc)$. 
We assume that $\Gfn''(x_0(\acc);\acc) \neq 0$.
Let
\begin{equation} \label{eq:defn_of_a_pseudo_crit}
	a = \acc+ \frac{\alpha}{n},
\end{equation}
where $\alpha$ is in a compact subset of $\realR$. 

First, consider $\tilde{\psi}_{n-j}(x)$. Note that the estimates~\eqref{eq:tildepsiwhenalessacc} and~\eqref{eq:estimation_of_tilde_psi_near212} still hold. However, $\Gfn(x_0(\acc); \acc)+\Hfn(x_0(\acc);\acc)-2\Hfn(c(\acc);\acc)=0$ at $x=x_0$, so when $a = \acc$, $\langle \tilde{\psi}_{n-j}, \psi_{n-j} \rangle_{(\redge+\epsilon, \infty)}$ is no longer exponentially small.  We need an asymptotic formula of $\tilde{\psi}(x)$ like \eqref{eq:spc1}. By 
inserting  the asymptotics \eqref{eq:tilde_gamma_case_1_critical_II} of $\bfGamma_{n-1}(a)$ into \eqref{eq:tilde_psi_n-1_less_than_c_1}, %{eq:exact_approx_tilde_psi_sub}, 
similar to \eqref{eq:spc1} we obtain for $x > x_0 - \epsilon'$ where $\epsilon'$ is a positive constant defined similarly as the $\epsilon'$ in \eqref{eq:spc1}, and in particular $x\in \Intx$, that 
\begin{equation} \label{eq:pseudo_spc1} 
	\tilde{\psi}_{n-j}(x)= \frac{C_1(\alpha)}{C_0 + C_1(\alpha)}  %\\
%\times 
\sqrt{\frac{-n\Gfn''(x_0(a))}{2\pi}} \frac{1}{\M_{j,n}(x_0(a))} e^{n(\Gfn(x;a)+\Hfn(x;a)-2\Gfn(x_0(a);a))/2} \big(1+o(1) \big),
\end{equation}
where (we omit the dependence of $C_0$ and $C_1(\alpha)$ on $n$ and $j$ to make the notations simple)
\begin{align} 
	C_0 = & \frac{-i \tilde{\M}_{j,n}(c(\acc))}{\sqrt{\Hfn''(c(\acc);\acc)}}, \label{eq:defn_of_C_j(alpha)_pseudo_crit_0} \\
C_1(\alpha) = & \frac{\M_{j,n}(x_0(\acc))}{\sqrt{-\Gfn''(x_0(\acc);\acc)}} e^{\alpha (x_0(\acc) - c(\acc))}. \label{eq:defn_of_C_j(alpha)_pseudo_crit_1}
\end{align}
The constants $C_0$ and $C_1(\alpha)$ are positive from Proposition~\ref{prop:asy_for_mult_cut}. If $\alpha$ is fixed, $C_0$, $C_1(\alpha)$, $C^{-1}_0$ and $C^{-1}_1(\alpha)$ are uniformly bounded in $n$.  %If $n$ is fixed, $C_1(\alpha) \to 0$ as $\alpha \to -\infty$ and $C_1(\alpha) \to \infty$ as $\alpha \to \infty$. Define 
%The factors $p_{j,n}(\alpha)$ and $p^{(0)}_{j,n}(\alpha)$ in Theorem~\ref{thm:critical_traditional_split}\ref{enu:thm:critical_traditional_split:b} are defined as
Set
\begin{equation} \label{eq:defn_of_p(alpha)_pseudo_crit_1}
	p_{j,n}(\alpha) := \frac{C_0}{C_0 + C_1(\alpha)}. %\quad p^{(0)}_{j,n}(\alpha) := \frac{C_1(\alpha)}{C_0 + C_1(\alpha)},
\end{equation}
From the definition, $p_{j,n}(\alpha)$ is a decreasing function in $\alpha$, $p_{j,n}(\alpha) \to 0$ as $\alpha \to \infty$ and $p_{j,n}(\alpha) \to 1$ as $\alpha \to -\infty$ for each fixed $n$. %, and $p_{j,n}(\alpha) + p^{(0)}_{j,n}(\alpha) = 1$. 
Also for a fixed $\alpha$, $p_{j,n}(\alpha)$ is in a compact subset of $(0,1)$ uniformly in $n$.
Note that when the support of the equilibrium consists of one interval, then $C_0$ and $C_1(\alpha)$ are independent of $n$, and hence so is $p_{j,n}(\alpha)$.
We prove formulas \eqref{eq:part_2_of_thm_4} and \eqref{eq:part_1_of_thm_4} in Theorem \ref{thm:critical_traditional_split} separately.

\bigskip

%\subsubsection{Proof of  \eqref{eq:part_2_of_thm_4}}

The proof of~\eqref{eq:part_2_of_thm_4}
is similar to that in Subsubsection \ref{sec:abigacclesshalfVp}.  
For the evaluation of $\langle \tilde{\psi}_{n-j}, \psi_{n-j} \rangle_{\Intx(x_0)}$, we repeat the arguments of \eqref{eq:exact_formula_of_the_product}--\eqref{eq:abigaccGa}. Noting that the formula of $\tilde{\psi}_{n-j}$ in \eqref{eq:pseudo_spc1} is the same as that in \eqref{eq:spc1} except for the multiplicative factor $C_1(\alpha)/(C_0 + C_1(\alpha))$, we obtain similar to \eqref{eq:abigaccGa} that
\begin{equation} \label{eq:evaluation_of_inner_prod_tilde_psi_psi}
	\langle \tilde{\psi}_{n-j}, \psi_{n-j} \rangle_{\Intx(x_0)} = (1-p_{j,n}(\alpha)) \int^{\infty}_T e^{-\frac{1}{2}\xi^2} \frac{d\xi}{\sqrt{2\pi}} (1 + o(1)),
\end{equation}
where $p^{(0)}_{j,n}(\alpha)$ is defined in \eqref{eq:defn_of_p(alpha)_pseudo_crit_1}.
The estimate~\eqref{eq:the_end_of_est_of_K_tilde_psi_super} follows from the same calculations in Subsection \ref{sec:abigacclesshalfVp}, and we obtain from~\eqref{eq:maind2} that 
\begin{equation}
	\det\left( 1- \chi_{\Intx(x_0))} \tilde{K}_{n-j+1,n} \chi_{\Intx(x_0))}\right) = p_{j,n}(\alpha) + (1-p_{j,n}(\alpha)) \erf(T) +o(1),
\end{equation}
and \eqref{eq:part_2_of_thm_4} is proved.

\bigskip
%\subsubsection{Proof of \eqref{eq:part_1_of_thm_4}}
We now prove~\eqref{eq:part_1_of_thm_4}. When $a$ is given by \eqref{eq:defn_of_a_pseudo_crit}, the estimate \eqref{eq:tildepsiwhenalessacc} still holds. Similar to \eqref{eq:evaluation_of_inner_prod_tilde_psi_psi}, we obtain by estimates \eqref{eq:pseudo_spc1}, \eqref{eq:tildepsiwhenalessacc} and \eqref{eq:estimation_of_tilde_psi_near} that 
\begin{equation} \label{eq:evaluation_of_inner_prod_tilde_psi_psi_full}
%\begin{split}
	\langle \tilde{\psi}_{n-j}, \psi_{n-j} \rangle_{\Int} 
	=  (1-p_{j,n}(\alpha)) \int_{-\infty}^\infty \frac1{\sqrt{2\pi}} e^{-\frac12 \xi^2} d\xi (1+ o(1)) %\\
	=  (1-p_{j,n}(\alpha)) (1+ o(1)).
%\end{split}
\end{equation}
%Finally, we show that $\lVert \chi_{E^c_n} K_{n-1,n} \chi_{E^c_n} \hat{K}_n\rVert_1 = o(1)$. 
To show that $K_{n-j,n}\chi_{\Int}\tilde{\psi}_{n-j}$ is uniformly bounded in $L^2(\Int)$,  
we proceed as in~\eqref{eq:Kpsitildesub0011} and~\eqref{eq:Kpsitildesub0012}, and obtain 
\begin{equation} \label{eq:Kpsitildecri0011}
\begin{split}
	(K_{n-j,n}\chi_{\Int}\tilde{\psi}_{n-j})(x)
	= \begin{cases} O(n^{1/6}e^{-\factor |\xi|^{3/2}}), \quad &x\in E_{T, \epsilon}\\
	O(e^{n(\Gfn(x)-\Hfn(x))/2}) & x\ge \redge+\epsilon \end{cases}
\end{split}
\end{equation}
where $\xi$ is defined by~\eqref{eq:scaling_of_x_around_a_N+1}. Hence $\| K_{n-j,n}\chi_{\Int}\tilde{\psi}_{n-j}\|_{L^2(\Int)}= O(n^{-1/6})$.
Therefore, by~\eqref{eq:maind1}, we obtain
\begin{equation}
	\det\left( 1- \chi_{\Int}\tilde{K}_{n-j+1, n}  \chi_{\Int}\right) = p_{j,n}(\alpha)\FGUE(T) +o(1)
\end{equation}
and \eqref{eq:part_1_of_thm_4} is proved.

%%%%%%%%%%%%

\subsection{Proof of Theorem~\ref{thm:supercritical_split} when $a_0\in (\acc, \frac12 V'(\redge))$} % and $a\in \mathcal{J}_V$}
 \label{subsection:generic_case_5}

We prove Theorem~\ref{thm:supercritical_split}  when $a_0<\frac12V'(\redge)$. The case when 
$a_0\ge \frac12 V'(\redge)$ will be discussed in Sections~\ref{subsection:generic_case_3} and~\ref{subsection:The_critical_case_rank_1}.

Let $a_0\in (\acc, \frac12 V'(\redge))$ and $a_0\in \mathcal{J}_V$. Hence $a_0$ is a secondary critical point. In this case, the maximum of $\Gfn(x; a_0)$, $x\in (c,\infty)$, is attained at more than one point. The case when the maximum of $\Gfn(x;a_0)$ is attained at more than two points can be attained by a straightforward extension and this yields~\eqref{eq:threeormoremax}. We omit the details in that case. 

%We only consider the situation when the maximum is attained at two points in $(c,\infty)$ since the analysis is similar when there are more maximizers.  
Denote the two maximizers of $\Gfn(x;a_0)$ by $x_1:=x_1(a_0)$ and $x_2:=x_2(a_0)$. Let  $x_1(a_0)<x_2(a_0)$. Assume that 
\begin{equation}\label{eq:twoder}
\Gfn''(x_1(a_0);a_0)\neq 0, \quad \Gfn''(x_2(a_0);a_0)\neq 0.
\end{equation}
The case when one of the derivative vanishes 
is discussed in Subsection~\ref{subsection:generic_case_5_2}.
Let
\begin{equation}\label{eq:aa0sc1}
	a=a_0+ \frac{\alpha}{n}
\end{equation}
where $\alpha$ is in a compact subset of $\R$.

%The analysis is similar to the case 1 of Subsection \ref{subsection:idea_of_proffs} and  Subsection~\ref{sec:abigacclesshalfVp}. %, we only sketch the  proof and highlight the difference. 
%We denote $x_1(a_0)$ $x_2(a_0)$ as $x_1$ and $x_2$ if there is no confusion.

%Now we evaluate $<\tilde{\psi}_{n-1}, \psi_{n-1}>_{E_{n,j}}$, $j=1,2$. This is the part that changes from Subsection~\ref{sec:abigacclesshalfVp}.

%\paragraph{Proof of Theorem~\ref{thm:supercritical_split}:}

First we evaluate $\tilde{\psi}_{n-j}(x)$. The asymptotics of $\bfGamma_{n-j}(a)$ in this case is not explicitly computed in Subsection~\ref{sec:GammalessV}, hence we first compute this by extending the formula~\eqref{eq:tilde_gamma_case_1_super} (see  the paragraph following~\eqref{eq:tilde_gamma_case_1_super11}).
There are two differences from the case leading to \eqref{eq:tilde_gamma_case_1_super} in Subsection \ref{sec:GammalessV}. The first is that there are two maximizers of $\Gfn(x;a_0)$ and the second is that $a$ scales in $n$ as in~\eqref{eq:aa0sc1}. The first difference simply results in adding the contributions from the both maximizers since both term are of the same order due to the condition~\eqref{eq:twoder}. Regarding the second difference, note that 
since $\Gfn(x;a_0)$ has  maximum at $x_1(a_0)$ and $x_2(a_0)$, $\Gfn(x; a)$ has two local maxima at two points, denoted by $x_1(a)$ and $x_2(a)$, which are close to $x_1(a_0)$ and $x_2(a_0)$, respectively. (Indeed, one can easily check that $x_i(a)= x_i(a_0)+ \frac{\alpha}{-\Gfn''(x_i(a_0))n} + O(n^{-2})$.) 
Using the definition of $\Gfn$ and the fact that $x_j(a_0)$ is a critical value of $\Gfn(x; a_0)$, we find  
\begin{equation}
	\frac{d}{da}\bigg|_{a=a_0} \Gfn(x_i(a); a) = x_i(a_0), \quad i=1,2.
\end{equation}
Hence 
\begin{equation} \label{eq:relation_between_G(xa)_G(xa_0)}
	\Gfn(x_i(a); a)= \Gfn(x_i(a_0); a_0) + x_i(a_0) \frac{\alpha}{n} + O(n^{-2}), \quad i=1,2.
\end{equation}
Therefore, as in~\eqref{eq:tilde_gamma_case_1_super} we obtain as $n\to\infty$ (note that $\Gfn(x_1; a_0)= \Gfn(x_2; a_0)$ and $x_i:=x_i(a_0)$)
\begin{equation}\label{eq:bla}
\begin{split}
 	e^{n\ell/2} \bfGamma_{n-j}(a) = %& \sqrt{\frac{2\pi}{-n\Gfn''(x_1(a))}}  \M_{1,0}(x_1(a))  e^{n\Gfn(x_1(a);a)}  (1+o(1)) \\
% &+ \sqrt{\frac{2\pi}{-n\Gfn''(x_2(a))}}  \M_{1,0}(x_2(a))  e^{n\Gfn(x_2(a);a)}  (1+o(1))\\
 & e^{n\Gfn(x_1; a_0)} \Bigg[ \sqrt{\frac{2\pi}{-n\Gfn''(x_1)}}  \M_{j,n}(x_1)  e^{x_1\alpha}   %\\
 + \sqrt{\frac{2\pi}{-n\Gfn''(x_2)}}  \M_{j,n}(x_2)  e^{x_2\alpha}  \Bigg]  (1+o(1)).
\end{split}
\end{equation}
With this asymptotics of $\bfGamma_{n-j}(a)$, the rest of the analysis is similar to~\eqref{eq:spc1}, and we obtain for $x\in \Intx(x_1)$, 
\begin{multline} \label{eq:asy_formula_of_tilde_psi_tow_outliers}
\tilde{\psi}_{n-j}(x) = \Bigg[ \sqrt{\frac{2\pi}{-n\Gfn''(x_1)}}  \M_{j,n}(x_1)  e^{x_1\alpha} + \sqrt{\frac{2\pi}{-n\Gfn''(x_2)}}  \M_{j,n}(x_2)  e^{x_2\alpha}  \Bigg]^{-1} \\
\times e^{n(\Gfn(x;a)+\Hfn(x;a)-2\Gfn(x_1; a_0))/2} (1+o(1)).
\end{multline}

We now compute the inner products $\langle \tilde{\psi}_{n-j}, \psi_{n-j} \rangle_{\Intx(x_i)}$, $i=1,2$. Using~\eqref{eq:psin101}, (cf.~\eqref{eq:exact_formula_of_the_product})
\begin{equation} \label{eq:asy_of_prod_psi_two_crit_gen}
	\tilde{\psi}_{n-j}(x) \psi_{n-j}(x)= 
	\sqrt{\frac{n}{2\pi}}\frac{\M_{n,j}(x) e^{n(\Gfn(x;a)-\Gfn(x_1; a_0))}}
	{A_1(\alpha)+A_2(\alpha)} \big(1+o(1) \big)
\end{equation}
where
\begin{equation}\label{eq:defpj2}
	A_i(\alpha):=\frac{\M_{j,n}(x_i(a_0); a_0) }{\sqrt{-\Gfn''(x_i(a_0); a_0)}} e^{x_i(a_0)\alpha}, \qquad i=1,2.
\end{equation}
Like $C_0$ in~\eqref{eq:defn_of_C_j(alpha)_pseudo_crit_0} and $C_1(\alpha)$ in \eqref{eq:defn_of_C_j(alpha)_pseudo_crit_1}, 
$A_i(\alpha)$ is positive and is of finite distance away from $0$ uniformly in $n$. 
For each $i=1,2$, if we set 
\begin{equation}\label{eq:nongensupphalfV1}
	x= x_i+ \frac{\xi}{\sqrt{-n\Gfn''(x_i)}},
\end{equation}
then for $\xi$ in a compact subset of $\R$,  we find using the Taylor expansion in $x$, and~\eqref{eq:relation_between_G(xa)_G(xa_0)} that 
\begin{equation} \label{simple_relation_between_x(a_x(a_0))}
	\Gfn(x; a)= \Gfn(x_i; a_0) + x_i\frac{\alpha}{n}- \frac{\xi^2}{2n} + O(n^{-3/2}).
\end{equation}
Thus we find 
\begin{equation} \label{eq:asym_oftildepsi_psi_two_peaks}
	\tilde{\psi}_{n-j}(x) \psi_{n-j}(x)= 
	\sqrt{\frac{n \Gfn''(x_i; a_0)}{2\pi}}\frac{A_i(\alpha)}
	{A_1(\alpha)+A_2(\alpha)} e^{- \frac12 \xi^2} \big(1+o(1) \big)
\end{equation}
for $x$ given in~\eqref{eq:nongensupphalfV1} and $\xi$ in a compact subset of $\realR$, for each $i=1,2$.
Together with an easy estimate when $x$ is away from $x_1(a_0)$ and $x_2(a_0)$, this implies, as in Subsection~\ref{sec:abigacclesshalfVp}, that 
\begin{align}
	\langle \tilde{\psi}_{n-j}, \psi_{n-j} \rangle_{\Intx(x_2)} = & p^{(2)}_{j,n}(\alpha) \int_T^\infty \frac1{\sqrt{2\pi}} e^{-\frac12 \xi^2}d\xi +o(1), \label{eq:probability_right_peak_2} \\
	\langle \tilde{\psi}_{n-j}, \psi_{n-j} \rangle_{\Intx(x_1)} = & p^{(1)}_{j,n}(\alpha) \int_T^\infty \frac1{\sqrt{2\pi}} e^{-\frac12 \xi^2}d\xi + p^{(2)}_{j,n}(\alpha) +o(1) \label{eq:probability_left_peak_1}
\end{align}
for any fixed $T$, where for $i = 1,2$
\begin{equation}\label{eq:defpj1}
	p^{(i)}_{j,n}(\alpha) = \frac{A_i(\alpha)}{A_1(\alpha)+A_2(\alpha)}.
\end{equation}
The properties of $p_{j,n}^{(i)}(\alpha)$ stated in Theorem~\ref{thm:supercritical_split} can be easily checked. 

%in Theorem \ref{thm:critical_traditional_split}\ref{enu:thm:critical_traditional_split:b} and Subsection \ref{sec:abigacclesshalfVp}, we have $p^{(1)}_{j,n}(\alpha) + p^{(2)}_{j,n}(\alpha) = 1$, $p^{(1)}_{j,n}(\alpha)$ is a decreasing function, $p^{(1)}_{j,n}(\alpha) \to 0$ as $\alpha \to \infty$ and $p^{(1)}_{j,n}(\alpha) \to 1$ as $\alpha \to -\infty$ for each fixed $n$. Also for a fixed $\alpha$, $p^{(1)}_{j,n}(\alpha)$ is in a compact subset of $(0,1)$ uniformly in $n$.

%Note that both inner products are not equal to $1$ for all large $n$. 

The $L^2$ norm $\lVert K_{n-j,n}\chi_{\Intx(x_0)} \tilde{\psi}_{n-j}\rVert_{L^2(\Intx(x_i))}$ is estimated by the same argument as that for~\eqref{eq:the_end_of_est_of_K_tilde_psi_super} above  in  Subsection \ref{sec:abigacclesshalfVp} by using the asymptotics \eqref{eq:asy_formula_of_tilde_psi_tow_outliers} of $\tilde{\psi}_{n-j}$. The result is the same exponentially decaying bound.

Therefore, by~\eqref{eq:maind2}, we obtain 
\begin{align}
	 \det \left( 1 - \chi_{\Intx(x_1)} \tilde{K}_{n-j+1, n} \chi_{\Intx(x_1)} \right) = & p^{(1)}_{j,n}(\alpha)\erf(T) + o(1), \\
	 \det \left( 1 - \chi_{\Intx(x_2)} \tilde{K}_{n-j+1, n} \chi_{\Intx(x_2)}\right) = & p^{(1)}_{j,n}(\alpha) + p^{(2)}_{j,n}(\alpha))\erf(T) +o(1).
\end{align}
Thus Theorem \ref{thm:supercritical_split} when $a<\frac12 V'(\redge)$ is proved.

\subsection{Proof of Theorem~\ref{thm:noname} when $a\in (\acc, \frac12 V'(\redge))$}% and $a\in \mathcal{J}_V$} 
\label{subsection:generic_case_5_2}

We prove Theorem~\ref{thm:noname}  when $a_0<\frac12V'(\redge)$. The case when 
$a_0\ge \frac12 V'(\redge)$ will be discussed in Sections~\ref{subsection:generic_case_3} and~\ref{subsection:The_critical_case_rank_1}.

%Let $a_0\in (\acc, \frac12 V'(\redge))$ satisfy $a_0\in \mathcal{J}_V$. 
Under the assumption of Theorem~\ref{thm:noname}, for some $k > 1$
\begin{equation}\label{eq:Gpppr}
\begin{split}
	\Gfn''(x_1(a_0); a_0)\neq & 0,  \\
 	\Gfn^{(i)}(x_2(a_0); a_0)= & 0 \quad \textnormal{for $i=1, \dots, 2k-1$, }\\
 	\Gfn^{(2k)}(x_2(a_0); a_0) \neq & 0.
\end{split}
\end{equation}
We consider the double-scaling situation when 
\begin{equation}
	a=a_0 - q \frac{\log n}{n} + \frac{\alpha}{n}, \quad \textnormal{where} \quad q := \frac{\frac12 - \frac{1}{2k}}{x_2(a_0)-x_1(a_0)},
\end{equation}
for $\alpha$  in a compact subset of $\R$. 

The analysis is similar to Subsection~\ref{subsection:generic_case_5}.
For each $i=1,2$, we have, as in~\eqref{eq:relation_between_G(xa)_G(xa_0)}, 
\begin{equation}
	\Gfn(x_j(a);a) = \Gfn(x_j(a_0);a_0) + x_j(a_0)\frac{-q\log n + \alpha}{n} + o(n^{-1}).
\end{equation}
For 
\begin{align}
	x = & x_1(a_0) + \frac{\xi_1}{\sqrt{-n\Gfn''(x_1(a_0);a_0)}}, \label{eq:defn_of_xi_1}
\end{align}
we obtain, as in~\eqref{simple_relation_between_x(a_x(a_0))}, 
\begin{align}
	\Gfn(x;a) & = \Gfn(x_1(a_0); a_0) + x_1(a_0)\frac{-q\log n + \alpha}{n} - \frac{\xi^2_1}{2n} + o(n^{-1})
\end{align}
for $\xi_1$ in a compact subset of $\R$.
Similarly, for 
\begin{align}
	x = & x_2(a_0) + \left( \frac{(2k)!}{-n\Gfn^{(2k)}(x_2(a_0);a_0)} \right)^{1/(2k)}\xi_2, \label{eq:defn_of_xi_2}
\end{align}
we have (using~\eqref{eq:Gpppr})
\begin{align}
\Gfn(x;a) & = \Gfn(x_2(a_0); a_0) + x_2(a_0)\frac{-q\log n + \alpha}{n} - \frac{\xi^{2k}_2}{n} + o(n^{-1})
\end{align}
for $\xi_2$ in a compact subset of $\R$. Therefore, as in~\eqref{eq:bla} above (cf. \eqref{eq:tilde_gamma_case_1_super11}),
\begin{multline} \label{eq:calculation_of_Gamma_thm_7}
	e^{n\ell/2} \bfGamma_{n-j}(a) = e^{n\Gfn(x_1; a_0)} \Bigg[ \sqrt{\frac{2\pi}{-n\Gfn''(x_1)}}  \M_{j,n}(x_1) n^{-x_1 q} e^{x_1\alpha}   \\
+ \bigg(\frac{(2k)!}{-n\Gfn^{(2k)}(x_2;a_0)}\bigg)^{1/(2k)}  \M_{j,n}(x_2) n^{-x_2 q} e^{x_2\alpha}\int_{-\infty}^\infty e^{-\xi^{2k}}d\xi  \Bigg] (1+o(1)).
\end{multline}
Since
\begin{equation}
%Q = 
	\frac{1}{2} + x_1q = \frac{1}{2k} + x_2q = \frac{\frac{x_2}{2}-\frac{x_1}{2k}}{x_2-x_1},
\end{equation}
we can write~\eqref{eq:calculation_of_Gamma_thm_7} as 
\begin{equation} 
	e^{n\ell/2} \bfGamma_{n-j}(a) = n^{-\frac{\frac{x_2}{2}-\frac{x_1}{2k}}{x_2-x_1}} e^{n\Gfn(x_1; a_0)} \big[ B_1(\alpha)+B_2(\alpha) \big] (1+o(1))
\end{equation}
where
\begin{align}
B_1(\alpha) := & \sqrt{\frac{2\pi}{-\Gfn''(x_1)}}  \M_{j,n}(x_1)  e^{x_1\alpha}, \label{eq:definition_of_B1} \\
B_2(\alpha) := & \bigg( \frac{(2k)!}{-\Gfn^{(2k)}(x_2;a)}\bigg)^{1/(2k)}  \M_{j,n}(x_2)  e^{x_2\alpha}\int_{-\infty}^\infty e^{-\xi^{2k}_2}d\xi_2, \label{eq:definition_of_B2}.
\end{align}
As in~\eqref{eq:asy_formula_of_tilde_psi_tow_outliers}, for $x\in 
\Intx(x_1)$, 
\begin{equation} \label{eq:asy_formula_of_tilde_psi_tow_outliers333}
\begin{split}
 	\tilde{\psi}_{n-j}(x) = 
 &  n^{\frac{\frac{x_2}{2}-\frac{x_1}{2k}}{x_2-x_1}} \frac{e^{n(\Gfn(x;a)+\Hfn(x;a)-2\Gfn(x_1(a_0); a_0))/2}}{B_1(\alpha)+B_2(\alpha)}   (1+o(1)).
\end{split}
\end{equation}
From this we obtain, as in~\eqref{eq:asym_oftildepsi_psi_two_peaks}, that 
\begin{equation}
\tilde{\psi}_{n-j}(x)\psi_{n-j}(x) = p^{(1)}_{j,n}(\alpha) \sqrt{\frac{-n\Gfn''(x_1;a_0)}{2\pi}} e^{-\frac{1}{2}\xi^2_1} (1+o(1))
\end{equation}
for $x$ given in~\eqref{eq:defn_of_xi_1} and $\xi_1$ in a compact subset of $\R$, 
and 
\begin{equation}
\tilde{\psi}_{n-j}(x)\psi_{n-j}(x) = %\\
p^{(2)}_{j,n}(\alpha) \left( \frac{-n\Gfn^{(2k)}(x_2;a_0)}{(2k)!} \right)^{1/(2k)} \frac{e^{-\xi^{2k}_2}}{\int^{\infty}_{-\infty} e^{-\xi^{2k}_2} d\xi_2} (1+o(1))
\end{equation}
for $x$ given in~\eqref{eq:defn_of_xi_2} and $\xi_1$ in a compact subset of $\R$, where for $i = 1,2$
\begin{equation}\label{eq:defpj1_2}
	p^{(i)}_{j,n}(\alpha) := \frac{B_i(\alpha)}{B_1(\alpha)+B_2(\alpha)}.
\end{equation}
From the definition, the properties of $p^{(i)}_{j,n}(\alpha)$ in Theorem~\ref{thm:noname} follow easily.
%Note that similar to the $p_{j,n}(\alpha)$ and $p^{(0)}_{j,n}(\alpha)$ in Theorem \ref{thm:critical_traditional_split}\ref{enu:thm:critical_traditional_split:b} and Subsection \ref{sec:abigacclesshalfVp}, and the $p_{j,n}(\alpha)$ and $p^{(0)}_{j,n}(\alpha)$ in Theorem \ref{thm:supercritical_split} and Subsection \ref{subsection:generic_case_5}, we have $p^{(1)}_{j,n}(\alpha) + p^{(2)}_{j,n}(\alpha) = 1$, $p^{(1)}_{j,n}(\alpha)$ is a decreasing function, $p^{(1)}_{j,n}(\alpha) \to 0$ as $\alpha \to \infty$ and $p^{(1)}_{j,n}(\alpha) \to 1$ as $\alpha \to -\infty$ for each fixed $n$. Also for a fixed $\alpha$, $p^{(1)}_{j,n}(\alpha)$ is in a compact subset of $(0,1)$ uniformly in $n$.

Thus it follows as in~\eqref{eq:probability_right_peak_2} and~\eqref{eq:probability_left_peak_1} that 
\begin{align}
\langle \tilde{\psi}_{n-j}(x), \psi_{n-j}(x) \rangle_{\Intxx(x_2;k)} = & p^{(2)}_{j,n}(\alpha) \frac{1}{\int_{-\infty}^\infty e^{- \xi^{2k}}d\xi}\int_T^\infty e^{- \xi^{2k}}d\xi +o(1), \label{eq:inner_prod_two_max_nongen_1} \\
\langle \tilde{\psi}_{n-j}(x), \psi_{n-j}(x) \rangle_{\Intx(x_1)} = & p^{(1)}_{j,n}(\alpha) \int_T^\infty \frac1{\sqrt{2\pi}} e^{-\frac12 \xi^2}d\xi + p^{(2)}_{j,n}(\alpha) +o(1). \label{eq:inner_prod_two_max_nongen_2}
\end{align}

As in the proof of Theorem \ref{thm:thm_rank_1}\ref{enu:thm:thm_rank_1:b} in Subsections \ref{sec:abigacclesshalfVp} and in the proof of Theorem \ref{thm:supercritical_split} in Subsection \ref{subsection:generic_case_5}, we have $\lVert K_{n-j,n} \chi_{\Intx(x_i)} \tilde{\psi}_{n-j} \rVert_{L^2(\Intx(x_i))} \to 0$.

Thus we obtain, from~\eqref{eq:maind2}, that 
\begin{align}
 	\det \left( 1 - \chi_{\Intx(x_1)} \tilde{K}_{n-j+1,n} \chi_{\Intx(x_1)} \right) = & p^{(1)}_{j,n}(\alpha)\erf(T) +o(1), \\
	\det \left( 1 - \chi_{\Intxx(x_2;k)} \tilde{K}_{n-j+1,n} \chi_{\Intxx(x_2;k)}  \right) = & p^{(1)}_{j,n}(\alpha) + p^{(2)}_{j,n}(\alpha)) \frac{\int^T_{-\infty} e^{\xi^{2k}} d\xi}{\int^{\infty}_{-\infty} e^{\xi^{2k}} d\xi} +o(1), 
\end{align}
and Theorem \ref{thm:noname} when $a_0< \frac12 V'(\redge)$ is proven.

%%%%%%%%%%%

\section{When $a>\frac12 V'(\redge)$}\label{subsection:generic_case_3}

Note that if $a>\frac12 V'(\redge)$, then $a>\acc$. In this section, we prove Theorem~\ref{thm:convex}\ref{enu:thm:thm_rank_1:b} and  Theorems~\ref{thm:thm_rank_1}\ref{enu:thm:thm_rank_1:b},~\ref{thm:supercritical_split} and~\ref{thm:noname} for the case when $a$ (or $a_0$) $>\frac12 V'(\redge)$.
After a small change at the first step, the analysis is the same as in the case when $\acc<a<\frac12 V'(\redge)$ discussed in Subsection~\ref{sec:abigacclesshalfVp},\ref{subsection:generic_case_5} and~\ref{subsection:generic_case_5_2}. The proof of  Theorem~\ref{thm:convex} \ref{enu:thm:convex:b} is identical to the proof of Theorem~\ref{thm:thm_rank_1} \ref{enu:thm:thm_rank_1:b}. 

%If $a > \frac{1}{2}V'(\redge)$, then $a > \acc$ since $\acc \leq \frac{1}{2}V'(\redge)$. Both the results and proofs for the $a > \frac{1}{2}V'(\redge)$ case are similar to the $\acc < a < \frac{1}{2}V'(\redge)$ case, and we only highlight the difference.

Note that $c(a)=\redge$ in this case (see Definition~\ref{def:ca}).
Since $\Gfn'(\redge)>0$ when $a>\frac{1}{2}V'(\redge)$, $\Gfn_{\max}(a):= \max\{ \Gfn(x;a): x\in [\redge, \infty)\}$ satisfies $\Gfn_{\max}(a)>\Gfn(\redge)=\Hfn(\redge)$. Let $\epsilon>0$ be small enough so that all the maximizers of $\Gfn$ are in $(\redge+2\epsilon, \infty)$ and 
\begin{equation}\label{eq:GfnmaxHfn}
	\Gfn_{\max}(a)>\Hfn(\redge + 2\epsilon)>\Hfn(\redge+\epsilon).
\end{equation}
We have the following formula of $\tilde{\psi}_{n-j}(x)$. This is the analogue of Lemma~\ref{eq:approx_of_tilde_psi_far}.

\begin{lemma}\label{lem:tildepsilargea}
Let $a > \frac{1}{2}V'(\redge)$. As $n\to\infty$ while $j=O(1)$, 
\begin{equation} \label{eq:expression_of_Gamma_a>half_Vprime}
	\bfGamma_{n-j}(a) = \int^{\infty}_{\redge+\epsilon} \varphi_{n-j}(y)e^{nay} dy (1 + o(1))
\end{equation}
and
\begin{equation} \label{eq:tilde_psi_n-1_less_than_c0_a_greater}
\begin{split}
	\tilde{\psi}_{n-j}(x) = e^{n (\Gfn(x)-\Hfn(x))/2} 
	\bigg\{ \frac1{e^{n\ell/2}  \bfGamma_{n-j}(a)}e^{n\Hfn(x)}   + O((1+|x|)^{-j}) \bigg\}
\end{split}
\end{equation}
for $x \in [\redge+2\epsilon, \infty)$. 
\end{lemma}

\begin{proof}
As in~\eqref{eq:division_of_psi_n_exponent}, we write $\bfGamma_{n-j}(a) = \bfGamma_{n-j}(a;n)$ as
\begin{equation} \label{eq:integral_expression_of_tilde_gamma_in_case_1_prime}
	\bfGamma_{n-j}(a) = - \int_{\bar{\Gamma}_+\cup\bar{\Gamma}_-} (C\varphi_{n-j})(z)e^{naz} dz + \int^{\infty}_{\redge+\epsilon} \varphi_{n-j}(y)e^{nay} dy,
\end{equation}
where, for a large enough but fixed positive constant $C_{\bar{\Gamma}}$, (\cf\ the contour $\Gamma$ defined in \eqref{eq:definition_of_contour_Gamma})
\begin{equation} \label{eq:definition_of_contour_Gamma_bar}
	\bar{\Gamma}_+: =  \{ \redge+\epsilon + it \mid 0 < t  \leq C_{\bar{\Gamma}} \}\cup  \{ \redge+\epsilon + iC_{\bar{\Gamma}} - t \mid t \geq 0 \}
\end{equation}
and $\bar{\Gamma}_-$ is the reflected image of $\bar{\Gamma}_+$ about the real axis. The contours are oriented as indicated in Figure~\ref{fig:Gamma_bar}.
\begin{figure}[htp]
\centering
\includegraphics{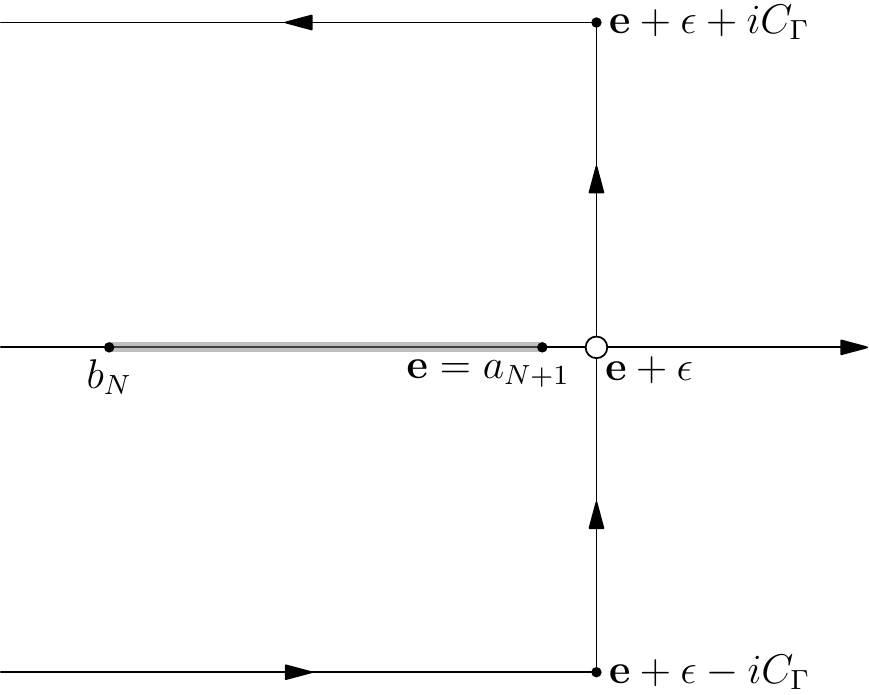}
\caption{The contours $\bar{\Gamma}_+$ and $\bar{\Gamma}_-$}\label{fig:Gamma_bar}
\end{figure}

As in \eqref{eq:CphyMremark}, by using \eqref{eq:asy_of_Cvarphi_n-j} % and~\eqref{eq:asy_of_tilde_M_j_mult} 
for $(C\varphi_{n-j})(z)$, the contour integral over $\bar{\Gamma}$ in \eqref{eq:integral_expression_of_tilde_gamma_in_case_1_prime} satisfies
\begin{equation}\label{eq:CphyMremark_bar}
\begin{split}
	&\int_{\bar{\Gamma}_+\cup\bar{\Gamma}_-}  (C\varphi_{n-j})(z)e^{naz} dz  = \int_{\bar{\Gamma}_+\cup\bar{\Gamma}_-} \tilde{M}_{j,n} (z)  e^{n(\Hfn(z;a) - \ell/2)}  dz.
\end{split}
\end{equation}
Now $\Hfn(z)$ has no saddle point in $\bar{\Gamma}_{\pm}$ since $\Hfn(z) \neq 0$ for all $z \in \bar{\Gamma}$. However, 
it is easy to check that for $z(t) = \redge+\epsilon + it$, $0<t< C_{\bar{\Gamma}}$ and  $z(t) = \redge + \epsilon + iC_{\bar{\Gamma}}-t$, the formulas \eqref{eq:decreasing_counterclockwise_1} and \eqref{eq:decreasing_counterclockwise_2} still hold verbatim except that $c$ becomes $\redge+\epsilon$, provided that $C_{\bar{\Gamma}}$ is large enough, say, $C_{\bar{\Gamma}} > \frac{1}{2a}$. Hence $\Re \Hfn(z)$ decreases strictly as $z$ travels along $\bar{\Gamma}_+$ in the direction of the orientation. Similarly, $\Re \Hfn(z)$ increases strictly as $z$ travels along $\bar{\Gamma}_-$ in the direction of the orientation. Noting that for fixed $n$, $\tilde{M}_{j,n}(z) \to z^{j-1}$ as $z \to \infty$ and for fixed $z$, $\tilde{M}_{j,n}(z)$ is uniformly bounded in $n$ from Proposition~\ref{prop:asy_for_mult_cut}, we obtain 
\begin{equation} \label{eq:case_3_contour_integral_Cpsi}
\begin{split}
	 \int_{\bar{\Gamma}_+\cup\bar{\Gamma}_-}  (C\varphi_{n-j})(z)e^{naz} dz
	= O\big( e^{n\Hfn(\redge+\epsilon)-n\ell/2} \big).
\end{split}
\end{equation}
On the other hand, consider the second 
integral in \eqref{eq:integral_expression_of_tilde_gamma_in_case_1_prime}:
\begin{equation}\label{eq:GtoMthen0101}
	\int^{\infty}_{\redge+\epsilon} \varphi_{n-j}(y)e^{nay} dy 
	= \int^{\infty}_{\redge+\epsilon}  M_{j,n}(y) e^{n(\Gfn(y;a)-\ell/2)}dy.
\end{equation}
%We also assume that all maximizers of $\Gfn(x;a)$ in $[c,\infty)$ are included in $(\redge+\epsilon, \redge+\epsilon^{-1})$. 
By using  the Laplace's method, we find an estimate similar to~\eqref{eq:GtoMthen}. Hence we find that~\eqref{eq:GtoMthen0101} is exponentially larger than~\eqref{eq:case_3_contour_integral_Cpsi} due to the assumption~\eqref{eq:GfnmaxHfn}. Thus \eqref{eq:expression_of_Gamma_a>half_Vprime} is proven.

Now consider $\tilde{\psi}_{n-j}(x)$.
Analogous to \eqref{eq:formula_of_tilde_psi_less_than_c} in Lemma \ref{lem:tpsinew}, we have, for $x \in \realR \setminus \{ \redge + \epsilon \}$, 
\begin{equation}
\begin{split}
	&\bfGamma_{n-j}(a) \tilde{\psi}_{n-j}(x) = e^{n( ax - V(x)/2)} 1_{(\redge + \epsilon, \infty)}(x) \\
&\quad +  \int_{\bar{\Gamma}_+\cup\bar{\Gamma}_-}   \CK_{n-j,n}(x,z) e^{naz} dz   
- \int^{\infty}_{\redge + \epsilon} K_{n-j,n}(x,y)  e^{n(ay-V(y)/2)} dy .
\end{split}
\end{equation}
Using this formula, due to the property of the $\Hfn(x;a)$ on $\bar{\Gamma}_{\pm}$ and $\Gfn(x;a)$ on $(\redge+\epsilon, \infty)$, the analysis of the proof of Lemma~\ref{eq:approx_of_tilde_psi_far} applies without any changes. If we restrict $x\ge \redge+2\epsilon$, then the error term $O(\sqrt{n}(1+|x|)^{-j})$ in~\eqref{eq:tilde_psi_n-1_less_than_c_1} can be replaced by $O((1+|x|)^{-j})$ since $|x-z|\ge \epsilon$ for $z\in \bar{\Gamma}_{\pm}$ as in the first part of the proof of Lemma~\ref{lem:tpsinew}. We skip the details. 
\end{proof}

Since $\tilde{\psi}_{n-j}$ is the only term that depends on $a$ and its asymptotic formula for $a>\frac12V'(\redge)$ is same %(apart from the change of $O(\sqrt{n})$ by $O(1)$)
as the case when $\acc <a<\frac12V'(\redge)$ in  Section~\ref{subsection:generic_case_1}, all the analysis in Section~\ref{subsection:generic_case_1} hold without any changes. Therefore, we obtain the proof of Theorem~\ref{thm:convex}\ref{enu:thm:convex:b} and  Theorems~\ref{thm:thm_rank_1}\ref{enu:thm:thm_rank_1:b},~\ref{thm:supercritical_split} and~\ref{thm:noname} for the case when $a$ (or $a_0$) $>\frac12 V'(\redge)$.

%%%%%%%%%%%%%%

\section{When $a=\frac12 V'(\redge)$} \label{subsection:The_critical_case_rank_1}

First, suppose that $a=\frac12V'(\redge)>\acc$. Then $\Gfn_{\max}(a):= \max\{ \Gfn(x;a) \mid x\in [\redge, \infty)\}$ satisfies $\Gfn_{\max}(a)>\Hfn(\redge)=\Gfn(\redge)$ (recall Definition~\ref{eq:definition_of_acc} and~\eqref{eq:definition_of_calA}). This property is enough to prove Lemma~\ref{lem:tildepsilargea} and the analysis of  Section~\ref{subsection:generic_case_3} applies without any change. Hence we obtain the proof of Theorem \ref{thm:thm_rank_1}\ref{enu:thm:thm_rank_1:b}, %Theorem~\ref{thm:critical_traditional_split}\ref{enu:thm:critical_traditional_split:b}, 
\ref{thm:supercritical_split} and~\ref{thm:noname} when $a=\frac12V'(\redge)>\acc$. Combining the results of the previous two sections, we have proved all theorems except for Theorems~\ref{thm:convex}\ref{enu:thm:convex:b}, \ref{thm:critical_traditional_split}\ref{enu:thm:critical_traditional_split:a} and~\ref{thm:rank_1_transit_crit}.

Theorem~\ref{thm:convex}\ref{enu:thm:convex:b} and Theorem~\ref{thm:critical_traditional_split}\ref{enu:thm:critical_traditional_split:a} share the same proof and this is given in  
Subsection~\ref{sec:4.2}.
The proof of Theorem~\ref{thm:rank_1_transit_crit} is in Subsection~\ref{subsection:Proof_of_Theorem1.4}.

%The rest of this section is a proof of  Theorem~\ref{thm:critical_traditional_split}\ref{enu:thm:critical_traditional_split:a} and Theorem~\ref{thm:rank_1_transit_crit}.
%Let $V$ be a potential such that $\acc=\frac12V'(\redge)$. 

\subsection{Proof of Theorem~\ref{thm:convex}(b) and~\ref{thm:critical_traditional_split}\ref{enu:thm:critical_traditional_split:a}}\label{sec:4.2}

Let $V$ be a potential such that $\acc=\frac12V'(\redge)$. 
We assume that $\acc\notin \mathcal{J}_V$. (This holds under the assumption of convexity of Theorem~\ref{thm:convex}(a).) Then $\Gfn(\redge;\acc)>\Gfn(x;\acc)$ for all $x>\redge$. 
We consider a double-scaling situation when 
\begin{equation}\label{eq:anearcrigen}
	a = \acc+\frac{\beta\alpha}{n^{1/3}},
\end{equation}
where $\alpha$ is a real number in a compact subset of $\realR$.

\subsubsection{Computation of $\bfGamma_{n-j}(a)$}  \label{sec:4.2.1}

\begin{lemma}\label{lem:GammaBQQ}
We have 
\begin{equation} \label{eq:estimation_of_tilde_gamma_critical}
	\bfGamma_{n-j}(a) =  \frac{Q_n}{\beta\sqrt{n}} e^{\alpha^3/3} (\B_{j, n}(\redge)+o(1)), 
\end{equation}
where %(we omit the dependence of $Q_n$ on $a$ to lighten the notation)
\begin{equation} \label{eq:defn_of_C(alpha)}
	Q_n =Q_n(a):= e^{n(\Gfn(\redge; a)-\ell/2)}.
\end{equation}
Here $\B_{j,n}(z)$ is given in~\eqref{eq:psi_general}. Note that $\B_{j,n}(\redge)$ is in a compact subset of $(0,\infty)$ independent of $n$. 
\end{lemma}

\begin{proof}
In the proof of Lemma~\ref{lem:tpsinew} when $a<\frac12 V'(\redge)$, we have taken the contour $\Gamma_{\pm}$ to pass the point $c=c(a)$ at which $\Re \Hfn(z;a)$, $z\in\Gamma_{\pm}$, takes its maximum (see~\eqref{eq:division_of_psi_n_exponent}). Near this point, we had $\Hfn(z;a)-\Hfn(c(a); a)\sim \kappa (z-c(a))^2$ for some constant $\kappa>0$. This quadratic term changes when $a=\acc=\frac12 V'(\redge)$. In this case, $c(\acc)=\redge$, and (note~\eqref{eq:properties_of_gfn})
\begin{equation}
	\Hfn'(z; \acc)= -\gfn'(z)+ \acc = -\gfn'(z)+\gfn'(\redge)
\end{equation}
vanishes at $z=\redge$. 
Now since $\gfn'(z)$ is the Cauchy transform of the equilibrium measure $\Psi(x)$ (see~\eqref{eq:eqmeasure}), which vanishes like a square root at $x=\redge$, we find that there is a constant $\kappa'>0$ such that 
$\Hfn'(z;\acc)\sim \kappa'(z-\redge)^{1/2}$ for $z$ near $\redge$ such that $z-\redge\notin\R_-$. Hence 
$\Hfn(z;\acc)-  \Hfn(\redge; \acc)\sim \kappa(z-\redge)^{3/2}$  for a constant $\kappa>0$. This $3/2$-order of vanishing implies that near $z=\redge$, $\Re\Hfn(z; \acc)$ decreases most rapidly in the direction of angle $2\pi/3$ and $-2\pi/3$ as $z$ travels away from $\redge$. 

With the above preliminary computation in mind, we define\footnote{Here, the exact shape and the angle of the contour from $z=\redge$ is not important. For example, we can use the contour that extends straightly upward from $z=\redge$ as in Figure~\ref{fig:Gamma_complex} with $c$ replaced by $\redge$. The local behavior $\Hfn(z;\acc)-  \Hfn(\redge; \acc)\sim \kappa(z-\redge)^{3/2}$ near $z=\redge$ shows that $\Re\Hfn(z)$ decays as $z$ travels vertically away from $\redge$ at least locally. One can check $\Re\Hfn(z)$ indeed decreases as $z$ moves away from $\redge$ along on the entire curve. Our choice of the contour $\Sigma$ is made for the convenience of the formulas that appear later.}
the contours $\Sigma_{\pm}$ as (see Figure~\ref{fig:Gamma_critical})
\begin{equation} \label{eq:decomposition_of_Gamma_crit}
\Sigma_+ =  \Sigma_{2+} \cup \Sigma_{2'+} \cup \Sigma_{3}, \quad \Sigma_- = \Sigma_{1} \cup \Sigma_{2'-} \cup \Sigma_{2-},
\end{equation} 
and
\begin{equation}\label{eq:decomposition_of_Gamma_crit11}
\begin{split}
\Sigma_{1} = & \{ \redge+\omega^2\epsilon-iC_{\crit} + t \mid t \leq 0 \}, \\
\Sigma_{2'-} = & \{ \redge+\omega^2\epsilon+it \mid -C_{\crit} \leq t \leq 0 \}, \\
\Sigma_{2-} = & \{ \redge-\omega^2 t \mid -\epsilon \leq t < 0 \}, \\
\Sigma_{2+} = & \{ \redge+\omega t \mid 0 < t \leq \epsilon \}, \\
\Sigma_{2'+} = & \{ \redge+\omega\epsilon+it \mid 0 \leq t \leq C_{\crit} \}, \\
\Sigma_{3} = & \{ \redge+\omega\epsilon+iC_{\crit} - t \mid t \geq 0 \},
\end{split}
\end{equation}
where $\omega:= e^{2\pi i/3}$.
Here $\epsilon$ is a fixed constant chosen to satisfy the condition~\eqref{eq:condition_of_epsilon_crit} below, and  $C_{\crit}$ is a positive fixed constant large enough, say, greater than $1/(2a)$. As in~\eqref{eq:division_of_psi_n_exponent}, we have
\begin{equation} \label{eq:evaluation_of_tilde_gamma_critical}
	\bfGamma_{n-j}(a) =  - \int_{\Sigma_+\cup\Sigma_-} (C\varphi_{n-j})(z)e^{naz} dz
+ \int^{\infty}_{\redge} \varphi_{n-j}(y)e^{nay} dy.
\end{equation}

\begin{figure}
\begin{center}
\includegraphics{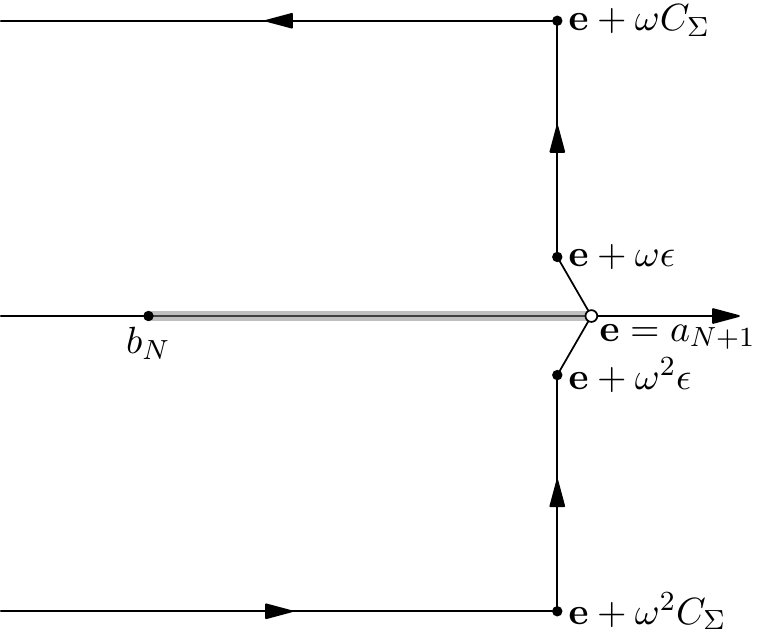}
\caption{The contours $\Sigma_+$ and $\Sigma_-$.}\label{fig:Gamma_critical} 
\end{center}
\end{figure}

Let
\begin{equation}
\Sigma^{\ess}_{2\pm} = \Sigma_{2\pm} \cap \{ z \in \compC \mid \lvert z-\redge \rvert < n^{-11/21} \}.
\end{equation}
We first consider the part of the first integral in~\eqref{eq:evaluation_of_tilde_gamma_critical} over $\Sigma^{\ess}_{2+}$. Inserting the asymptotics~\eqref{eq:Cpsi_upper_plan_general} for $(C\varphi_{n-j})(z)$, there are two terms, one involving $\Ai(\omega^2\Phi(z))$ and the other involving $\Ai'(\omega^2\Phi(z))$. We compute each of the integrals using the change of variables $z\mapsto\xi$ defined by 
\begin{equation} 
	z := \redge+\frac{\omega}{\beta n^{2/3}}\xi.
\end{equation}
This change of variables and  the double scaling~\eqref{eq:anearcrigen} imply that 
\begin{equation} \label{eq:simple_estimation_of_e^n(ax-V(x)/2)}
	e^{n(-\frac12V(z)+az)} = Q_n e^{\alpha\omega\xi} (1+o(1))
\end{equation}
uniformly in $z\in \Sigma^{\ess}_{2+}$, 
where (recall Lemma~\ref{fact:first}\ref{enu:fact:first:c})
\begin{equation}
	Q_n = e^{n(-\frac{1}2V(\redge) + a\redge)}=e^{n(\Hfn(\redge; a)-\ell/2)} = e^{n(\Gfn(\redge; a)-\ell/2)},
\end{equation}
%\begin{multline} 
%	Q_n = e^{n(-\frac{1}2V(\redge) + a\redge)}= e^{n(-\frac12V(\redge)+\acc \redge) + n^{2/3}\beta\alpha \redge}  = \\
%e^{n(\Hfn(\redge; a)-\ell/2)} = e^{n(\Gfn(\redge; a)-\ell/2)}
%\end{multline}
as defined in~\eqref{eq:defn_of_C(alpha)}. 
Therefore, 
using the property~\eqref{eq:asy_pf_Phi_at_e} of $\Phi(z)$, and noting that $|z-\redge|\le n^{-11/21}$, the two integrals involving $\Ai(\omega^2\Phi(z))$ and $\Ai'(\omega^2\Phi(z))$ satisfy  
\begin{equation} \label{eq:the_first_asy_involving_Airy}
	\int_{\Sigma^{\ess}_{2+}} \Ai(\omega^2\Phi(z))B_{j,n}(z) e^{n(-\frac12 V(z)+az)}dz 
=\frac{\omega Q_n}{\beta n^{2/3}} \left( \B_{j,n}(\redge) \int^{\infty}_0 \Ai(\xi)e^{\omega\alpha\xi}d\xi  + o(1) \right)
\end{equation}
and
\begin{equation} \label{eq:the_first_asy_involving_Airy_like}
	\int_{\Sigma^{\ess}_{2+}} \Ai'(\omega^2\Phi(z))D_{j,n}(z) e^{n(-\frac12 V(z)+az)}dz
=\frac{\omega Q_n}{\beta n^{2/3}} \left( \D_{j,n}(\redge) \int^{\infty}_0 \Ai'(\xi)e^{\omega\alpha\xi}d\xi + o(1) \right).
\end{equation}
Observe that the integrals involving $\Ai(\xi)$ and $\Ai'(\xi)$ are convergent as these functions decay faster than exponential functions as $\xi\to+\infty$. 
From these, we find that 
\begin{equation} \label{eq:the_first_asy_involving_Airy00}
\int_{\Sigma^{\ess}_{2+}} (C\varphi_{n-j})(z) e^{naz}dz = -\frac{Q_n}{\beta\sqrt{n}}  \left( \B_{j,n}(\redge) \int^{\infty}_0 \Ai(\xi)e^{\omega\alpha\xi}d\xi  + o(1) \right).
\end{equation}

Now consider $\Sigma_{2+} \setminus \Sigma^{\ess}_{2+}$. By the property \eqref{eq:asy_pf_Phi_at_e} of $\Phi(z)$, we have that for $\lvert z - \redge \rvert < 1$, there exists $c_1 > 0$ such that
\begin{equation}
	\lvert n^{-2/3}\Phi(z) - \beta(z-\redge) \rvert < c_1 \lvert z-\redge \rvert^2.
\end{equation}
Hence the asymptotics of $\Ai(\xi)$ and $\Ai'(\xi)$ as $\xi\to\infty$ (\cite[10.4.59 and 10.4.61]{Abramowitz-Stegun64}) imply that $z \in \Sigma_{2+}$, if $|z-\redge|\le \beta(1 - (3/4)^{2/3})/c_1$, then
\begin{align} 
\Ai(\omega^2\Phi(z)) = & O((1+|\xi|)^{-1/4} e^{-\frac{1}{2}\xi^{3/2}}), \label{eq:est_of_Ai_Sigma_ess} \\
\Ai'(\omega^2\Phi(z)) = & O((1+|\xi|)^{1/4} e^{-\frac{1}{2}\xi^{3/2}}).\label{eq:est_of_Ai_prime_Sigma_ess}
\end{align}
Hence~\eqref{eq:Cpsi_upper_plan_general} implies that 
\begin{equation} \label{eq:est_of_Cphi_on_Sigma_ess}
	(C\varphi_{n-j})(z) e^{\frac{n}{2}V(z)} = O(n^{1/6} e^{-\frac{1}{2} \xi^{3/2}}).
\end{equation}
Also for $\lvert z - \redge \rvert < 1$, there exists $c_2 > 0$ such that (\cf\ \eqref{eq:simple_estimation_of_e^n(ax-V(x)/2)})
\begin{equation} \label{eq:est_of_exp_of_V-av_near_e}
	Q^{-1}_n e^{n(-\frac{1}{2}V(z) + az)} = O(e^{\alpha\omega\xi + n^{-1/3}c_2\xi^2}).
\end{equation}
Hence if we take $\epsilon$ in~\eqref{eq:decomposition_of_Gamma_crit11} small enough so that 
\begin{equation} \label{eq:condition_of_epsilon_crit}
\epsilon < 1, \quad \epsilon < \beta c^{-1}_1 \left( 1 -(3/4)^{2/3} \right) \quad \textnormal{and} \quad \epsilon < \beta^{-1} \left( \frac{1}{4} \right)^2 c^{-2}_2,
\end{equation}
then combining \eqref{eq:est_of_Cphi_on_Sigma_ess} and \eqref{eq:est_of_exp_of_V-av_near_e}, we have, for $n\ge (8\alpha)^{14}/\beta^7$, 
\begin{equation} \label{eq:estimate_on_Gamma_crit2_noness}
\begin{split}
& \int_{\Sigma_{2+} \setminus \Sigma^{\ess}_{2+}} (C\varphi_{n-j})(z)e^{naz} dz \\
= & Q_n \int_{\beta n^{1/7}}^{\epsilon \beta n^{2/3}} O \left( n^{1/6} e^{-\frac{1}{2} \xi^{3/2} + \alpha\omega\xi + n^{-1/3}c_2 \xi^2} \right) \frac{d\xi}{n^{2/3}} \\
= & Q_n \int_{\beta n^{1/7}}^{\epsilon \beta n^{2/3}} O \left( n^{-1/2} e^{-\frac{1}{2} \xi^{3/2} +\frac18\xi^{3/2}+ \frac14 \xi^{3/2}} \right) d\xi 
= Q_n O(e^{-\frac{\beta^{3/2}}{8} n^{3/14}}).
\end{split}
\end{equation}
%where $\bar{\epsilon}$ is a positive constant depending on $\epsilon$.

For the rest of $\Sigma_+$, by a direct calculation as in the inequalities \eqref{eq:decreasing_counterclockwise_1} and \eqref{eq:decreasing_counterclockwise_2}, we find that $\Re (H(z;a))$ decreases strictly as $z$ travels away from $\redge+\omega\epsilon$ along  $\Sigma_{2'+} \cup \Sigma_3$. Also by direct calculation we verify that
\begin{equation}\label{eq:Hfnreald}
	 \bar{\epsilon}' := \Re \Hfn(\redge; \acc)- \Re \Hfn(\redge + \omega\epsilon; \acc)  > 0,
\end{equation}
where $\bar{\epsilon}'$ is a positive constant depending on $\epsilon$. Since $\Hfn(z;a)-\Hfn(z;\acc)= \frac{\beta\alpha}{n^{1/3}}z \to 0$ as $n\to\infty$ for a fixed $z$, the difference~\eqref{eq:Hfnreald} with $\acc$ replaced by $a$ is also bounded below by $\frac12 \bar{\epsilon}'$ for large enough $n$. Thus, from Proposition \ref{prop:asy_for_mult_cut}\ref{enu:prop:asy_for_mult_cut:b}, 
\begin{equation} \label{eq:estimate_on_Gamma_crit_3}
\begin{split}
\int_{\Sigma_{2'+} \cup \Sigma_{3}} (C\varphi_{n-j})(z)e^{naz} dz = & O\bigg( \int_{\Sigma_{2'+} \cup \Sigma_{3}} |z|^{j-1} e^{n(\Re\Hfn(z;a)-\ell/2)} d \lvert z \rvert \bigg)\\
= & e^{n(\Hfn(\redge; a)-\ell/2)}O(e^{-\frac14 n\bar{\epsilon}'}) \\
= & Q_n O(e^{-\frac14 n\bar{\epsilon}'}).
\end{split}
\end{equation}

Combining \eqref{eq:the_first_asy_involving_Airy00}, \eqref{eq:estimate_on_Gamma_crit2_noness} and \eqref{eq:estimate_on_Gamma_crit_3}, we obtain
\begin{equation} \label{eq:estimation_of_tilde_gamma_upper}
	\int_{\Sigma_+} (C\varphi_{n-j})(z)e^{naz} dz = -\frac{Q_n}{\beta\sqrt{n}} \left( \B_{j,n}(\redge) \int^{\infty}_0 \Ai(\xi)e^{\omega\alpha\xi}d\xi + o(1) \right).
\end{equation}

The integral over $\Sigma_-$ can be  evaluated in a similar way. Alternatively we can use the symmetry $(C\varphi_{n-j})(\bar{z})= -\overline{(C\varphi_{n-j})(z)}$. We have  
\begin{equation} \label{eq:estimation_of_tilde_gamma_lower}
\int_{\Sigma_-} (C\varphi_{n-j})(z)e^{naz} dz = -\frac{Q_n}{\beta\sqrt{n}}  \left( \B_{j,n}(\redge) \int^{\infty}_0 \Ai(\xi)e^{\omega^2\alpha\xi}d\xi + o(1) \right).
\end{equation}

\bigskip

For the integral over $(\redge, \infty)$ in~\eqref{eq:evaluation_of_tilde_gamma_critical}, we again consider three intervals $(\redge, \redge+n^{-11/21}]$, $(\redge+n^{-11/21}, \redge+\epsilon)$  and $[\redge+\epsilon, \infty)$, and proceed as before. We now use the asymptotics~\eqref{eq:psi_general} for $\varphi_{n-j}(z)$ in the first two intervals and \eqref{eq:asy_of_varphi_n-j} for the third one. %, and asymptotics similar to \eqref{eq:est_of_Ai_Sigma_ess} and \eqref{eq:est_of_Ai_prime_Sigma_ess} that for $x \in [\redge, \redge+\epsilon]$ and $x = \redge + \beta^{-1}n^{-2/3}\xi$
%\begin{align} 
%\Ai(\Phi(x)) = & \Ai(\xi) + O(n^{-2/3}e^{-\frac{1}{3}\xi^{3/2}}) = O(n^{-2/3}e^{-\frac{1}{3}\xi^{3/2}}), \label{eq:est_of_Ai_Sigma_horizon} \\
%\Ai'(\Phi(x)) = & \Ai'(\xi) + O(n^{-2/3}e^{-\frac{1}{3}\xi^{3/2}}) = O(n^{-2/3}e^{-\frac{1}{3}\xi^{3/2}}). \label{eq:est_of_Ai_prime_Sigma_horizon}
%\end{align}
Note the similarity of~\eqref{eq:psi_general} and~\eqref{eq:Cpsi_upper_plan_general}. The calculation is similar and we obtain 
\begin{equation} \label{eq:integral_over_linear_critical}
	\int^{\infty}_{\redge} \varphi_{n-j}(y)e^{nay}dy = \frac{Q_n}{\beta\sqrt{n}} \left( \B_{j,n}(\redge) \int^{\infty}_0 \Ai(\xi)e^{\alpha\xi}d\xi + o(1) \right).
\end{equation}

Combining~\eqref{eq:estimation_of_tilde_gamma_upper}, ~\eqref{eq:estimation_of_tilde_gamma_lower} and~\eqref{eq:integral_over_linear_critical}, we find
\begin{equation} 
	\bfGamma_{n-j}(a) = \frac{Q_n}{\beta\sqrt{n}} \left( \B_{j,n}(\redge) \int^{\infty}_0 \Ai(\xi)(e^{\alpha\xi}+ e^{\omega\alpha\xi}+e^{\omega^2\alpha\xi}) d\xi + o(1) \right).
\end{equation}
But it is easy to check that 
\begin{equation} \label{eq:one_consequence_of_Leach}
	\int^{\infty}_0 \Ai(\xi)\left( e^{\alpha\xi} + e^{\omega\alpha\xi} + e^{\omega^2\alpha\xi} \right) d\xi = e^{\alpha^3/3}
\end{equation}
for $\alpha\in\compC$. Indeed if we denote the left-hand-side of~\eqref{eq:one_consequence_of_Leach} by $f(\alpha)$, then by using $\Ai''(\xi)=\xi\Ai(\xi)$ and $1+\omega+\omega^2=0$, we find $f'(\alpha)=\alpha^2 f(\alpha)$. Recalling that $\int_0^\infty \Ai(\xi)d\xi = \frac13$ (see e.g. \cite[9.10.11]{Abramowitz-Stegun64}), we obtain~\eqref{eq:one_consequence_of_Leach}.
Thus we obtain~\eqref{eq:estimation_of_tilde_gamma_critical}.
% the statement that $\B_{j,n}(\redge)$ is in a compact subset of $(0, \infty)$ is a direct consequence of Proposition \ref{prop:asy_for_mult_cut}\ref{enu:prop:asy_for_mult_cut:c}.

\end{proof}

\subsubsection{Evaluation of $\tilde{\psi}_{n-j}(x)$}\label{sec:4.2.2}

We evaluate $\tilde{\psi}_{n-j}(x)$ for $x\in \Int=[\redge + \beta^{-1}n^{-2/3}T, \infty)$ as $n\to\infty$. Let $\delta_0$ be the constant in Section~\ref{section:Result_of_RHP}.

\begin{lemma}\label{lem:tildepsi43}
Let $0<\epsilon<2\delta_0$ be the constant in \eqref{eq:decomposition_of_Gamma_crit11}, satisfying the condition~\eqref{eq:condition_of_epsilon_crit}. For $x\in E_{T,\epsilon/2}:= [\redge + \beta^{-1}n^{-2/3}T, \redge+\epsilon/2]$, we have 
\begin{equation} \label{eq:estimation_of_gamma_psi_near}
\begin{split}
	\tilde{\psi}_{n-j}(x) 
	&= \frac{\beta \sqrt{n}}{\B_{j,n}(\redge)}  \left[ C_{-\alpha}(\xi) + e^{-\alpha^3/3} \bigg( \frac1{Q_n} e^{n(ax - V(x)/2)} - e^{\alpha\xi}\bigg) +o(1) \right] 
\end{split}
\end{equation}
where $\xi$ is defined by the relation $x= \redge+\beta^{-1}n^{-2/3}\xi$ as in \eqref{eq:scaling_of_x_around_a_N+1}, $Q_n$ is given in~\eqref{eq:defn_of_C(alpha)} and $C_\alpha(\xi)$ is defined in~\eqref{eq:Calphadef}.
For $x\ge \redge+\epsilon/2$, we have 
\begin{equation} \label{eq:alternative_def_of_s^(1)_far}
\begin{split}
	\tilde{\psi}_{n-j}(x)  
	&= \frac{\beta \sqrt{n} e^{-\alpha^3/3}}{\B_{j,n}(\redge)} e^{n(\Gfn(x;a)+\Hfn(x;a)-2\Gfn(\redge;a))/2} (1+o(1)).
\end{split}
\end{equation}
\end{lemma}

\bigskip

Note that  $\xi\in [T,\infty)$. Let $\bar{C} < T$ be a real number, and set $C_n := \beta^{-1}n^{-2/3}\bar{C}$. 
For $x  \in \Int$ (hence $x>\redge+C_n$), we have, as in Lemma~\ref{lem:tpsinew},
\begin{multline} \label{eq:expression_of_gamma_psi_critical}
	\bfGamma_{n-j}(a) \tilde{\psi}_{n-j}(x) 
= e^{n( ax - V(x)/2)} + \int_{(\Sigma_++C_n)\cup(\Sigma_-+C_n)} \CK_{n-j,n}(x,z) e^{naz} dz\\
\phantom{e^{n( ax - V(x)/2)}} - \int^{\infty}_{\redge+C_n} K_{n-j,n}(x,y)  e^{n(ay-V(y)/2)} dy.
% &\qquad \qquad\quad  - \int^{\infty}_{c} \frac{\psi_{n-1}(x)\varphi_{n-2}(y) - \psi_{n-2}(x)\varphi_{n-1}(y)}{x-y} e^{nay} dy  \bigg]
\end{multline}
Here $\Sigma_{\pm}+C_n$ denotes the contour $\Sigma_{\pm}$ translated by $C_n$. For example, $\Sigma_++C_n=(\Sigma_{2+}+C_n) \cup (\Sigma_{2'+}+C_n) \cup (\Sigma_{3}+C_n)$, \cf\ \eqref{eq:decomposition_of_Gamma_crit}. We divide the proof of Lemma \ref{lem:tildepsi43} into two parts.

\begin{proof}[Proof of~\eqref{eq:estimation_of_gamma_psi_near}]
First we consider the integral over $\Sigma_++C_n$ in~\eqref{eq:expression_of_gamma_psi_critical}. 
For $x\in E_{T, \epsilon/2}$ and $z\in\Sigma_{2+}+C_n$, from~\eqref{eq:psi_general} and~\eqref{eq:Cpsi_upper_plan_general},
\begin{equation}\label{eq:CKUVVU}
\begin{split}
\CK_{n-j,n}(x,z) = & \bigg[   n^{1/3} U_1(x,z) \Ai(\Phi(x))\Ai(\omega^2\Phi(z)) \\
	& + V_1(x,z) \Ai(\Phi(x))\omega^2\Ai'(\omega^2\Phi(z)) \\
	& + V_2(x,z) \Ai'(\Phi(x))\Ai(\omega^2\Phi(z)) \\
	& + n^{-1/3} U_2(x,z) \Ai'(\Phi(x))\omega^2\Ai'(\omega^2\Phi(z)) \bigg]e^{\pi i/3}  e^{-nV(z)/2},
\end{split}
\end{equation}
where
\begin{equation}
\begin{split}
	U_1(x,z)& = \frac{\gamma_{n-j-1}}{\gamma_{n-j}} \frac{B_{j,n}(x)B_{j+1,n}(z)-B_{j+1,n}(x)B_{j,n}(z)}{x-z} \\
	  V_1(x,z) & = \frac{\gamma_{n-j-1}}{\gamma_{n-j}} \frac{B_{j,n}(x)D_{j+1,n}(z)-B_{j+1,n}(x)D_{j,n}(z)}{x-z}=: \frac{W_1(x,z)}{x-z} \\
	V_2(x,z) & = \frac{\gamma_{n-j-1}}{\gamma_{n-j}} \frac{D_{j,n}(x)B_{j+1,n}(z)-D_{j+1,n}(x)B_{j,n}(z)}{x-z}=: \frac{W_2(x,z)}{x-z}  \\
	U_2(x,z)& = \frac{\gamma_{n-j-1}}{\gamma_{n-j}} \frac{D_{j,n}(x)D_{j+1,n}(z)-D_{j+1,n}(x)D_{j,n}(z)}{x-z} .
\end{split}
\end{equation}
We now show that the main contribution to $\CK_{n-j,n}(x,z)$ comes from the middle two terms on the right-hand side of~\eqref{eq:CKUVVU}.
Clearly, $U_1(x,z)=O(1)$ and $U_2(x,z)=O(1)$. 
Since  $W_1(x,x)=1=-W_2(x,x)$ from~\eqref{eq:cross_product_of_B_D}, it follows that 
\begin{equation}\label{eq:VV}
\begin{split}
	&V_1(x,z) \Ai(\Phi(x))\omega^2\Ai'(\omega^2\Phi(z))  + V_2(x,z) \Ai'(\Phi(x))\Ai(\omega^2\Phi(z)) \\
	= & \frac{ \Ai(\Phi(x))\omega^2\Ai'(\omega^2\Phi(z)) -\Ai'(\Phi(x))\Ai(\omega^2\Phi(z))}{x-z}\\
	& + \frac{W_1(x,z)-W_1(x,x)}{x-z} \Ai(\Phi(x))\omega^2\Ai'(\omega^2\Phi(z)) \\
	& + \frac{W_2(x,z)-W_2(x,x)}{x-z}  \Ai'(\Phi(x))\Ai(\omega^2\Phi(z)).
\end{split}
\end{equation}
Observe that $\frac{W_i(x,z)-W_i(x,x)}{x-z} =O(1)$, $i=1,2$.
For
\begin{equation}\label{eq:xzrange10}
	x=\redge+ \frac{\xi}{\beta n^{2/3}} \in E_{T, \epsilon/2}, \qquad z=\redge+ \frac{\eta}{\beta n^{2/3}} \in \Sigma_{2+} +C_n,
\end{equation}
using the estimates~\eqref{eq:est_of_Ai_Sigma_ess} and \eqref{eq:est_of_Ai_prime_Sigma_ess} for $\Ai(\omega^2\Phi(z))$ and $\Ai'(\omega^2\Phi(z))$, and analogous estimates for $\Ai(\Phi(x))$ and $\Ai'(\Phi(x))$, %and \eqref{eq:est_of_Ai_Sigma_horizon} and \eqref{eq:est_of_Ai_prime_Sigma_horizon} 
we find that (noting that the $\xi$ in \eqref{eq:est_of_Ai_Sigma_ess} and \eqref{eq:est_of_Ai_prime_Sigma_ess} are slightly different from the $\xi$ and $\eta$ in \eqref{eq:xzrange10})
\begin{equation}\label{eq:tem101}
	\CK_{n-j,n}(x,z)= %\\
\bigg[  \beta n^{2/3}  \frac{\omega^2\Ai(\xi)\Ai'(\omega^2\eta) - \Ai'(\xi)\Ai(\omega^2\eta)}{\xi-\eta} %\\
+ O(n^{1/3} e^{-\frac{1}{2} (|\xi|^{3/2}+|\eta|^{3/2})})  \bigg]e^{\pi i/3}  e^{-nV(z)/2}
\end{equation}
for $x$ and $z$ in~\eqref{eq:xzrange10}.
For $z\in (\Sigma_{2'+}+C_n)\cup(\Sigma_{3+}+C_n)$, noting that $|x-z|\ge 1/\epsilon$, a straightforward calculation using~\eqref{eq:psi_general} and~\eqref{eq:asy_of_Cvarphi_n-j} implies that 
\begin{equation}\label{eq:temp102}
\begin{split}
	\CK_{n-j,n}(x,z)
	&=  O\big( n^{1/6} e^{-\frac12|\xi|^{3/2}} (1+|z|)^j e^{n(-\gfn(z)+\ell/2)} \big).
\end{split}
\end{equation}
Hence we obtain, by noting~\eqref{eq:simple_estimation_of_e^n(ax-V(x)/2)}, \eqref{eq:est_of_exp_of_V-av_near_e} and the condition \eqref{eq:condition_of_epsilon_crit} satisfied by $\epsilon$, using $-\gfn(z)+\ell/2+az= \Hfn(z;a)-\ell/2$ and noting the calculation in~\eqref{eq:estimate_on_Gamma_crit_3}, that 
\begin{equation} \label{eq:expression_of_gamma_psi_criticalKKK}
\begin{split}
	 &\int_{\Sigma_{+}+C_n} \CK_{n-j,n}(x,z) e^{naz} dz \\
	 = & Q_n\left[ e^{\frac{\pi i}{3}} \int_{\bar{C}}^{\omega \cdot \infty} \frac{\Ai(\xi)\Ai'(\omega^2\eta)\omega^2 - \Ai'(\xi)\Ai(\omega^2\eta)}{\xi-\eta} e^{\alpha\eta} d\eta \right. + \left. \vphantom{\int_{\bar{C}}^{\omega \cdot \infty}} O(n^{-1/3}e^{-\frac{1}{2}\lvert \xi \rvert^{3/2}}) \right].
\end{split}
\end{equation}

\medskip

The estimates for the integral over the contour $\Sigma_-+C_n$ can be obtained either by Schwarz reflection principle or by a similar calculation. We find 
\begin{equation}\label{eq:estimation_of_gamma_psi_critical_lower}
\begin{split}
	& \int_{\Sigma_{-}+C_n} \CK_{n-j,n}(x,z) e^{naz} dz \\
	= & Q_n \left[ -e^{\frac{\pi i}{3}} \int^{\bar{C}}_{\omega^2 \cdot \infty} \frac{\Ai(\xi)\Ai'(\omega\eta) - \omega^2\Ai'(\xi)\Ai(\omega\eta)}{\xi-\eta} e^{\alpha\eta} d\eta \right.  + \left. \vphantom{\int^{\bar{C}}_{\omega^2 \cdot \infty}} O(n^{-1/3}e^{-\frac{1}{2}\lvert \xi \rvert^{3/2}}) \right].
\end{split}
\end{equation}

For the integral over $(\redge+C_n, \infty)$ in \eqref{eq:expression_of_gamma_psi_critical}, we need asymptotics of $K_{n-j,n}(x,y)$. 
% analogous to that of $\CK_{n-j,n}(x,z)$ stated in \eqref{eq:tem101}. 
For $x \in E_{T,\epsilon/2}$ and $y \in (\redge+C_n, \redge+\epsilon)$, setting $x = \redge + \beta^{-1}n^{-2/3}\xi$ and $y = \redge + \beta^{-1}n^{-2/3}\eta$, we have 
%we obtain from \eqref{eq:psi_general}, \eqref{eq:est_of_Ai_Sigma_horizon}, \eqref{eq:est_of_Ai_prime_Sigma_horizon} and the discussion above \eqref{eq:tem101} that
\begin{equation} \label{eq:est_of_K(xy)_near}
%\begin{split}
K_{n-j,n}(x,y) =  \beta n^{2/3}  \frac{\Ai(\xi)\Ai'(\eta) - \Ai'(\xi)\Ai(\eta)}{\xi-\eta} %\\
 + O(n^{1/3} e^{-\frac{1}{2} (|\xi|^{3/2}+|\eta|^{3/2})}).
%\end{split}
\end{equation}
This follows from the analysis similar to that of~\eqref{eq:tem101}. A weaker estimate is in~\eqref{eq:KtoKairywitherror}, which is actually enough for our purpose. 
%\begin{equation} \label{eq:est_of_K(xy)_near}
%\begin{split}
%K_{n-j,n}(x,y) = & n^{1/3} U_1(x,y) \Ai(\Phi(x))\Ai(\Phi(y)) \\
%& + V_1(x,y) \Ai(\Phi(x))\omega^2\Ai'(\omega^2\Phi(y)) \\
%& + V_2(x,y) \Ai'(\Phi(x))\Ai(\omega^2\Phi(y)) \\
%& + n^{-1/3} U_2(x,y) \Ai'(\Phi(x))\omega^2\Ai'(\omega^2\Phi(y)) \\
%= & \beta n^{2/3}  \frac{\omega^2\Ai(\xi)\Ai'(\eta) - \Ai'(\xi)\Ai(\eta)}{\xi-\eta} \\
%& + O(n^{1/3} e^{-\frac{1}{3} (|\xi|^{3/2}+|\eta|^{3/2})}).
%\end{split}
%\end{equation}
Hence for $x \in E_{T, \epsilon/2}$, 
% and $y \in (\redge+C_n, \redge+\epsilon)$ and the asymptotics \eqref{eq:est_of_exp_of_V-av_near_e}, we obtain
\begin{equation}\label{eq:DDD4:1}
\begin{split}
& \int_{\redge+C_n}^{\redge+\epsilon} K_{n-j,n}(x,y) e^{n(ay-\frac12 V(y))}  dy \\
= & Q_n\bigg[ \int_{\bar{C}}^\infty   \frac{\Ai(\xi)\Ai'(\eta) - \Ai'(\xi)\Ai(\eta)}{\xi-\eta} e^{\alpha \eta} d\eta   +  O(n^{-1/3}e^{-\frac{1}{2}\lvert \xi \rvert^{3/2}}) \bigg].
\end{split}
\end{equation}
Now consider the integral over $(\redge+\epsilon, \infty)$. 
By the estimate \eqref{eq:lemma_enu:estimate_K_n-1,n:4} of $K_{n-j,n}(x,y)$ for $x \in E_{T, \epsilon/2}$ and $y \in (\redge+\epsilon, \infty)$ and the identity
\begin{equation} \label{eq:alternative_expression_of_ay-V(y)/2}
	-\frac{1}{2}V(y) + ay - \frac{1}{n}\log Q_n = \frac{\Gfn(y;a)+\Hfn(y;a)-2\Gfn(\redge;y)}{2},
\end{equation}
%we have that there exist $\bar{\epsilon}' > 0$ dependong on $\epsilon$
we obtain 
\begin{equation}\label{eq:DDD4:2}
\begin{split}
& \int_{\redge+\epsilon}^{\infty} K_{n-j,n}(x,y) e^{n(ay-\frac12 V(y))}  dy \\
= & Q_n \int_{\redge+\epsilon}^{\infty} O \left( n^{1/6}e^{-\factor \lvert \xi \rvert^{3/2}} (1+ \lvert y \rvert)^{-j}e^{n(\Gfn(y;a)-\Gfn(\redge;a))} \right) dy = Q_n O( e^{-\factor \lvert \xi \rvert^{3/2} - \frac{\bar{\epsilon}'}{2}n}),
%= & Q_n o(e^{\bar{\epsilon}'n} e^{-\frac{2}{3} \lvert \xi \rvert^{3/2}}).
\end{split}
\end{equation}
where 
\begin{equation}
\bar{\epsilon}' = \min_{y \geq \redge+\epsilon} \Gfn(\redge;a) - \Gfn(y;a) > 0.
\end{equation}
Combining \eqref{eq:DDD4:1} and \eqref{eq:DDD4:2}, we have
\begin{equation}\label{eq:DDD4}
\begin{split}
& \int_{\redge+C_n}^{\infty} K_{n-j,n}(x,y) e^{n(ay-\frac12 V(y))}  dy \\
= & Q_n\bigg[ \int_{\bar{C}}^\infty   \frac{\Ai(\xi)\Ai'(\eta) - \Ai'(\xi)\Ai(\eta)}{\xi-\eta} e^{\alpha \eta} d\eta  +  O(n^{-1/3}e^{-\factor\lvert \xi \rvert^{3/2}}) \bigg].
\end{split}
\end{equation}
\bigskip

We now insert the results~\eqref{eq:expression_of_gamma_psi_criticalKKK},~\eqref{eq:estimation_of_gamma_psi_critical_lower} and~\eqref{eq:DDD4} into~\eqref{eq:expression_of_gamma_psi_critical} and evaluate $\tilde{\psi}_{n-j}(x)$ for $x\in E_{T,\epsilon}$. The combination of the three integrals can be simplified. For this purpose, note that 
\begin{align}
	\frac{\Ai(\xi)\Ai'(\eta) - \Ai'(\xi)\Ai(\eta)}{\xi-\eta} = & \int^{\infty}_0 \Ai(\xi+t)\Ai(\eta+t) dt, \\
\frac{\omega^2\Ai(\xi)\Ai'(\omega^2\eta) - \Ai'(\xi)\Ai(\omega^2\eta)}{\xi-\eta} = & \int^{\infty}_0 \Ai(\xi+t)\Ai(\omega^2(\eta+t)) dt, \label{eq:second_Airy_property} \\
\frac{\omega\Ai(\xi)\Ai'(\omega\eta)-\Ai'(\xi)\Ai(\omega\eta)}{\xi-\eta} = &  \int^{\infty}_0 \Ai(\xi+t)\Ai(\omega(\eta+t)) dt, \label{eq:third_Airy_property} 
\end{align}
where we require $\Re\eta < \Re\xi$ in \eqref{eq:second_Airy_property} and \eqref{eq:second_Airy_property}. 
This can be verified by using $\Ai''(z)=z\Ai(z)$ and the asymptotics \cite[10.4.59 and 10.4.61]{Abramowitz-Stegun64} of $\Ai(z)$ and $\Ai'(z)$ as $z \to \infty$. Using the above Airy function identities, we find that 
\begin{equation}
\begin{split}
 \int^{\infty}_{\bar{C}} \frac{\Ai(\xi)\Ai'(\eta) - \Ai'(\xi)\Ai(\eta)}{\xi-\eta} e^{\alpha\eta}d\eta %\\
= & \int^{\infty}_{\bar{C}} \int^{\infty}_0 \Ai(\xi+t)\Ai(\eta+t) dt e^{\alpha\eta}d\eta \\
= & \int^{\infty}_0 \Ai(\xi+t) e^{-\alpha t}  \left( \int^{\infty}_{\bar{C}+t} \Ai(\bar{\eta})e^{\alpha\bar{\eta}} d\bar{\eta}\right) dt.
\end{split}
\end{equation}
Similarly, 
\begin{equation}
	\int^{\omega \cdot \infty}_{\bar{C}} \frac{\omega^2\Ai(\xi)\Ai'(\omega^2\eta) - \Ai'(\xi)\Ai(\omega^2\eta)}{\xi-\eta} e^{\alpha\eta}d\eta %\\ 
	=\int^{\infty}_0  \Ai(\xi+t) e^{-\alpha t} \left( \int^{\omega \cdot \infty}_{\bar{C}+t} \Ai(\omega^2\bar{\eta})e^{\alpha\bar{\eta}}  d\bar{\eta}\right) dt,
\end{equation}
and 
\begin{equation}
	\int^{\bar{C}}_{\omega^2 \cdot \infty} \frac{\Ai(\xi)\Ai'(\omega\eta) - \omega^2\Ai'(\xi)\Ai(\omega\eta)}{\xi-\eta} e^{\alpha\eta} d\eta %\\ 
	=\omega^2 \int^{\infty}_0 \Ai(\xi+t) e^{-\alpha t} \left( \int^{\bar{C}+t}_{\omega^2 \cdot \infty} \Ai(\omega \bar{\eta})e^{\alpha\bar{\eta}} d\bar{\eta} \right) dt.
\end{equation}
From these results and~\eqref{eq:expression_of_gamma_psi_critical},~\eqref{eq:expression_of_gamma_psi_criticalKKK},~\eqref{eq:estimation_of_gamma_psi_critical_lower} and~\eqref{eq:DDD4}, we find that for $x\in E_{T,\epsilon/2}$, 
\begin{multline}
\bfGamma_{n-j}(a) \tilde{\psi}_{n-j}(x) = 
e^{n( ax - V(x)/2)} 
-Q_n\bigg[ \int^{\infty}_0 \Ai(\xi+t) e^{-\alpha t} \bigg( \int^{\infty}_{\bar{C}+t} \Ai(\bar{\eta})e^{\alpha\bar{\eta}} d\bar{\eta} \\ 
+ \omega^2 \int^{\omega \cdot \infty}_{\bar{C}+t} \Ai(\omega^2\bar{\eta})e^{\alpha\bar{\eta}}  d\bar{\eta} + \omega \int_{\bar{C}+t}^{\omega^2 \cdot \infty} \Ai(\omega \bar{\eta})e^{\alpha\bar{\eta}} d\bar{\eta} \bigg) dt +o(1) \bigg].
\end{multline}
Now the sum of three integrals inside the parentheses equals $e^{\alpha^3/3}$ (cf.~\eqref{eq:one_consequence_of_Leach}), for all $t$. In order to see this, first note that the sum is independent of $t$ since its derivative with respect to $t$ equals $0$ from the Airy function identity 
\begin{equation}\label{eq:Airyidensum}
	\Ai(z)+\omega\Ai(\omega z)+\omega^2\Ai(\omega^2 z) = 0, \qquad z\in \compC.
\end{equation}
Then set $t=0$ and call the sum $S(\alpha)$. Taking the derivative of $S(\alpha)$ with respect to $\alpha$ and using~\eqref{eq:Airyidensum} and the differential equation for the Airy function, we find that 
$S'(\alpha)=\alpha^2 S(\alpha)$. Now by noting that $S(0)= 1$ since $\int_0^\infty \Ai(\bar{\eta})d\bar{\eta}= 1/3$, we obtain that $S(\alpha)= e^{\alpha^3/3}$.  
%and then take the derivative with respect to $\alpha$.
%note that the This follows by solving the linear first-order differential equation that this term satisfies as a function of $\alpha$ and recalling that 
Hence
\begin{equation}\label{eq:Gt1}
	\bfGamma_{n-j}(a)\tilde{\psi}_{n-j}(x) = Q_n
	\bigg[ \frac1{Q_n} e^{ax - \frac12V(x)}- e^{\alpha^3/3} \int^{\infty}_0 \Ai(\xi+t)e^{-\alpha t}dt  + o(1) \bigg]
\end{equation}
uniformly for $x\in E_{T, \epsilon/2}$.
Therefore using~\eqref{eq:estimation_of_tilde_gamma_critical} we find that 
\begin{equation} \label{eq:estimation_of_gamma_psi_near-1}
\begin{split}
	\tilde{\psi}_{n-j}(x) &= \frac{\beta \sqrt{n}}{\B_{j,n}(\redge)}  \left[ e^{-\alpha^3/3} \frac1{Q_n} e^{ax - \frac12V(x)} - \int^{\infty}_0 \Ai(\xi+t)e^{-\alpha t}dt +o(1) \right] \\
	&= \frac{\beta \sqrt{n}}{\B_{j,n}(\redge)}  \left[ C_{-\alpha}(\xi) + e^{-\alpha^3/3} \bigg( \frac1{Q_n} e^{ax - \frac12V(x)} - e^{\alpha\xi}\bigg) +o(1) \right] 
\end{split}
\end{equation}
uniformly for $x=\redge+\beta^{-1}n^{-2/3}\xi \in E_{T, \epsilon/2}$. In the last line, we used the identity 
\begin{equation} \label{eq:alternative_def_of_s^(1)}
	e^{-\alpha^3/3} e^{\alpha\xi} - \int^{\infty}_0 \Ai(\xi+t)e^{-\alpha t}dt = C_{-\alpha}(\xi).
\end{equation}

\end{proof}

\begin{proof}[Proof of~\eqref{eq:alternative_def_of_s^(1)_far}]

Let $x\ge \redge+\epsilon/2$. Using~\eqref{eq:lemma_enu:estimate_K_n-1,n:3} and~\eqref{eq:est_of_exp_of_V-av_near_e}, a straightforward estimate implies that 
\begin{equation}\label{eq:1DDD3}
	\frac{1}{Q_n} \int_{\redge+C_n}^{\redge+\epsilon/4}K_{n-j,n}(x,y) e^{n(ay-\frac12 V(y))}  dy = 
O(n^{-1/2} (1 + \lvert x \rvert)^{-j} e^{n(\Gfn(x;a)-\Hfn(x;a))/2}).
\end{equation}
Similarly, we obtain using~\eqref{eq:lemma_enu:estimate_K_n-1,n:1}  
% that there exists $\bar{\epsilon} > 0$ depending on $\epsilon$ such that
\begin{equation}\label{eq:1DDD01}
	\frac{1}{Q_n} \int_{\redge+\epsilon/4}^\infty K_{n-j,n}(x,y) e^{n(ay-\frac12 V(y))}  dy = 
O( (1 + \lvert x \rvert)^{-j} e^{n(\Gfn(x;a)-\Hfn(x;a))/2}).
%O(e^{-\bar{\epsilon}n} (1 + \lvert x \rvert^{-j}) e^{n(\Gfn(x;a)-\Hfn(x;a))/2}).
\end{equation}

For  the integral on $\Sigma_{\pm}+C_n$, the calculation is easier than the proof of ~\eqref{eq:estimation_of_gamma_psi_near} since $|x-z|\ge \epsilon/2$. Straightforward estimates using Proposition~\ref{prop:asy_for_mult_cut} imply that 
\begin{equation}\label{eq:1DDE3}
	 \frac1{Q_n} \int_{\Sigma_{\pm}+C_n} \CK_{n-j,n}(x,z) e^{naz} dz= 
	 O(n^{-1/2} (1 + \lvert x \rvert)^{-j} e^{n(\Gfn(x;a)-\Hfn(x;a))/2}).
\end{equation}

Therefore, from~\eqref{eq:expression_of_gamma_psi_critical}, 
\begin{equation} \label{eq:alternative_def_of_s^(1)_far01}
\begin{split}
	\frac1{Q_n} \bfGamma_{n-j}(a) \tilde{\psi}_{n-j}(x)  %\\
	&  =  \frac1{Q_n} e^{n(-\frac12V(x)+ax)}  + O((1 + \lvert x \rvert)^{-j} e^{n(\Gfn(x;a)-\Hfn(x;a))}) \\
	&= e^{n(\Gfn(x;a)+\Hfn(x;a)-2\Gfn(\redge;a))/2} \big(1+O((1 + \lvert x \rvert)^{-j}  e^{-n(\Hfn(x;a)-\Gfn(x;a))}) \big) \\
	&= e^{n(\Gfn(x;a)+\Hfn(x;a)-2\Gfn(\redge;a))/2} (1+o(1)).
\end{split}
\end{equation}
Hence using \eqref{eq:estimation_of_tilde_gamma_critical},
%of $\bfGamma_{n-j}(a)$ into \eqref{eq:crit_tildepsi_far_almost}, 
we obtain~\eqref{eq:alternative_def_of_s^(1)_far}.

%Since for $x \in (\redge+\epsilon/2, \infty)$, we have from \eqref{eq:alternative_expression_of_ay-V(y)/2} and the fact $\Gfn(\redge; a) = \Hfn(\redge; a) + O(n^{-1/3})$
%\begin{multline} \label{eq:crit_tildepsi_far_almost}
%\frac{1}{2}V(x) + ax - \frac{1}{n}\log Q_n = \frac{\Gfn(x;a)-\Hfn(x;a)}{2} \\
%+ (\Hfn(x;a)-\Hfn(\redge;a)) + O(n^{-1/3}) > \frac{\Gfn(x;a)-\Hfn(x;a)}{2}
%\end{multline}
%for large enough $n$, from~\eqref{eq:expression_of_gamma_psi_critical} we find that  
%\begin{equation} \label{eq:alternative_def_of_s^(1)_far01}
%\begin{split}
%	\frac1{Q_n} \bfGamma_{n-j}(a) \tilde{\psi}_{n-j}(x)  
%	&  =  \frac1{Q_n} e^{n(-\frac12V(x)+ax)}  + O(e^{n(\Gfn(x;a)-\Hfn(x;a))}) \\
%	&= e^{n(\Gfn(x;a)+\Hfn(x;a)-2\Gfn(\redge;a))/2} (1+o(1)).
%\end{split}
%\end{equation}
%Hence substituting \eqref{eq:estimation_of_tilde_gamma_critical} of $\bfGamma_{n-j}(a)$ into \eqref{eq:crit_tildepsi_far_almost}, we obtain~\eqref{eq:alternative_def_of_s^(1)_far}.

\end{proof}

\subsubsection{Proof}\label{sec:4.2.3}

Recall the outline \ref{emu:sub_method_b2} in Section~\ref{subsection:outline_of_the_proof}.

From Lemma~\ref{lem:tildepsi43} and Proposition~\ref{prop:asy_for_mult_cut}, and using~\eqref{eq:est_of_exp_of_V-av_near_e} to estimate $\frac1{Q_n} e^{n(ax - V(x)/2)} - e^{\alpha\xi}$, 
we obtain
%We consider $\langle \tilde{\psi}_{n-j}, \psi_{n-j} \rangle_{\Int}$ in the same way we calculated $\langle e^{n(ay-V(y)/2)}, \psi_{n-j}(y) \rangle_{(\redge, \infty)}$ in \eqref{eq:integral_over_linear_critical}, so that we divide
%\begin{multline} \label{eq:tri_divide_of_inner_prod_psi}
%\langle \tilde{\psi}_{n-j}, \psi_{n-j} \rangle_{\Int} = \langle \tilde{\psi}_{n-j}, \psi_{n-j} \rangle_{(\redge+\beta^{-1}n^{-2/3}T, \redge+n^{-11/21})} \\
%+ \langle \tilde{\psi}_{n-j}, \psi_{n-j} \rangle_{(\redge+n^{-11/21}, \redge+\epsilon/2)} + \langle \tilde{\psi}_{n-j}, \psi_{n-j} \rangle_{(\redge+\epsilon/2, \infty)}.
%\end{multline}
%Using formulas \eqref{eq:estimation_of_gamma_psi_near} and \eqref{eq:psi_general}, and asymptotics \eqref{eq:est_of_Ai_Sigma_horizon}, \eqref{eq:est_of_Ai_prime_Sigma_horizon} and \eqref{eq:est_of_exp_of_V-av_near_e} in the evaluation of the first two inner products on the right hand side of \eqref{eq:tri_divide_of_inner_prod_psi}, and the asymptotics \eqref{eq:alternative_def_of_s^(1)_far} and \eqref{eq:psin101} in the evaluation of the third inner product on the right hand side of \eqref{eq:tri_divide_of_inner_prod_psi}, we obtain
\begin{equation}\label{eq:p1}
\langle \tilde{\psi}_{n-j}, \psi_{n-j} \rangle_{\Int} \to \int_{T}^\infty C_{-\alpha}(\xi) \Ai(\xi) d\xi.
\end{equation}

\bigskip

Now we evaluate $u_{j,n}(\xi): =\frac1{\sqrt{n}} (K_{n-j,n} \chi_{\Int}\tilde{\psi}_{n-j})(\redge+\beta^{-1}n^{-2/3}\xi)$ defined in \eqref{eq:defu} asymptotically in $L^2([T, \infty))$. 
Using Lemma~\ref{lem:tildepsi43}, and then estimating as in Subsubsection \ref{sec:4.2.2}, we obtain 
\begin{equation}
	\int_{\Int} K_{n-j,n}(x,y) \tilde{\psi}_{n-j} dy = \frac{\beta\sqrt{n}}{\B_{j,n}(\redge)} (K_{\Airy}\chi_{[T,\infty)}C_{-\alpha})(\xi) + O(n^{-1/3}e^{-\frac{1}{2} \lvert \xi \rvert^{3/2}})
\end{equation}
for $x=\redge + \beta n^{2/3}\xi \in E_{T, \epsilon/2}$. 
For $x\ge \redge+\epsilon/2$, 
\begin{equation}
\int_{\Int} K_{n-j,n}(x,y) \tilde{\psi}_{n-j} dy = O((1 + \lvert x \rvert)^{-j} e^{n(\Gfn(x;a)-\Hfn(x;a))/2}).
\end{equation}
The calculation is similar to Subsubsection \ref{sec:4.2.2} and we skip the details. 
Thus 
\begin{equation}\label{eq:p2}
	u_{j,n} -\frac{\beta}{\B_{j,n}(\redge)} (K_{\Airy} \chi_{[T,\infty)} C_{-\alpha}) \to 0, 
	\qquad \text{in $L^2[T,\infty)$.}
\end{equation}

Therefore, from~\eqref{eq:maind3}, 
\begin{equation}\label{eq:trickyfinal}
\begin{split}
	\Prob_{n-j+1, n}(a; \Int) 
	&=  \FGUE(T)\cdot \big(1- \langle C_{-\alpha}, \Ai \rangle_{[T,\infty)} \\
	&\quad	- \langle (1-\chi_{[T,\infty)} K_{\Airy} \chi_{[T,\infty)} )^{-1}  K_{\Airy} \chi_{[T,\infty)}C_{-\alpha}, \Ai \rangle_{[T,\infty)} \big)\\
	&=\FGOE(T; -\alpha).  
\end{split}
\end{equation}
Hence Theorem~\ref{thm:convex}\ref{enu:thm:convex:b} and Theorem~\ref{thm:critical_traditional_split}\ref{enu:thm:critical_traditional_split:a} are proved. 

%%%%%%%%%%

\subsection{Proof of Theorem~\ref{thm:rank_1_transit_crit}} \label{subsection:Proof_of_Theorem1.4}
%When $\acc = \frac{1}{2}V'(\redge)$ and $\acc \not\in \mathcal{J}_V$}

Note that $\Gfn(x_0(\acc);\acc)=\Gfn(\redge; \acc)>\Gfn(x; \acc)$ for all $x\in (\redge, \infty)\setminus\{x_0(\acc)\}$. 
For $a$ given in either~\eqref{eq:aina} or~\eqref{eq:ainb}, let $x_0(a)$ be the point near $x_0(\acc)$ such that $\Gfn(x;a)$ achieves its local maximum. The point $x_0(a)$ is well defined as long as $\lvert a - \acc \rvert$ is small enough. Note that for $a\ge \acc$, $x_0(a)$ is same as in the definition of $x_0(a)$ in Lemma \ref{lem:x0}. However, for $a<\acc$, $x_0(a)$ is not defined in Lemma \ref{lem:x0}. We extend the definition of $x_0(a)$ here for $a<\acc$ when $a-\acc$ is small enough.
%However, we should note that the definition of $x_0(a)$ here is not strictly in the sense of Lemma \ref{lem:x0} if $a < \acc$.

\subsubsection{Proof of Theorem~\ref{thm:rank_1_transit_crit}\ref{enu:thm:rank_1_transit_crit:a}}\label{sec:4.3.1}

We consider the double scaling situation
\begin{equation} \label{eq:two_double_scalings}
	a = \acc + \frac{\beta \alpha}{n^{1/3}} 
\end{equation}
where $a$ is in a compact subset of $(-\infty,0)$. Since we assume $G''(x_0(\acc);\acc)\neq 0$, we have $x_0(a) -x_0(\acc) =O(|a-\acc|)= O(n^{-1/3})$. We also have, as in~\eqref{eq:fluctuation_of_G(x_0,a)}, using $\frac{\partial}{\partial a} \Gfn(x;a)=x$, 
\begin{equation}\label{eq:tq2}
\begin{split}
	\Gfn(x_0(a);a) &= \Gfn(x_0(\acc);\acc) +  x_0(\acc) (a-\acc)+ O(|a-\acc|^2) \\
	 &= \Gfn(x_0(\acc);\acc) + \frac{\beta\alpha}{n^{1/3}} x_0(\acc)+ O(|a-\acc|^2).
\end{split}
\end{equation}
Hence since $\Gfn(\redge;a) = \Gfn(\redge;\acc) + (a-\acc)\redge$ by the definition of $\Gfn$, 
\begin{equation}\label{eq:tq3}
	\Gfn(x_0(a);a) - \Gfn(\redge; a) = \frac{\beta\alpha}{n^{1/3}} (x_0(\acc)-\redge)+O(n^{-2/3}).
\end{equation}

We first evaluate $\bfGamma_{n-j}(a)$ as in Lemma~\ref{lem:GammaBQQ}. 
Note that in Subsection~\ref{sec:4.2.1}, we used properties of $\Hfn(z;a)$ for the integrals over $\Sigma_{\pm}$ and properties of $\Gfn(x;a)$ for the integral over $(\redge, \infty)$. Since there is no change in the properties of $\Hfn(z;a)$,   the integrals over $\Sigma_+$ and $\Sigma_-$ are computed exactly the same as given by~\eqref{eq:estimation_of_tilde_gamma_upper} and~\eqref{eq:estimation_of_tilde_gamma_lower}. For the integral over $(\redge,\infty)$, note that the main contribution to~\eqref{eq:integral_over_linear_critical} was from the part of $y$ near $\redge$ since $\Gfn(y;a)$ takes its maximum for $y$ near $\redge$. However, now due to~\eqref{eq:tq3} we need to add a contribution from $y$ near $x_0(a)$. 
By using the standard Laplace's method as in~\eqref{eq:tilde_gamma_case_1_critical_II}, the contribution to the integral near $x_0(a)$ equals 
\begin{equation}\label{eq:contfromx0}
	\sqrt{\frac{2\pi}{-n\Gfn''(x_0(\acc); \acc))}} \M_{j,n}(x_0(\acc))  e^{n(\Gfn(x_0(a);a)-\ell/2)} (1+o(1)),
\end{equation}
Adding~\eqref{eq:contfromx0} and~\eqref{eq:integral_over_linear_critical}, 
%Hence, as $\Gfn(\redge(\acc);\acc)=\Gfn(x_0(\acc);\acc)$, 
\begin{multline}\label{eq:tq1}
	\int_{\redge}^\infty \varphi_{n-j}(y)e^{nay}dy 
	= \frac{Q_n}{\beta\sqrt{n}} \bigg( \B_{j,n}(\redge) \int_0^\infty \Ai(\xi) e^{\alpha \xi} d\xi  \\
	+ \beta \sqrt{\frac{2\pi}{-\Gfn''(x_0(\acc); \acc)}} \M_{j,n}(x_0(\acc))  e^{n(\Gfn(x_0(a);a)-\Gfn(\redge; a))} +o(1)\bigg).
\end{multline}
But~\eqref{eq:tq3} implies that 
\begin{equation}\label{eq:Gfndifftempa}
	n(\Gfn(x_0(a);a) - \Gfn(\redge; a)) = n^{2/3}\beta\alpha(x_0(\acc)-\redge) + O(n^{1/3}) \ll 0
\end{equation}
for all large enough $n$ since $\alpha<0$. Hence we find that 
\begin{equation}\label{eq:tm1}
\begin{split}
	&\int_{\redge}^\infty \varphi_{n-j}(y)e^{nay}dy 
	= \frac{Q_n}{\beta\sqrt{n}} \bigg( \B_{j,n}(\redge) \int_0^\infty \Ai(\xi) e^{\alpha \xi} d\xi +o(1)\bigg),
\end{split}
\end{equation}
as in Lemma~\ref{lem:GammaBQQ}.
Adding the integrals on the contours $\Sigma_{\pm}$, we obtain 
\begin{equation}
	\bfGamma_{n-j}(a) =  \frac{Q_n}{\beta\sqrt{n}} e^{\alpha^3/3} (\B_{j,n}(\redge)+o(1)),
\end{equation}
which is same as~\eqref{eq:estimation_of_tilde_gamma_critical}.

\bigskip

The evaluation of $\tilde{\psi}_{n-j}(x)$ is similar. We use the formula \eqref{eq:expression_of_gamma_psi_critical} as in Lemma~\ref{lem:tildepsi43}. 
The  evaluation of the integral over $(\Sigma_++C_n)\cup(\Sigma_-+C_n)$ is exactly the same as in Subsubsection \ref{sec:4.2.2}. For the evaluation of the integral over $(\redge+C_n, \infty)$ when $x \in E_{T,\epsilon/2}$, we find that \eqref{eq:DDD4:1} is unchanged. For the integral over $(\redge+\epsilon, \infty)$,   the contribution near $y=x_0(a)$ implies that \eqref{eq:DDD4:2} becomes
\begin{equation} \label{eq:1.4(a)_near_weakened}
\begin{split}
	&\int_{\redge+\epsilon}^{\infty} K_{n-j,n}(x,y) e^{n(ay-\frac12 V(y))} dy  = Q_n O(n^{1/6}e^{-\factor \lvert \xi \rvert^{3/2}}  \frac{e^{n(\Gfn(x_0(a);a)-\Gfn(\redge;a))}}{\sqrt{n}}).
\end{split}
\end{equation}
However, due to~\eqref{eq:Gfndifftempa}, this is again $Q_n O(n^{-1/3}e^{-\factor \lvert \xi \rvert^{3/2}})$ as in~\eqref{eq:DDD4:2} .
Therefore the result \eqref{eq:estimation_of_gamma_psi_near} still holds for $x\in E_{T,\epsilon/2}$. For $x \in (\redge+\epsilon/2, \infty)$, the estimates \eqref{eq:1DDE3}  and \eqref{eq:1DDD3}  hold without any change. Moreover, it is straightforward to check that \eqref{eq:1DDD01} still holds.
%\begin{multline} \label{eq:1.4(a)_far_weakened}
%frac{1}{Q_n} \int_{\redge+\epsilon/2}^\infty K_{n-j,n}(x,y) e^{n(ay-\frac12 V(y))}  dy = \\
%O(n^{-1/2} e^{-\bar{\epsilon}n^{2/3}} (1 + \lvert x \rvert^{-j}) e^{n(\Gfn(x;a)-\Hfn(x;a))/2})
%\end{multline}
%by the same reason as for \eqref{eq:1.4(a)_near_weakened}, and the $\bar{\epsilon}$ is also the same as in \eqref{eq:1.4(a)_near_weakened}. 
Therefore \eqref{eq:alternative_def_of_s^(1)_far} holds for $x\ge \redge+\epsilon/2$.
Therefore, Lemma~\ref{lem:tildepsi43} holds without any changes. 

\bigskip

Now we proceed as in Subsection~\ref{sec:4.2.3}.
In evaluating $\langle \tilde{\psi}_{n-j}, \psi_{n-j} \rangle_{\Int}$, the integral over $(\redge+\epsilon/2)$ becomes, due to the contribution near $x_0(a)$, 
\begin{equation}\label{eq:tq5}
\begin{split}
	& \langle \tilde{\psi}_{n-j}, \psi_{n-j} \rangle_{(\redge+\epsilon, \infty)}
	= \int_{\redge+\epsilon}^{\infty} O(\sqrt{n} e^{n(\Gfn(x;a)-\Gfn(\redge; a))}) dx = O (e^{n(\Gfn(x_0(a);a)-\Gfn(\redge; a))} ) \to 0
\end{split}
\end{equation}
by~\eqref{eq:Gfndifftempa}. This implies that~\eqref{eq:p1} holds without a change.
Similarly, it is straightforward to check that~\eqref{eq:p2} holds. Therefore, we obtain~\eqref{eq:trickyfinal} 
and Theorem~\ref{thm:rank_1_transit_crit}\ref{enu:thm:rank_1_transit_crit:a} is proved.

\subsubsection{Proof of Theorem~\ref{thm:rank_1_transit_crit}\ref{enu:thm:rank_1_transit_crit:b}}

We consider the  double scaling situation
\begin{equation} \label{eq:two_double_scalings02}
	a  = \acc + \frac{\alpha'}{n}
\end{equation}
where $\alpha'$ is in a compact subset of $\R$. 
In this case,~\eqref{eq:tq2} and~\eqref{eq:tq3} are changed to 
\begin{equation} \label{eq:tq6_pre}
	\Gfn(x_0(a); a)= \Gfn(x_0(\acc); \acc) + x_0(\acc) \frac{\alpha'}{n} + O(n^{-2}),
\end{equation}
and 
\begin{equation}\label{eq:tq6}
	\Gfn(x_0(a);a) -\Gfn(\redge; a)
	=  \frac{\alpha'}{n}  (x_0(\acc) -\redge) + O(n^{-2}).
\end{equation}

First we consider $\bfGamma_{n-j}(a)$. There are two changes from the previous subsubsection. The first is that since $\alpha$ defined in \eqref{eq:two_double_scalings} and $\alpha'$ defined in \eqref{eq:two_double_scalings02} are related as $\alpha=\beta^{-1}n^{-2/3}\alpha'$, we have $\alpha\to 0$ and hence in~\eqref{eq:estimation_of_tilde_gamma_upper}, \eqref{eq:estimation_of_tilde_gamma_lower} and~\eqref{eq:tq1}, we have $\alpha=0$ in the integrals involving the Airy function. The second is that~\eqref{eq:tm1} does not follows from~\eqref{eq:tq1} since~\eqref{eq:Gfndifftempa} no longer holds. Instead, due to~\eqref{eq:tq6},~\eqref{eq:tq1} implies that 
\begin{equation}
\begin{split}
	&\int_{\redge}^\infty \varphi_{n-j}(y)e^{nay}dy 
	= \frac{Q_n}{\beta\sqrt{n}} \bigg( \B_{j,n}(\redge) \int_0^\infty \Ai(\xi) d\xi  \\
	&\quad + \beta \sqrt{\frac{2\pi}{-\Gfn''(x_0(\acc); \acc))}} \M_{j,n}(x_0(\acc)) e^{\alpha'(x_0(\acc)-\redge)} +o(1)\bigg).
\end{split}
\end{equation}
Hence adding~\eqref{eq:estimation_of_tilde_gamma_upper} and~\eqref{eq:estimation_of_tilde_gamma_lower} (with $\alpha=0$), 
we obtain 
\begin{equation}\label{eq:tq8}
	\bfGamma_{n-j}(a) = \frac{Q_n}{\beta \sqrt{n}} \B_{j,n}(\redge) \left( \frac{D_0 + D_1(\alpha')}{D_0}+o(1) \right),
\end{equation}
where
\begin{equation}\label{eq:tq9}
	D_0(\alpha') = \frac{\B_{j,n}(\redge)}{\beta}, \quad D_1(\alpha') = \sqrt{\frac{2\pi}{-\Gfn''(x_0(\acc);\acc)}} \M_{j,n}(x_0(\acc)) e^{ \alpha'(x_0(\acc)-\redge)}.
\end{equation}

\bigskip

Now consider $\tilde{\psi}_{n-j}(x)$. Like in Subsubsection \ref{sec:4.3.1}, most of the estimates of Subsubsection \ref{sec:4.2.2} remain the same. The only changes are the contribution near $x_0(\acc)$. Substituting \eqref{eq:tq6} into \eqref{eq:1.4(a)_near_weakened} with $a = \beta^{-1}n^{-2/3}\alpha'$, we have 
\begin{equation}
\int_{\redge+\epsilon}^{\infty} K_{n-j,n}(x,y) e^{n(ay-\frac12 V(y))} dy = Q_n O(n^{-1/3} e^{-\factor \lvert \xi \rvert^{3/2}})
\end{equation}
for $x\in E_{T, \epsilon/2}$. This is similar to  \eqref{eq:DDD4:2} except that $e^{-\frac{\bar{\epsilon}'}2n}$ is replaced by $n^{-1/3}$. However, it is easy to check that this bound is enough for the rest of the analysis.  
The rest of the analysis  of  Subsubsection~\ref{sec:4.2.2} continues without changes and we obtain the asymptotics~\eqref{eq:Gt1} of $\bfGamma_{n-j}(a)\tilde{\psi}_{n-j}(x)$ for $x\in E_{T, \epsilon/2}$ and~\eqref{eq:alternative_def_of_s^(1)_far01} for $x\ge \redge+\epsilon/2$. 
Using~\eqref{eq:tq8}, we find that  
\begin{equation} \label{eq:tq10r}
\begin{split}
	\tilde{\psi}_{n-j}(x) 
	&= p_{j,n}(\alpha')\frac{\beta \sqrt{n}}{\B_{j,n}(\redge)}  \left[ C_{0}(\xi) + e^{-\alpha^3/3} \bigg( \frac1{Q_n} e^{an - \frac12V(x)} - e^{\alpha\xi}\bigg) +o(1) \right] 
\end{split}
\end{equation}
for $x\in E_{T, \epsilon/2}$ and 
\begin{equation} \label{eq:tq11}
\begin{split}
	\tilde{\psi}_{n-j}(x)  
	&= p_{j,n}(\alpha')\frac{\beta \sqrt{n}}{\B_{j,n}(\redge)} e^{n(\Gfn(x;a)+\Hfn(x;a)-2\Gfn(\redge;a))/2} (1+o(1))
\end{split}
\end{equation}
for $x\ge \redge+\epsilon/2$,
where
\begin{equation}\label{eq:tq30}
	p_{j,n}(\alpha'):= \frac{D_0}{D_0+D_1(\alpha')}.
\end{equation}
The formulas \eqref{eq:tq10r} and \eqref{eq:tq11} are different from \eqref{eq:estimation_of_gamma_psi_near} and \eqref{eq:alternative_def_of_s^(1)_far} only by the factor $p_{j,n}(\alpha')$.

\bigskip

We now prove the theorem. First, consider~\eqref{eq:ppaa}. From~\eqref{eq:tq11} and~\eqref{eq:psin101}, we obtain
\begin{equation}\label{eq:tq14}
\begin{split}
	\langle \tilde{\psi}_{n-j}, \psi_{n-j} \rangle_{\Intx(x_0(\acc))} %\\
	= & p_{j,n}(\alpha')\frac{\beta \sqrt{n}}{\B_{j,n}(\redge)}  \int_{\Intx(x_0(\acc))}
	M_{j,n}(x) e^{n(\Gfn(x;a)-\Gfn(\redge;a))}(1+o(1)) dx \\
	= & p_{j,n}(\alpha')\frac{\beta \sqrt{n}}{\B_{j,n}(\redge)} \M_{j,n}(x_0(\acc)) \frac{e^{n(\Gfn(x_0(a);a)-\Gfn(\redge;a))}}{\sqrt{-n\Gfn''(x_0(a);a)/2\pi}}\int_T^\infty e^{-\frac12 \xi^2}d\xi (1+o(1)) \\
	= & (1-p_{j,n}(\alpha')) (1-\erf(T))(1+o(1)),
\end{split}
\end{equation}
by using \eqref{eq:tq6_pre} and \eqref{eq:tq6}.
Also, for $x\in \Intx(x_0(\acc))$, by \eqref{eq:lemma_enu:estimate_K_n-1,n:1} and \eqref{eq:tq11}
\begin{equation}\label{eq:tq15}
\begin{split}
	(K_{n-j,n}\chi_{\Intx(x_0(\acc))} \tilde{\psi}_{n-j})(x)
	= &\int_{\Intx(x_0(\acc))} K_{n-j,n}(x,y)\tilde{\psi}_{n-j}(y)dy \\
	= & O(e^{n(\Gfn(x;a)-\Hfn(x;a))/2} )\int_{\Intx(x_0(\acc))} O(\sqrt{n}e^{n(\Gfn(y;a)-\Gfn(\redge;a))})dy \\
	= & O(e^{n(\Gfn(x;a)-\Hfn(x;a))/2} ).
\end{split}
\end{equation}
Hence $K_{n-j,n}\chi_{\Intx(x_0(\acc))} \tilde{\psi}_{n-j}\to 0$
in $L^2(\Intx(x_0(\acc)))$. Therefore, we find from~\eqref{eq:maind2} that  
\begin{equation}\label{eq:tq16}
\begin{split}
	\Prob_{n-j+1, n}(a; \Intx(x_0(\acc))) = 1 - (1-p_{j,n}(\alpha'))(1 - \erf(T))(1+o(1))
\end{split}
\end{equation}
and~\eqref{eq:ppaa} is proved. 

Second, consider~\eqref{eq:ppaarr}. We proceed as in 
the outline \ref{emu:sub_method_b2} in Section~\ref{subsection:outline_of_the_proof}. We first evaluate $\langle \tilde{\psi}_{n-j}, \psi_{n-j} \rangle_{\Int}=\langle \tilde{\psi}_{n-j}, \psi_{n-j}\rangle_{E_{T, \epsilon/2}}+\langle \tilde{\psi}_{n-j}, \psi_{n-j} \rangle_{(\redge+\epsilon/2,\infty)}$. 
For second term, a computation as in~\eqref{eq:tq14} yields that 
\begin{equation}\label{eq:tq13}
	\langle \tilde{\psi}_{n-j}, \psi_{n-j} \rangle_{(\redge+\epsilon/2, \infty)}=  (1-p_{j,n}(\alpha') )(1+o(1)).
\end{equation}
The first term is calculated as in Subsubsection~\ref{sec:4.3.1} with the only change that the prefactor $p_{j,n}(\alpha')$ is multiplied: 
\begin{equation}\label{eq:tq20}
	\langle \tilde{\psi}_{n-j}, \psi_{n-j} \rangle_{E_{T, \epsilon/2}}=  p_{j,n}(\alpha') \int_T^\infty C_{0}(\xi)\Ai(\xi)d\xi (1+o(1)).
\end{equation}
Hence
\begin{equation}\label{eq:tq21}
	\langle \tilde{\psi}_{n-j}, \psi_{n-j} \rangle_{\Int}=  p_{j,n}(\alpha') \int_T^\infty C_{0}(\xi)\Ai(\xi)d\xi 
	+  (1-p_{j,n}(\alpha') )+o(1).
\end{equation}
Now we evaluate $u_{j,n}(\xi) := n^{-1/2} K_{n-j,n}\chi_{\Int}\tilde{\psi}_{n-j}(x)$, $x = \redge + \beta^{-1}n^{-2/3}\xi \in \Int$. But it is straightforward to check that, as in Subsubsection~\ref{sec:4.3.1}, 
\begin{equation}\label{eq:tq24}
\begin{split}
	u_{j,n}
	- p_{j,n}(\alpha')\frac{\beta}{\B_{j,n}(\redge)} (K_{\Airy}\chi_{[T,\infty)} C_{0})
	\to 0 \quad \text{in $L^2([T, \infty))$.}
\end{split}
\end{equation}
Therefore, from~\eqref{eq:maind3}, 
\begin{equation}\label{eq:tq25}
\begin{split}
	\Prob_{n-j+1, n}(a; \Int) %\\
	= & \FGUE(T) \cdot \big( 1- (p_{j,n}(\alpha')\langle C_0, \Ai \rangle_{[T, \infty)} + 1-p_{j,n}(\alpha')) \\
	& - p_{j,n}(\alpha') \langle (1-\chi_{[T, \infty)}K_{\Airy}\chi_{[T, \infty)})^{-1} (K_{\Airy}\chi_{[T, \infty)}C_0), \Ai \rangle_{[T,\infty)} \big) +o(1)\\
	= & p_{j,n}(\alpha')\FGOE(T; 0)+o(1).
\end{split}
\end{equation}
Thus~\eqref{eq:ppaarr} is proved.

%%%%%%%%

\section{Summary of asymptotics of orthogonal polynomials and the Christoffel-Darboux kernel} \label{section:Result_of_RHP}

Define the matrix-valued function 
\begin{equation} \label{eq:standard_RHP}
Y_k(z;n) = 
\begin{pmatrix}
\gamma^{-1}_k(n) p_k(z;n) & \gamma^{-1}_k(n) (C\varphi_k)(z;n)  \\
-2\pi i \gamma_{k-1}(n) p_{k-1}(z;n) & -2\pi i \gamma_{k-1}(n) (C\varphi_{k-1})(z;n)
\end{pmatrix},
\end{equation}
for $z \in \compC \setminus \realR$. %Here $p_k(x;n)$ is the $k$th orthonormal polynomial with respect to the varying weight $e^{-nV(x)}$ defined in Introduction. The positive constant $\gamma_k = \gamma_k(n)$ is the leading coefficient of $p_k(x;n)$. The function $\varphi_k(x;n)=p_k(x;n)e^{-nV(x)}$ and 
Here $(C\varphi_k)(z;n)$ is the Cauchy transform of $\varphi_k(x;n) = p_k(x;n)e^{-nV(x)}$. 
The matrix $Y_{k}(z;n)$ is the solution to the Riemann-Hilbert problem for orthogonal polynomials with weight $e^{-nV(z)}$ (see \cite{Fokas-Its-Kitaev92}). We are interested in the asymptotics of $Y_{n-j}(z;n)$ when $j=O(1)$ and $n\to\infty$. We indicate the changes from the analysis of \cite{Deift-Kriecherbauer-McLaughlin-Venakides-Zhou99} and state the results. 
A similar derivation for the discrete weight can be found in \cite{Baik-Kriecherbauer-McLaughlin-Miller07}.

Fix $\delta>0$ small enough. Let (see Figure~\ref{fig:AB_delta})
\begin{align}
A_{\delta} := & \{ z \in \compC \mid \lvert z-\redge \rvert < \delta \}, \\
B_{\delta} := & \{ z \in \compC \mid \Re z \geq \redge \textnormal{ and } \lvert z - \redge \rvert > \delta \} %\notag \\
 \cup \{ z \in \compC \mid \Re z < \redge \textnormal{ and } \lvert \Im z \rvert > \delta \},
\end{align}
where $\redge$ is the rightmost  end-point of the support of the equilibrium measure. %See Figure \ref{fig:A_delta_and_B_delta}. 
Comparing with notations in \cite{Deift-Kriecherbauer-McLaughlin-Venakides-Zhou99}, $A_{\delta}$ is the circle $D_{\epsilon, a_{N+1}}$ with $\delta$ corresponding to the radius $\epsilon$. in \cite[Figure 1.4]{Deift-Kriecherbauer-McLaughlin-Venakides-Zhou99}. As in \cite[Figure 1.4]{Deift-Kriecherbauer-McLaughlin-Venakides-Zhou99}, $A_{\delta}$ is divided into four regions I, II, III and IV. Let $\Sigma_R$ be the contour in \cite[Figure 4.9]{Deift-Kriecherbauer-McLaughlin-Venakides-Zhou99}. We assume that the boundary of $A_{\delta}$ is a part of $\Sigma_R$ and $B_{\delta}$ is outside of the lens-shaped regions, \cf\ \cite[Formula (4.116)]{Deift-Kriecherbauer-McLaughlin-Venakides-Zhou99}.
\begin{figure}[htp] \label{fig:A_delta_and_B_delta}
\begin{center}
\includegraphics{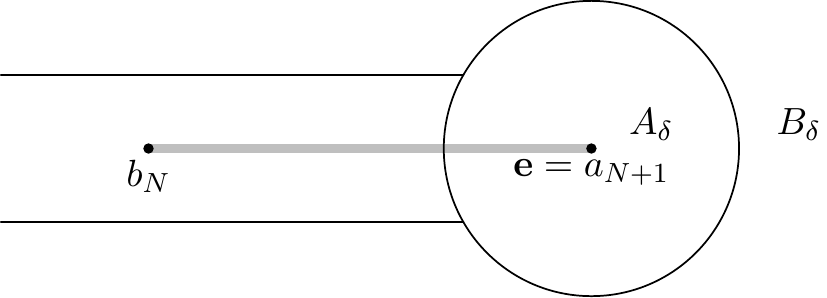}
\caption{$A_\delta$ and $B_{\delta}$.}\label{fig:AB_delta}
\end{center}
\end{figure}

Several notations from  \cite{Deift-Kriecherbauer-McLaughlin-Venakides-Zhou99} are used in this section, and we summarize  them in Table \ref{table:notations}. Other notations may be slightly different but should be clear.

%In this section we adapt lots of notations from \cite{Deift-Kriecherbauer-McLaughlin-Venakides-Zhou99}. We note that the notation $\gfn(z)$ in our paper is the $g(z)$ defined in \cite[Formula (1.18)]{Deift-Kriecherbauer-McLaughlin-Venakides-Zhou99} and $\gamma_k = \gamma_k(n)$ in our paper is the $\gamma^{(n)}_k$ defined above \cite[Formula (1.3)]{Deift-Kriecherbauer-McLaughlin-Venakides-Zhou99}. There are other notations defined in our paper that appears in \cite{Deift-Kriecherbauer-McLaughlin-Venakides-Zhou99} with the same meaning, like $N$, $a_l$, $b_l$, $\ell$, etc. There are notations we borrow from \cite{Deift-Kriecherbauer-McLaughlin-Venakides-Zhou99}, as shown in Table \ref{table:notations}.

\begin{table}[h] 
\centering
\begin{tabular}{|ll|} 
\hline
Notation & Definition in \cite{Deift-Kriecherbauer-McLaughlin-Venakides-Zhou99} \\
\hline \hline
$\vec{\Omega}$ & defined in Formula (1.21) \\
$\theta$ & defined in Formula (1.24) \\
$\gamma$ & defined in Formula (1.26) \\
$u$ & defined in Formula (1.29) \\
$d$ & defined in Formula (1.30) \\
$u_+(\infty)$ & explained below Theorem 1.1 \\
$\Phi = \Phi_{a_{N+1}}$ & defined in Formula (1.34) \\
$P$ & defined in Formulas (1.38)--(1.40) \\
$\sigma_3$ & mentioned between Formulas (1.40) and (1.41) \\
$\Sigma^{(1)}$ & shown in Figure 1.5 \\
$v^{(\infty)}$ & defined in Formulas (1.104)--(1.107) \\
$G$ & defined in Formula (3.44) \\
$\Sigma_R$ & shown in Figure 4.9 \\
$v_R$ & defined in Formula (4.108) \\
\hline
\end{tabular}
\caption{Notations taken from \cite{Deift-Kriecherbauer-McLaughlin-Venakides-Zhou99}}\label{table:notations}
\end{table}

By following the procedure of \cite{Deift-Kriecherbauer-McLaughlin-Venakides-Zhou99}, we find asymptotics of $Y$. Noting the symmetry
\begin{equation}
	Y(z)= \begin{pmatrix} 1&0\\0&-1\end{pmatrix} \overline{Y(\bar{z})} \begin{pmatrix} 1&0\\0&-1\end{pmatrix} ,
\end{equation}
it is enough to consider $z\in \compC_+$. 
The asymptotic  formulas  of $Y(z)$ are different in $B_\delta$ and $A_\delta$. 
For $z \in B_{\delta} \cap \compC_+$, we have (\cf\ \cite[Formula 4.116]{Deift-Kriecherbauer-McLaughlin-Venakides-Zhou99})
\begin{equation} \label{eq:RHP_out_of_lens}
Y_{n-j}(z;n) = e^{\frac{n\ell}{2} \sigma_3} R_{j,n}(z) M^{(\infty)}_{j,n}(z) e^{n(\gfn(z)-\frac{\ell}{2}) \sigma_3}.
\end{equation}
The remainder 
$R_{j,n}(z)$ solves the Riemann-Hilbert problem similar to the Riemann-Hilbert problem for $R$ in \cite[Subsection 4.6]{Deift-Kriecherbauer-McLaughlin-Venakides-Zhou99} (\cf\ \cite[Formulas (4.106)--(4.108)]{Deift-Kriecherbauer-McLaughlin-Venakides-Zhou99}) with the jump matrix which has the same estimate as in $v_R$ shown in \cite[Figure 1.4]{Deift-Kriecherbauer-McLaughlin-Venakides-Zhou99}. This implies that uniformly for all $z \in \compC \setminus \Sigma_R$ and $n$ (\cf\ \cite[Formula 4.115]{Deift-Kriecherbauer-McLaughlin-Venakides-Zhou99})
\begin{equation} \label{eq:leading_term_of_R_j}
R_{j,n}(z) = I + O(n^{-1}).
\end{equation}

The outer parametrix  $M^{(\infty)}_{j,n}(z)$ solves the Riemann-Hilbert problem (\cf\ \cite[Formulas (4.24)--(4.26)]{Deift-Kriecherbauer-McLaughlin-Venakides-Zhou99})
\begin{align}
M^{(\infty)}_{j,n}(z) & \textnormal{ is analytic in $\compC \setminus \Sigma^{(1)}$} \\
(M^{(\infty)}_{j,n})_+(z) & = (M^{(\infty)}_{j,n})_-(z) v^{(\infty)}(z), \quad z \in \Sigma^{(1)}, \\
M^{(\infty)}_{j,n}(z) & = \left( I + O \left( \frac{1}{z} \right) \right)
\begin{pmatrix}
z^{-j} & 0 \\
0 & z^j
\end{pmatrix},
\quad z \to \infty, \quad z \in \compC \setminus \realR.
\end{align}
Note that the dependence on $j$ in the asymptotics as $z\to\infty$. 
The solution of this Riemann-Hilbert problem can be solved as in \cite[Lemma 4.3]{Deift-Kriecherbauer-McLaughlin-Venakides-Zhou99}).
Setting
\begin{align}
\Theta(\infty) := & \frac{\theta(u_+(\infty)+d)}{\theta(u_+(\infty)-d)}, \\
\tilde{\Theta}(z) := & \frac{\gamma(z)-\gamma(z)^{-1}}{\gamma(z)+\gamma(z)^{-1}} \frac{\theta(u(z)+d)}{\theta(u(z)-d)},
\end{align}
it is straightforward to check that  for $z \in B_{\delta} \cap \compC_+$, 
\begin{equation} \label{eq:M^infty_j_above}
\begin{split}
M^{(\infty)}_{j,n}(z) = &\diag \left( \left( \frac{\sum^N_{l=0} a_{l+1}-b_l}{4} \right)^{-j}, \left( \frac{\sum^N_{l=0} a_{l+1}-b_l}{4} \right)^j \right) \\
&\times \diag \left( \frac{\Theta(\infty)^{-j} \theta(u_+(\infty)+d)}{\theta(u_+(\infty) - \frac{n}{2\pi}\Q - (2j-1)d)}, \frac{\Theta(\infty)^j \theta(u_+(\infty)+d)}{\theta(u_+(\infty) + \frac{n}{2\pi}\Q + (2j+1)d)} \right) \\
&\times
\begin{pmatrix}
	\tilde{\Theta}(z)^j \frac{\gamma+\gamma^{-1}}{2} \frac{\theta(u(z) - \frac{n}{2\pi}\Q - (2j-1)d)}{\theta(u(z)+d)} & \tilde{\Theta}(z)^{-j} \frac{\gamma-\gamma^{-1}}{-2i} \frac{\theta(u(z) + \frac{n}{2\pi}\Q + (2j-1)d)}{\theta(u(z)-d)} \\
	\tilde{\Theta}(z)^j \frac{\gamma-\gamma^{-1}}{2i} \frac{\theta(u(z) - \frac{n}{2\pi}\Q - (2j+1)d)}{\theta(u(z)-d)} & \tilde{\Theta}(z)^{-j} \frac{\gamma+\gamma^{-1}}{2} \frac{\theta(u(z) + \frac{n}{2\pi}\Q + (2j+1)d)}{\theta(u(z)+d)}
\end{pmatrix},
\end{split}
\end{equation}
and 
for $z \in B_{\delta} \cap \compC_-$,
\begin{equation} \label{eq:M^infty_j_below}
\begin{split}
M^{(\infty)}_{j,n}(z) = & \left( \left( \frac{\sum^N_{l=0} a_{l+1}-b_l}{4} \right)^{-j}, \left( \frac{\sum^N_{l=0} a_{l+1}-b_l}{4} \right)^j \right)  \\
& \times \diag \left( \frac{\Theta(\infty)^{-j} \theta(u_+(\infty)+d)}{\theta(u_+(\infty) - \frac{n}{2\pi}\Q - (2j-1)d)}, \frac{\Theta(\infty)^j \theta(u_+(\infty)+d)}{\theta(u_+(\infty) + \frac{n}{2\pi}\Q + (2j+1)d)} \right) \\
& \times
\begin{pmatrix}
\tilde{\Theta}(z)^{-j} \frac{\gamma-\gamma^{-1}}{-2i} \frac{\theta(u(z) + \frac{n}{2\pi}\Q + (2j-1)d)}{\theta(u(z)-d)} & \tilde{\Theta}(z)^j \frac{\gamma+\gamma^{-1}}{-2} \frac{\theta(u(z) - \frac{n}{2\pi}\Q - (2j-1)d)}{\theta(u(z)+d)} \\
\tilde{\Theta}(z)^{-j} \frac{\gamma+\gamma^{-1}}{2} \frac{\theta(u(z) + \frac{n}{2\pi}\Q + (2j+1)d)}{\theta(u(z)+d)} & \tilde{\Theta}(z)^j \frac{\gamma-\gamma^{-1}}{-2i} \frac{\theta(u(z) - \frac{n}{2\pi}\Q - (2j+1)d)}{\theta(u(z)-d)}
\end{pmatrix}.
\end{split}
\end{equation}
Note that both formulas \eqref{eq:M^infty_j_above} and \eqref{eq:M^infty_j_below} can be extended to $z \in B_{\delta} \cap \realR$.

The asymptotics~\eqref{eq:RHP_out_of_lens} especially implies the asymptotics of $\gamma_{n-j}$. 
Since (\cf\ \cite[Formulas (3.10) and (3.11)]{Deift-Kriecherbauer-McLaughlin-Venakides-Zhou99a})
\begin{equation}
\gamma_{n-j} = \left( -2\pi i \lim_{\lvert z \rvert \to \infty} z^{-n+j+1} (R_{j,n}(z)M^{(\infty)}_{j,n}(z))_{12} e^{-n(\gfn(z)+\ell)} \right)^{-1/2},
\end{equation}
we have from \eqref{eq:leading_term_of_R_j} and \eqref{eq:M^infty_j_above} that 
\begin{equation}
\gamma_{n-j} = \bfgamma_{n-j}(1+O(n^{-1})),
\end{equation}
where
\begin{multline} \label{eq:formula_of_gamma_n-j}
\bfgamma^2_{n-j} = \left( \frac{\pi \left( \sum^N_{l=0} a_{l+1}-b_l \right) }{2} \frac{\theta(u_+(\infty)+d) \theta(u_+(\infty) + \frac{n}{2\pi}\Q + (2j-1)d)}{\theta(u_+(\infty) - \frac{n}{2\pi}\Q - (2j-1)d) \theta(u_+(\infty)-d)} \right)^{-1} \\
\times \left( \frac{\sum^N_{l=0} a_{l+1}-b_l}{4} \Theta(\infty) \right)^{2j} e^{-n\ell}.
\end{multline}
See \cite[Formulas (1.62) and (1.63)]{Deift-Kriecherbauer-McLaughlin-Venakides-Zhou99} for the cases $j=1$ and $j=0$.

Now consider  $z \in A_{\delta}$. Then the analysis of the local parametrix as in  \cite[Section 4.3]{Deift-Kriecherbauer-McLaughlin-Venakides-Zhou99} implies that (cf. \cite[(4.119)--(4.121)]{Deift-Kriecherbauer-McLaughlin-Venakides-Zhou99} )
\begin{equation} \label{eq:RHP_in_A_!_IV}
Y_{n-j}(z;n) = e^{\frac{n\ell}{2}\sigma_3} R_{j,n}(z)(M_{j,n})_p(z) e^{n(\gfn(z)-\frac{\ell}{2})\sigma_3}
\end{equation}
for $z$ is in regions I and IV in \cite[Figure 1.4]{Deift-Kriecherbauer-McLaughlin-Venakides-Zhou99},
\begin{equation} \label{eq:RHP_in_A_!I}
Y_{n-j}(z;n) = e^{\frac{n\ell}{2}\sigma_3} R_{j,n}(z)(M_{j,n})_p(z) 
\begin{pmatrix}
1 & 0 \\
e^{-nG(z)} & 1
\end{pmatrix}
e^{n(\gfn(z)-\frac{\ell}{2})\sigma_3}
\end{equation}
for  $z$ is in regions II in \cite[Figure 1.4]{Deift-Kriecherbauer-McLaughlin-Venakides-Zhou99}, and
\begin{equation} \label{eq:RHP_in_A_!II}
Y_{n-j}(z;n) = e^{\frac{n\ell}{2}\sigma_3} R_{j,n}(z)(M_{j,n})_p(z) 
\begin{pmatrix}
1 & 0 \\
-e^{nG(z)} & 1
\end{pmatrix}
e^{n(\gfn(z)-\frac{\ell}{2})\sigma_3}
\end{equation}
for  $z$ is in regions III in \cite[Figure 1.4]{Deift-Kriecherbauer-McLaughlin-Venakides-Zhou99}. Here the local parametrix $(M_{j,n})_p$ is given by (\cf\ \cite[Formulas (4.75) and (4.76)]{Deift-Kriecherbauer-McLaughlin-Venakides-Zhou99})
\begin{equation}
(M_{j,n})_p(z) := M^{(\infty)}_{j,n}(z) \frac{1}{\sqrt{2i}}
\begin{pmatrix}
i & -i \\
1 & 1
\end{pmatrix}
(\Phi(z))^{\sigma_3/4} P(\Phi(z)).
\end{equation}
where $\Phi(z)$ denotes $\Phi_{a_{N+1}}(z)$ in \cite{Deift-Kriecherbauer-McLaughlin-Venakides-Zhou99}. We note that by definition
\begin{equation}
\Phi(z)  = \left[ -\frac{3n}{4} (2\gfn(z) - V(z) - \ell) \right]^{2/3}, %\quad z\in \compC\setminus(-\infty, c],
\end{equation}
so that 
\begin{equation} \label{eq:asy_pf_Phi_at_e}
\Phi(z) = \beta n^{2/3} (z-\redge) (1 + O(\lvert z-\redge \rvert)) \quad \textnormal{as $z \to \redge$}
\end{equation}
with $\beta$ defined in~\eqref{eq:definition_of_beta} (see \cite[Equations (1.34), (1.35),  (4.74)]{Deift-Kriecherbauer-McLaughlin-Venakides-Zhou99}).

%%%%%%%%

We now summarize the asymptotics the orthonormal polynomials and their Cauchy transformations. For notational convenience, we denote for $z \in B_{\delta}$
\begin{align}
\M_{j,n}(z) := & \bfgamma_{n-j}e^{\frac{n\ell}{2}}(M^{(\infty)}_{j,n})_{11}(z), \label{eq:def_of_slash_M} \\
\tilde{\M}_{j,n}(z) := & \bfgamma_{n-j}e^{\frac{n\ell}{2}}(M^{(\infty)}_{j,n})_{12}(z),
\label{eq:def_of_slash_Mhat}
\end{align}
and for $z \in A_{\delta}$
\begin{align}
\B_{j,n}(z) := & \sqrt{\pi}n^{-1/6} \bfgamma_{n-j}e^{\frac{n\ell}{2}} \left( (M^{(\infty)}_{j,n})_{11}(z) - i(M^{(\infty)}_{j,n})_{12}(z) \right) \Phi(z)^{1/4}, \label{asy_formula_of_B_jn} \\
\D_{j,n}(z) := & \sqrt{\pi}n^{1/6} \bfgamma_{n-j}e^{\frac{n\ell}{2}} \left( -(M^{(\infty)}_{j,n})_{11}(z) - i(M^{(\infty)}_{j,n})_{12}(z) \right) \Phi(z)^{-1/4}. \label{asy_formula_of_D_jn}
\end{align}
\begin{rmk} \label{rmk:when_N=0}
The formulas \eqref{eq:M^infty_j_above} and \eqref{eq:M^infty_j_below} contain the $N$-variable Theta function $\theta$. If $N=0$, \ie, the equilibrium measure is supported on one interval, then $\theta \equiv 1$, and the expressions of $\bfgamma_{n-j}$, $\M_{j,n}(z)$, $\M_{j,n}(z)$, $\B_{j,n}(z)$ and $\D_{j,n}(z)$ are much simplified. In particular $\M_{j,n}(z)$, $\M_{j,n}(z)$, $\B_{j,n}(z)$ and $\D_{j,n}(z)$ do not depend on $n$ when $N=0$. For example, for $z \in \compC_+$
\begin{align}
\M_{j,n}(z) = & \sqrt{\frac{2}{\pi(a_1-b_0)}} \frac{\gamma + \gamma^{-1}}{2} \left( \frac{\gamma - \gamma^{-1}}{\gamma + \gamma^{-1}} \right)^j, \\
\tilde{\M}_{j,n}(z) = & \sqrt{\frac{2}{\pi(a_1-b_0)}} \frac{\gamma - \gamma^{-1}}{-2i} \left( \frac{\gamma - \gamma^{-1}}{\gamma + \gamma^{-1}} \right)^{-j},
\end{align}
and
\begin{equation}
\B_{j,n}(\redge) = \sqrt{2} (a_1-b_0)^{-1/4}\beta^{1/4}.
\end{equation}
\end{rmk}

\begin{prop} \label{prop:asy_for_mult_cut}
There exists $\delta_0>0$ such that  for each fixed $\delta \in (0, \delta_0]$, the following holds as $n\to\infty$ and $j=O(1)$.

\begin{enumerate}[label=(\alph*)]
\item \label{enu:prop:asy_for_mult_cut:a}
For $z \in B_{\delta}$, 
\begin{equation} \label{eq:asy_of_varphi_n-j}
	\varphi_{n-j}(z;n) = M_{j,n}(z) e^{n(\gfn(z)-V(z)-\ell/2)},
\end{equation}
where $M_{j,n}(z)$ is an analytic function in $B_{\delta}$ and 
\begin{equation} \label{eq:relation_of_M_and_slash_M}
M_{j,n}(z) = \M_{j,n}(z) (1+O(n^{-1}))
\end{equation}
uniformly in $z$ and $n$. The function $M_{j,n}$ satisfies that (i) $M_{j,n}(z) = O(z^{-j})$ uniformly in $n$ as $z\to\infty$, (ii)  in any compact subset $K\in B_{\delta}$, $M_{j,n}(z)$, $M'_{j,n}(z)$  and $1/M_{j,n}(z)$ are $O(1)$ uniformly in $n$ and $z\in K$, and (iii) $\M_{j,n}(x)>0$ and $M_{j,n}(x)>0$ for all  real $x>\redge$. 

\item \label{enu:prop:asy_for_mult_cut:b}
For $z \in B_{\delta} $,
\begin{equation} \label{eq:asy_of_Cvarphi_n-j}
	(C\varphi_{n-j})(z;n) = \tilde{M}_{j,n}(z) e^{n(-\gfn(z)+\ell/2)},
\end{equation} 
where $\tilde{M}_{j,n}(z)$ is analytic in $B_{\delta}\setminus\R$, continuous up to the boundary, and
\begin{equation} \label{eq:asy_of_tilde_M_j_mult}
	\tilde{M}_{j,n}(z) = \tilde{\M}_{j,n}(z)(1+O(n^{-1}))
\end{equation}
uniformly in $z$ and $n$ where $\tilde{\M}_{j,n}(z)$ defined in~\eqref{eq:def_of_slash_Mhat}  is analytic in $B_{\delta}$. The function $\tilde{M}_{j,n}$ satisfies that (i) $\tilde{M}_{j,n}(z)=O(z^{j})$ uniformly in $n$ as $z\to\infty$, (ii) in any compact subset $K\in B_{\delta}$, $\tilde{M}_{j,n}(z)$, $\tilde{M}'_{j,n}(z)$ and $1/\tilde{M}_{j,n}(z)$  are $O(1)$ uniformly in $n$ and $z\in K \setminus \realR$ and (iii)  $-i\tilde{\M}_{j,n}(x)>0$ for $x>\redge$.

%we have an analytic function $\tilde{\M}_{j,0}(z;n)$ for $z \in B_{\delta}$, such that uniformly for all $z \in B_{\delta} \setminus \realR$
%\begin{equation} \label{eq:asy_of_tilde_M_j_mult}
%\tilde{M}_{j,n}(z) = \tilde{\M}_{j,0}(z;n) (1+O(n^{-1})).
%\end{equation}

\item \label{enu:prop:asy_for_mult_cut:c}
For $z \in A_{\delta}$, 
\begin{equation} \label{eq:psi_general}
\varphi_{n-j}(z) = 
\left( n^{1/6} \Ai(\Phi(z))B_{j,n}(z) + n^{-1/6} \Ai'(\Phi(z))D_{j,n}(z) \right) e^{-\frac{n}{2}V(z)},
\end{equation}
where $B_{j,n}(z)$ and $D_{j,n}(z)$ are analytic functions in $A_{\delta}$ and
\begin{equation}
B_{j,n}(z) = \B_{j,n}(z)(1+O(n^{-1})), \quad D_{j,n}(z) = \D_{j,n}(z)(1+O(n^{-1}))
\end{equation}
uniformly in $z$ and $n$. The functions $B_{j,n}$ and $D_{j,n}$ satisfy (i) $B_{j,n}(z)$, $D_{j,n}(z)$, $B'_{j,n}(z)$, $D'_{j,n}(z)$, $1/B_{j,n}(z)$ and  $1/D_{j,n}(z)$ are $O(1)$ uniformly in $n$ and $z \in A_{\delta}$ and (ii) $\B_{j,n}(x) > 0$ and $B_{j,n}(x) > 0$ for $x \in (\redge-\delta, \redge+\delta)$.

%in $A_{\delta}$.  and we have analytic functions $\B_{j,0}(z;n)$ and $\D_{j,0}(z;n)$ such tha uniformly for all $z \in A_{\delta}$
%\begin{equation} \label{eq:asy_of_B_j_D_j_mult}
%B_{j,n}(z) = \B_{j,0}(z;n) + O(n^{-1}), \qquad D_{j,n}(z) = \D_{j,0}(z;n) + O(n^{-1}).
%\end{equation}

\item \label{enu:prop:asy_for_mult_cut:d}
For $z \in A_{\delta} \cap \compC_+$, we have 
\begin{equation} \label{eq:Cpsi_upper_plan_general}
(C\varphi_{n-j})(z) = e^{\pi i/3}\left( n^{1/6} \Ai(\omega^2\Phi(z))B_{j,n}(z) \right. %\\
+ \left. n^{-1/6} \omega^2\Ai'(\omega^2\Phi(z))D_{j,n}(z) \right) e^{-\frac{n}{2}V(z)},
\end{equation}
and for $z \in A_{\delta} \cap \compC_-$, we have
\begin{equation} \label{eq:Cpsi_lower_plan_general}
(C\varphi_{n-j})(z) = -e^{\pi i/3}\left( n^{1/6} \omega^2\Ai(\omega\Phi(z))B_{j,n}(z) \right. %\\
+ \left. n^{-1/6} \Ai'(\omega\Phi(z))D_{j,n}(z) \right) e^{-\frac{n}{2}V(z)},
\end{equation}
where $\omega = e^{2\pi i/3}$ and $B_{j,n}(z)$ and $D_{j,n}(z)$ are the same functions in~\eqref{eq:psi_general}. The formulas~\eqref{eq:Cpsi_upper_plan_general} and~\eqref{eq:Cpsi_lower_plan_general} hold up to the boundary $z\in A_\delta\cap \R$. 
\end{enumerate}
\end{prop}

\begin{proof}
By formulas \eqref{eq:M^infty_j_above} and \eqref{eq:M^infty_j_below} of $M^{(\infty)}_{j,n}(z)$, the properties of the theta function $\theta$ and the definition of $d$, we have that for $p,q = 1,2$, the functions $(M^{(\infty)}_{j,n})_{pq}(z)$ and $(M^{(\infty)}_{j,n})_{pq}(z)^{-1}$ are uniformly bounded for $z$ in any compact subset $K \subset B_{\delta}$ and the functions
\begin{gather*}
\left( i(M^{(\infty)}_{j,n})_{p1}(z) + (M^{(\infty)}_{j,n})_{p2}(z) \right) (z-\redge)^{1/4}, \\
\left( i(M^{(\infty)}_{j,n})_{p1}(z) + (M^{(\infty)}_{j,n})_{p2}(z) \right)^{-1} (z-\redge)^{-1/4}, \\
\left( -i(M^{(\infty)}_{j,n})_{p1}(z) + (M^{(\infty)}_{j,n})_{p2}(z) \right) (z-\redge)^{-1/4}, \\
\left( i(M^{(\infty)}_{j,n})_{p1}(z) + (M^{(\infty)}_{j,n})_{p2}(z) \right)^{-1} (z-\redge)^{1/4}
\end{gather*}
are uniformly bounded for $z \in A_{\delta}$ and $n$. We note that although $M^{(\infty)}_{j,n}(z)$ is not well defined on $J$, the functions considered above are well defined on $A_{\delta} \cap J$. Plugging in these estimates and the estimate \eqref{eq:leading_term_of_R_j} of $R_{j,n}(z)$ into \eqref{eq:RHP_out_of_lens}, \eqref{eq:RHP_in_A_!_IV}, \eqref{eq:RHP_in_A_!I} and \eqref{eq:RHP_in_A_!II} we obtain the estimates of $M_{j,n}(z)$, $1/M_{j,n}(z)$, $\tilde{M}_{j,n}(z)$, $1/\tilde{M}_{j,n}(z)$, $B_{j,n}(z)$, $1/B_{j,n}(z)$, $D_{j,n}(z)$ and $1/D_{j,n}(z)$. With the help of the Cauchy's integral formula, we further obtain the estimates of $M'_{j,n}(z)$, $\tilde{M}'_{j,n}(z)$, $B'_{j,n}(z)$ and $D'_{j,n}(z)$. From the formula of $M^{(\infty)}_{j,n}(z)$ we also derive the positivity of $\M_{j,n}(x)$, $M_{j,n}(x)$, $-i\tilde{\M}_{j,n}(x)$, $\B_{j,n}(x)$ and $B_{j,n}(x)$.
\end{proof}

\begin{rmk}
We use the following identity in the analysis. 
It is straightforward to derive from the Riemann-Hilbert problem of $Y_k(z;n)$ that $\det Y_k(z;n)\equiv 1$. This implies that 
\begin{equation} \label{eq:algebraic_property_of_RHP}
\frac{\gamma_k}{\gamma_{k-1}} = \frac{-1}{2\pi i} (p_k(z)C\varphi_{k-1}(z) - p_{k-1}(z)C\varphi_k(z)).
\end{equation}
Taking $k=n-j$ and using asymptotic formulas \eqref{eq:psi_general}, \eqref{eq:Cpsi_upper_plan_general} and \eqref{eq:Cpsi_lower_plan_general} in \eqref{eq:algebraic_property_of_RHP}, with the help of \cite[10.4.11 and 10.4.12]{Abramowitz-Stegun64}
\begin{align}
\Ai(z)\omega\Ai'(\omega z) - \Ai'(z)\Ai(\omega z) = & \frac{e^{-\pi i/6}}{2\pi}, \\
\Ai(z)\omega^2\Ai'(\omega^2 z) - \Ai'(z)\Ai(\omega^2 z) = & \frac{e^{\pi i/6}}{2\pi},
\end{align}
we find
\begin{equation} \label{eq:cross_product_of_B_D}
B_{j,n}(z)D_{j+1,n}(z) - B_{j+1,n}(z)D_{j,n}(z) = \frac{\gamma_{n-j}}{\gamma_{n-j-1}}.
\end{equation}
\end{rmk}

Proposition \ref{prop:asy_for_mult_cut} implies the following asymptotic properties of $\psi_{n-j}$. These are used in the main analysis extensively.

\begin{cor} \label{cor:only_one_last_sect}
Fix $T\in \R$. 
There exists $\delta_0>0$ such that for each fixed $\epsilon\in (0, \delta_0]$, the following holds
% and let $\delta$ be fixed, small enough positive constants. 
%Let $\Int$ be the interval defined in \eqref{eq:interval}, $E_{T,\delta}$ be the interval defined in \eqref{eq:defn_of_E_T_epsilon}, $\xi = \beta n^{2/3}(x-\redge)$ as in \eqref{eq:scaling_of_x_around_a_N+1} and $v_{j,n}(\xi)$ be defined in \eqref{eq:def_pf_v_jn}. 
as $n\to\infty$ and $j=O(1)$.

\begin{enumerate}[label=(\alph*)]
\item \label{enu:cor:only_one_last_sect:a}
For $x\ge \redge+\epsilon$,
\begin{equation}\label{eq:psin101}
\begin{split}
\psi_{n-j}(x) = & M_{j,n}(x) e^{n(\Gfn(x)-\Hfn(x))/2}= O(e^{n(\Gfn(x)-\Hfn(x))/2} (1 + \lvert x \rvert)^{-j}),
\end{split}
\end{equation}	
for $M_{j,n}(x)$ in~\eqref{eq:asy_of_varphi_n-j}.
% and $\M_{j,n}(z)$ defined in \eqref{eq:def_of_slash_M}.

\item \label{enu:cor:only_one_last_sect:b}
Let $E_{T,\epsilon}:= [\redge + \beta^{-1}n^{-2/3}T, \redge+\epsilon]$ be the interval defined in \eqref{eq:defn_of_E_T_epsilon}.
For $x\in E_{T,\epsilon}$,
\begin{equation}\label{eq:psin102}
\begin{split}
	\psi_{n-j}(x) 
	&= O(n^{1/6} e^{-\factor |\xi|^{3/2}} ), \qquad \xi := \beta n^{2/3}(x-\redge).
\end{split}
\end{equation}	
%where $\xi = \beta n^{2/3}(x-\redge)$ as in \eqref{eq:scaling_of_x_around_a_N+1}.

\item \label{enu:cor:only_one_last_sect:c}
Let $\Int:=\left[ \redge + \beta^{-1}n^{-2/3}T, \infty \right)$ be the interval defined in \eqref{eq:interval}.
Then 
\begin{equation} \label{eq:L^2_norm_of_psi_1}
	\lVert \psi_{n-j} \rVert_{L^2(\Int)}  = O(n^{-1/6}).
\end{equation}
Also for every $\bar{x}$ in $(\redge + \epsilon, \infty)$, there is $\epsilon'>0$ such that 
\begin{equation} \label{eq:psiL20}
	 \lVert \psi_{n-j} \lVert_{L^2([\bar{x}, \infty))} = O(e^{-\epsilon'n}).
\end{equation}

\item \label{enu:cor:only_one_last_sect:d}
As $n\to\infty$,  $v_{j,n}(\xi) := \psi_{n-j}(\redge+\beta^{-1}n^{-2/3}\xi)$  satisfies 
\begin{equation}\label{eq:psiL2}
	v_{j,n}(\xi)  \to \B_{j,n}(\redge)\Ai(\xi), \quad \text{in $L^2([T, \infty))$.}
\end{equation}

\end{enumerate}
\end{cor}

\begin{proof}
The result \ref{enu:cor:only_one_last_sect:a} follows from~\eqref{eq:asy_of_varphi_n-j} and noting that $2\gfn(z)-V(z)-\ell=\Gfn(z;a)-\Hfn(z;a)$.

For \ref{enu:cor:only_one_last_sect:b}, note that $T\le \xi\le \epsilon n^{2/3}$. Thus, $|\Ai(\xi)|\le Ce^{-\frac23 |\xi|^{3/2}}$and  $|\Ai'(\xi)|\le C(|\xi|^{1/4}+1)e^{-\frac23 |\xi|^{3/2}}
\le C'n^{1/6} e^{-\frac23 |\xi|^{3/2}}$  for some constants $C, C'>0$. From~\eqref{eq:psi_general} and the behavior of $\Phi(z)$ in $A_{\epsilon}$, we obtain the estimate with the factor $\frac23$ changed to a smaller constant which can be made arbitrarily close to $\frac23$ if we take $\epsilon$ smaller. To be definite, we fix this constant as $\factor$. 

For \ref{enu:cor:only_one_last_sect:c}, by the asymptotics \eqref{eq:psin101} and \eqref{eq:psin102} of $\psi_{n-1}$, we find
\begin{equation} \label{eq:L2_norm_pf_psi_sub}
\begin{split}
 \| \psi_{n-j} \|_{L^2(\Int)}^2 \leq & 2 \left( \lVert \psi_{n-j} \rVert^2_{L^2(E_{T,\epsilon})} + \lVert \psi_{n-j} \rVert^2_{L^2([\redge+\epsilon, \infty))} \right) \\
= & 2 \left[ \int^{\beta n^{2/3}\epsilon}_T O \left( n^{1/3} e^{-2\factor \lvert \xi \rvert^{3/2}} \right) \frac{d\xi}{\beta n^{2/3}} \right. %\\
 \left.  \vphantom{\int^{\beta n^{2/3}\epsilon}_T} 
+   \int^{\infty}_{\redge+\epsilon} O \left( e^{n(\Gfn(x)-\Hfn(x))} (1
     + \lvert x \rvert)^{-2j} \right) dx \right] \\
= & O(n^{-1/3}).
\end{split}
\end{equation}
The estimate~\eqref{eq:psiL20} is similar.

Item \ref{enu:cor:only_one_last_sect:d} follows from Proposition~\ref{prop:asy_for_mult_cut} (c). 
\end{proof}

%%%%

\bigskip

The above asymptotics for $\psi_{n-j}$ yields the asymptotics for the Christoffel-Darboux kernel $K_{n-j,n}(x,y)$. 

\begin{cor} \label{lemma:various_estimates_of_K_n-1,n}
Let $T \in \realR$ be fixed. There exists $\delta_0>0$ such that for each fixed $\epsilon\in (0, \delta_0]$, the followings hold  as $n\to\infty$ and $j=O(1)$.
Let $E_{T,\epsilon}$ be the interval defined in \eqref{eq:defn_of_E_T_epsilon}.
\begin{enumerate}[label=(\alph*)]
\item \label{lemma_enu:estimate_K_n-1,n:1} 
For $x,y \in (\redge + \epsilon/2, \infty)$,
\begin{equation} \label{eq:lemma_enu:estimate_K_n-1,n:1}
	K_{n-j,n}(x,y) = O(e^{n(\Gfn(x)-\Hfn(x) + \Gfn(y)-\Hfn(y))/2} (1 + \lvert x \rvert)^{-j}(1 + \lvert y \rvert)^{-j}).
\end{equation}

\item \label{lemma_enu:estimate_K_n-1,n:2} 
For $x,y \in E_{T,\epsilon}$,
\begin{equation} \label{eq:lemma_enu:estimate_K_n-1,n:2}
	K_{n-j,n}(x,y) = O(n^{2/3}e^{-\factor(\lvert \xi \rvert^{3/2} + \lvert \eta \rvert^{3/2})}),
\end{equation}
where $\xi := \beta n^{2/3}(x-\redge)$ and $\eta := \beta n^{2/3}(y-\redge)$.

\item \label{lemma_enu:estimate_K_n-1,n:3} For $x \in (\redge + \epsilon,  \infty)$ and $y \in E_{T,\epsilon/2}$,
\begin{equation} \label{eq:lemma_enu:estimate_K_n-1,n:3}
K_{n-j,n}(x,y) = O(n^{1/6} e^{n(\Gfn(x)-\Hfn(x))/2} e^{-\factor\lvert \eta \rvert^{3/2}} (1 + \lvert x \rvert)^{-j}).
\end{equation}

\item \label{lemma_enu:estimate_K_n-1,n:4} 
For $x \in E_{T,\epsilon/2}$ and $y \in (\redge + \epsilon, \infty)$,
\begin{equation} \label{eq:lemma_enu:estimate_K_n-1,n:4}
K_{n-j,n}(x,y) = O(n^{1/6} e^{-\factor\lvert \xi \rvert^{3/2}} e^{n(\Gfn(y)-\Hfn(y))/2} (1 + \lvert y \rvert)^{-j}).
\end{equation}
\end{enumerate}
All estimates above are uniform in $x, y$ in their domains and in $n$.
\end{cor}

\begin{proof}
Item \ref{lemma_enu:estimate_K_n-1,n:3} and \ref{lemma_enu:estimate_K_n-1,n:4} follow directly the asymptotics \eqref{eq:psin101} and \eqref{eq:psin102} of $\psi_{n-j}$, and the Christoffel-Darboux formula~\eqref{eq:Christoffel_Darboux_K_n-1,n} of $K_{n-j,n}(x,y)$, noting that $x-y$ never vanishes.

For $x,y \in (\redge + \epsilon/2, \infty)$,~\eqref{eq:psin101} implies that 
\begin{equation} \label{eq:lemma_enu:estimate_K_n-1,n:1_prepare}
K_{n-j,n}(x,y) = 
e^{n(\Gfn(x)-\Hfn(x) + \Gfn(y)-\Hfn(y))/2} \frac{M_{j,n}(x)M_{j+1,n}(y) - M_{j+1,n}(x)M_{j,n}(y)}{x-y}.
\end{equation}
Since $M_{j,n}$ and its derivatives are uniformly bounded, we obtain \ref{lemma_enu:estimate_K_n-1,n:1} . 

Item \ref{lemma_enu:estimate_K_n-1,n:2} follows from a similar calculation but using the asymptotics~\eqref{eq:psi_general}. The calculation is direct and is the same as \cite[Formula (3.8)]{Deift-Gioev07a}.
\end{proof}

We also need the following results for the Christoffel-Darboux kernel.

\begin{cor} \label{lemma:trace_norm_convergence_K_n-1,n}
Fix $T\in\R$ and let $\Int$ be the interval defined in \eqref{eq:interval}. Then we have the following:
\begin{enumerate}[label=(\alph*)]

\item \label{enu:lemma:trace_norm_convergence_K_n-1,n:a}
For any fixed $C$, we have 
\begin{equation}\label{eq:KtoKairywitherror}
	\frac1{\beta n^{2/3}}K_{n-j, n}(x,y)=K_{\Airy} (\xi, \eta) + o(e^{-C ( |\xi|+|\eta|)})
\end{equation}
for all $x,y \in \Int$ as $n \to \infty$, where $\xi:= (x-\redge)\beta n^{2/3}$ and $\eta:= (y-\redge)\beta n^{2/3}$.

\item \label{enu:lemma:trace_norm_convergence_K_n-1,n:b}
Define the operator $\K_{n-j, n}$ by kernel 
\begin{equation}\label{eq:Kscaled}
	\K_{n-j, n}(\xi, \eta):= \frac{1}{\beta n^{2/3}} K_{n-j, n}\big(
	\redge+ \frac{\xi}{\beta n^{2/3}}, \redge+ \frac{\eta}{\beta n^{2/3}}\big).
\end{equation}
Then
\begin{equation}\label{eq:Kinverse3}
	\big( 1 - \chi_{[T, \infty)} \K_{n-j,n} \chi_{[T, \infty)} \big)^{-1}
	\to 
	\big( 1 - \chi_{[T,\infty)} K_{\Airy} \chi_{[T,\infty)} \big)^{-1}
\end{equation}
in trace norm as $n \to \infty$. 

\item \label{enu:lemma:uniformbddinverse_K_n-1,n:c}
The operator norms of $\big( 1 - \chi_{\Int} K_{n-j,n} \chi_{\Int} \big)^{-1}$ are bounded uniformly in $n$. As a corollary, 
%\begin{equation} \label{eq:Kinverse1}
%	\big( 1 - \chi_{\Int} K_{n-j,n} \chi_{\Int} \big)^{-1}
%\end{equation}
The operator norms of 
\begin{equation} \label{eq:Kinverse2}
	 \big( 1 - \chi_{[\bar{x}_n,\infty)} K_{n-j,n} \chi_{[\bar{x}_n,\infty)} \big)^{-1},
\end{equation}
are also bounded uniformly in $n$ and in $\bar{x}_n$ as long as $\bar{x}_n$ are in a compact subset of $(\redge, \infty)$.%, are bounded.
% in $T$ and $\bar{x}$ respectively.

\item \label{enu:lemma:trace_norm_convergence_K_n-1,n:c}
We have 
\begin{equation} \label{eq:Fred_det_of_K_n-1,n_1}
	\lim_{n \to \infty}  \det \left( 1 - \chi_{\Int} K_{n-j,n} \chi_{\Int} \right)
	= \FGUE(T),
\end{equation}
and
\begin{equation} \label{eq:Fred_det_of_K_n-1,n_2}
	\lim_{n \to \infty}  \det \left( 1 - \chi_{[\bar{x},\infty)} K_{n-j,n} \chi_{[\bar{x},\infty)} \right)
	= 1.
\end{equation}
for  any $\bar{x}$ is in a compact subset of $(\redge, \infty)$.

\end{enumerate}
\end{cor}

\begin{proof}
The proof of a result similar to \ref{enu:lemma:trace_norm_convergence_K_n-1,n:a} for the non-varying weight is given in \cite[Formula (3.8)]{Deift-Gioev07a} . The varying weight case is proved in the same way. Note that our $C$ is the $c$ in \cite[Formula (3.8)]{Deift-Gioev07a}, which can be assumed to be an arbitrarily large positive number.

The proof of a result similar to \ref{enu:lemma:trace_norm_convergence_K_n-1,n:b} for the non-varying weight is given in the proof of the $\beta = 2$ case in \cite[Corollary 1.4]{Deift-Gioev07a} . The varying weight case is proved in the same way.

Items \ref{enu:lemma:uniformbddinverse_K_n-1,n:c} and \ref{enu:lemma:trace_norm_convergence_K_n-1,n:c} follow from \ref{enu:lemma:trace_norm_convergence_K_n-1,n:b}. 

\end{proof}

\def\cydot{\leavevmode\raise.4ex\hbox{.}}

%%%%%%%%%%

%\bibliographystyle{abbrv}
%\bibliography{bibliography}

\end{document}